\documentclass{fundam}
\pdfoutput=1

\usepackage{mysty}
\usepackage[T1]{fontenc}
 \usepackage{hyperref}

\publyear{24}
\papernumber{2194}
\volume{192}
\issue{3-4}

  \finalVersionForARXIV

\begin{document}

\title{Myhill-Nerode Theorem for Higher-Dimensional Automata}

\author{%
  Uli Fahrenberg\thanks{Address for correspondence:  EPITA Research Laboratory (LRE), France.}
   \\
  EPITA Research Laboratory (LRE), France\\
  uli@lrde.epita.fr
  \and Krzysztof Ziemia\'nski   \\
  University of Warsaw, Poland\\
  ziemians@mimuw.edu.pl
  }

\maketitle

\runninghead{U.~Fahrenberg and K.~Ziemia\'nski}{Myhill-Nerode Theorem for Higher-Dimensional Automata}

\begin{abstract}
  We establish a Myhill-Nerode type theorem for higher-dimensional
  automata (HDAs), stating that a language is regular if and only if
  it has finite prefix quotient.  HDAs extend standard automata with
  additional structure, making it possible to distinguish between
  interleavings and concurrency.  We also introduce deterministic HDAs
  and show that not all HDAs are determinizable, that is, there exist
  regular languages that cannot be recognised by a deterministic HDA.
  Using our theorem, we develop an internal characterisation of
  deterministic languages.  Lastly, we develop analogues of the
  Myhill-Nerode construction and of determinacy for HDAs with
  interfaces.
\end{abstract}

\section{Introduction}

\begin{figure}[tp]
  \centering
  \begin{tikzpicture}
    \begin{scope}[x=1.5cm, state/.style={shape=circle, draw,
        fill=white, initial text=, inner sep=1mm, minimum size=3mm}]
      \node[state, black] (10) at (0,0) {};
      \node[state, rectangle] (20) at (0,-1) {$\vphantom{b}a$};
      \node[state] (30) at (0,-2) {};
      \node[state, black] (11) at (1,0) {};
      \node[state, rectangle] (21) at (1,-1) {$b$};
      \node[state] (31) at (1,-2) {};
      \path (10) edge (20);
      \path (20) edge (30);
      \path (11) edge (21);
      \path (21) edge (31);
      \node[state, black] (m) at (.5,-1) {};
      \path (20) edge[out=15, in=165] (m);
      \path (m) edge[out=-165, in=-15] (20);
      \path (21) edge[out=165, in=15] (m);
      \path (m) edge[out=-15, in=-165] (21);
    \end{scope}
    \begin{scope}[xshift=4cm]
      \node[state] (00) at (0,0) {};
      \node[state] (10) at (-1,-1) {};
      \node[state] (01) at (1,-1) {};
      \node[state] (11) at (0,-2) {};
      \path (00) edge node[left] {$\vphantom{b}a$\,} (10);
      \path (00) edge node[right] {\,$b$} (01);
      \path (10) edge node[left] {$b$\,} (11);
      \path (01) edge node[right] {\,$\vphantom{b}a$} (11);
    \end{scope}
  \end{tikzpicture}
  \qquad\qquad
  \begin{tikzpicture}
    \begin{scope}[xshift=8cm]
      \path[fill=black!15] (0,0) to (-1,-1) to (0,-2) to (1,-1);
      \node[state] (00) at (0,0) {};
      \node[state] (10) at (-1,-1) {};
      \node[state] (01) at (1,-1) {};
      \node[state] (11) at (0,-2) {};
      \path (00) edge node[left] {$\vphantom{b}a$\,} (10);
      \path (00) edge node[right] {\,$b$} (01);
      \path (10) edge node[left] {$b$\,} (11);
      \path (01) edge node[right] {\,$\vphantom{b}a$} (11);
    \end{scope}
    \begin{scope}[x=1.5cm, state/.style={shape=circle, draw,
        fill=white, initial text=, inner sep=1mm, minimum size=3mm},
      xshift=10.5cm]
      \node[state, black] (10) at (0,0) {};
      \node[state, rectangle] (20) at (0,-1) {$\vphantom{b}a$};
      \node[state] (30) at (0,-2) {};
      \node[state, black] (11) at (1,0) {};
      \node[state, rectangle] (21) at (1,-1) {$b$};
      \node[state] (31) at (1,-2) {};
      \path (10) edge (20);
      \path (20) edge (30);
      \path (11) edge (21);
      \path (21) edge (31);
    \end{scope}
  \end{tikzpicture}
  \caption{Petri net and HDA models distinguishing interleaving
    from non-interleaving concurrency.  Left: Petri net and
    HDA models for $a. b+ b. a$; right: HDA and Petri net models for
    $a\para b$.}
  \label{fi:int-conc}\vspace*{-3mm}
\end{figure}
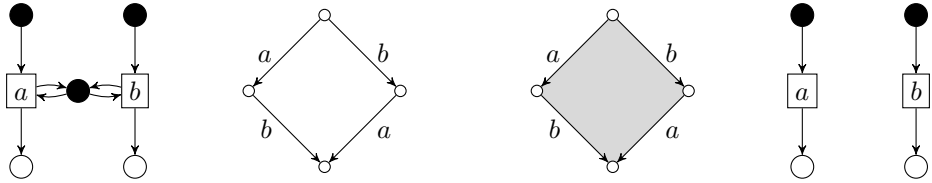

Higher-dimensional automata (HDAs),
introduced by Pratt and van Glabbeek \cite{Pratt91-geometry, Glabbeek91-hda, DBLP:journals/tcs/Glabbeek06},
extend standard automata with additional structure
that makes it possible to distinguish between interleavings and concurrency.
That puts them in a class with other non-interleaving models for concurrency
such as Petri nets \cite{book/Petri62},
event structures \cite{DBLP:journals/tcs/NielsenPW81},
configuration structures \cite{DBLP:conf/lics/GlabbeekP95, DBLP:journals/tcs/GlabbeekP09},
asynchronous transition systems \mbox{\cite{Bednarczyk87-async, DBLP:journals/cj/Shields85}},
and similar approaches \cite{pratt95chu, DBLP:journals/acta/GlabbeekG01, Pratt03trans_cancel, P15jlamp_STstruct},
while retaining some of the properties and intuition of automata-like models.
As an example,
Figure~\ref{fi:int-conc} shows Petri net and HDA models
for a system with two events,
labelled $a$ and $b$.
The Petri net and HDA on the left side
model the (mutually exclusive) interleaving of $a$ and $b$ as either $a. b$ or $b. a$;
those to the right model concurrent execution of $a$ and $b$.
In the HDA, this independence is indicated by a filled-in square.

We have recently introduced languages of HDAs \cite{Hdalang},
which consist of partially ordered multisets with interfaces (ipomsets),
and shown a Kleene theorem for them \cite{conf/concur/FahrenbergJSZ22, Kleenearxiv}.
Here we continue to develop the language theory of HDAs.
Our first contribution is a Myhill-Nerode type theorem for HDAs,
stating that a language is regular if and only if it has finite prefix quotient.
This provides a necessary and sufficient condition for regularity.
Our proof is inspired by the standard proofs of the Myhill-Nerode theorem,
but the higher-dimensional structure introduces some difficulties.
For example, we cannot use the standard prefix quotient relation
but need to develop a stronger one which takes concurrency of events into account.

As a second contribution,
we give a precise definition of deterministic HDAs
and show that there exist regular languages that cannot be recognised by deterministic HDAs.
Our Myhill-Nerode construction will produce a deterministic HDA for such deterministic languages,
and a non-deterministic HDA otherwise.
Our definition of determinism is more subtle than for standard automata
as it is not always possible to remove non-accessible
parts of HDAs.
We develop a language-internal characterisation of deterministic languages.

Thirdly, we develop a variant of the Myhill-Nerode construction
and of determinism which uses higher-dimensional automata with interfaces (iHDAs).
These were introduced in \cite{conf/concur/FahrenbergJSZ22}
and allow for some components to be missing which in HDAs would have to exist solely for structural reasons.
In iHDAs, non-accessible parts may be removed,
which allows for a more principled Myhill-Nerode construction.
HDAs and iHDAs are related via mappings called resolution and closure
which preserve languages.

\medskip
We start this paper by introducing languages of ipomsets in Section \ref{se:ipomsets}.
Section \ref{s:Step} develops important decomposition properties of ipomsets needed in the sequel,
and HDAs are introduced in Section~\ref{se:hda}.
Section \ref{se:mn} then states and proves our Myhill-Nerode theorem,
and Section \ref{se:det} introduces deterministic HDAs.
HDAs with interfaces are defined in Section \ref{se:ihda},
the Myhill-Nerode theorem using iHDAs is in Section \ref{se:imn},
and deterministic iHDAs are treated in Section \ref{se:idet}.
This paper is based on \cite{DBLP:conf/apn/FahrenbergZ23}
which was presented at the 44th International Conference on Application and Theory of Petri Nets and Concurrency.
Compared to this conference paper, proofs have been added and errors corrected,
and the material in Sections \ref{se:imn} and \ref{se:idet} is new.

\section{Pomsets with interfaces}
\label{se:ipomsets}

HDAs model systems in which labelled events have duration and may happen concurrently.
Every event has a time interval during which it is active:
it starts at some point, then remains active until its termination
and never reappears.
Events may be concurrent, that is, their activity intervals may overlap;
otherwise, one of the events precedes the other.
We also need to consider executions
in which some events are already active at the beginning (\emph{source events})
or are still active at the end (\emph{target events}).

\medskip
At any moment of an execution we observe a list of currently active events
(such lists are called \emph{conclists} below).
The relative position of any two concurrent events on these lists
remains the same,
regardless of the point in time.
This provides a secondary relation between events, which we call \emph{event order}.

To make the above precise, let $\Sigma$ be a finite alphabet.
A \emph{conclist} (for ``concurrency list'') $(U, {\evord}, \lambda)$
is a finite set $U$ with a total order $\evord$ called the \emph{event order}
and a labelling function $\lambda: U\to \Sigma$.
Conclists (or rather their isomorphism classes) are effectively strings
but consist of concurrent, not subsequent, events.

A \emph{labelled poset with event order} (\emph{lposet})
$(P, {<}, {\evord}, \lambda)$
consists of a finite set $P$ with two relations:
\emph{precedence} $<$ and \emph{event order} $\evord$,
together with a labelling function $\lambda:P\to \Sigma$.
Note that different events may carry the same label: we do \emph{not} exclude autoconcurrency.
We require that both $<$ and $\evord$
are strict partial orders, that is,
they are irreflexive and transitive (and thus asymmetric).
We also require that for each $x\ne y$ in $P$,
at least one of $x<y$ or $y<x$ or $x\evord y$ or $y\evord x$ must hold;
that is, if $x$ and $y$ are concurrent, then they must be related by $\evord$.

Conclists may be regarded as lposets with empty precedence relation;
the last condition enforces that their elements are totally ordered by $\evord$.
A temporary state of an execution is described by a conclist,
while the whole execution provides an lposet of its events.
The precedence order expresses that one event terminates before the other starts.
The event order of an lposet is generated by the event orders of temporary conclists.
Hence any two events which are active concurrently are unrelated by $<$ but related by $\evord$.

In order to accommodate source and target events,
we need to introduce lposets with interfaces (iposets).
An \emph{iposet}
$(P, {<}, {\evord}, S, T, \lambda)$ consists of an lposet
$(P, {<}, {\evord}, \lambda)$ together with subsets $S, T\subseteq P$
of source and target \emph{interfaces}.  Elements of $S$ must be
$<$-minimal and those of $T$ $<$-maximal; hence both $S$ and $T$ are conclists.
We often denote an iposet as above by $\ilo{S}{P}{T}$ or $(S, P, T)$, ignoring the orders and labelling,
or use $S_P=S$ and $T_P=T$ if convenient.
Source and target events will be marked by ``$\ibullet$'' at the left or right side, and
if the event order is not shown, we assume that it goes downwards.

\begin{figure}[!h]
  \centering
  \begin{tikzpicture}[x=.95cm]
    \def\possh{-1.3};
    \begin{scope}[shift={(9.6,0)}]
      \def\hw{0.3};
      \filldraw[fill=green!50!white,-](0,1.2)--(1.2,1.2)--(1.2,1.2+\hw)--(0,1.2+\hw);
      \filldraw[fill=pink!50!white,-](0.3,0.7)--(1.9,0.7)--(1.9,0.7+\hw)--(0.3,0.7+\hw)--(0.3,0.7);
      \filldraw[fill=blue!20!white,-](0.5,0.2)--(1.7,0.2)--(1.7,0.2+\hw)--(0.5,0.2+\hw)--(0.5,0.2);
      \draw[thick,-](0,0)--(0,1.7);
      \draw[thick,-](2.2,0)--(2.2,1.7);
      \node at (0.6,1.2+\hw*0.5) {$a$};
      \node at (1.1,0.7+\hw*0.5) {$b$};
      \node at (1.1,0.2+\hw*0.5) {$c$};
    \end{scope}
    \begin{scope}[shift={(9.6,\possh)}]
      \node (a) at (0.4,0.7) {$a$};
      \node at (.22,.7) {$\ibullet$};
      \node (c) at (0.4,-0.7) {$c$};
      \node (b) at (1.8,0) {$b$};
      \path[densely dashed, gray] (a) edge (b) (b) edge (c) (a) edge (c);
    \end{scope}
    \begin{scope}[shift={(6.4,0)}]
      \def\hw{0.3};
      \filldraw[fill=green!50!white,-](0,1.2)--(1.2,1.2)--(1.2,1.2+\hw)--(0,1.2+\hw);
      \filldraw[fill=pink!50!white,-](1.3,0.7)--(1.9,0.7)--(1.9,0.7+\hw)--(1.3,0.7+\hw)--(1.3,0.7);
      \filldraw[fill=blue!20!white,-](0.5,0.2)--(1.7,0.2)--(1.7,0.2+\hw)--(0.5,0.2+\hw)--(0.5,0.2);
      \draw[thick,-](0,0)--(0,1.7);
      \draw[thick,-](2.2,0)--(2.2,1.7);
      \node at (0.6,1.2+\hw*0.5) {$a$};
      \node at (1.6,0.7+\hw*0.5) {$b$};
      \node at (1.1,0.2+\hw*0.5) {$c$};
    \end{scope}
    \begin{scope}[shift={(6.4,\possh)}]
      \node (a) at (0.4,0.7) {$a$};
      \node at (.22,.7) {$\ibullet$};
      \node (c) at (0.4,-0.7) {$c$};
      \node (b) at (1.8,0) {$b$};
      \path (a) edge (b);
      \path[densely dashed, gray]  (b) edge (c) (a) edge (c);
    \end{scope}
    \begin{scope}[shift={(3.2,0)}]
      \def\hw{0.3};
      \filldraw[fill=green!50!white,-](0,1.2)--(1.2,1.2)--(1.2,1.2+\hw)--(0,1.2+\hw);
      \filldraw[fill=pink!50!white,-](1.3,0.7)--(1.9,0.7)--(1.9,0.7+\hw)--(1.3,0.7+\hw)--(1.3,0.7);
      \filldraw[fill=blue!20!white,-](0.5,0.2)--(1.1,0.2)--(1.1,0.2+\hw)--(0.5,0.2+\hw)--(0.5,0.2);
      \draw[thick,-](0,0)--(0,1.7);
      \draw[thick,-](2.2,0)--(2.2,1.7);
      \node at (0.6,1.2+\hw*0.5) {$a$};
      \node at (1.6,0.7+\hw*0.5) {$b$};
      \node at (0.8,0.2+\hw*0.5) {$c$};
    \end{scope}
    \begin{scope}[shift={(3.2,\possh)}]
      \node (a) at (0.4,0.7) {$a$};
      \node at (.22,.7) {$\ibullet$};
      \node (c) at (0.4,-0.7) {$c$};
      \node (b) at (1.8,0) {$b$};
      \path (a) edge (b) (c) edge (b);
      \path[densely dashed, gray]  (a) edge (c);
    \end{scope}
    \begin{scope}[shift={(0.0,0)}]
      \def\hw{0.3};
      \filldraw[fill=green!50!white,-](0,1.2)--(0.4,1.2)--(0.4,1.2+\hw)--(0,1.2+\hw);
      \filldraw[fill=pink!50!white,-](1.3,0.7)--(1.9,0.7)--(1.9,0.7+\hw)--(1.3,0.7+\hw)--(1.3,0.7);
      \filldraw[fill=blue!20!white,-](0.5,0.2)--(1.1,0.2)--(1.1,0.2+\hw)--(0.5,0.2+\hw)--(0.5,0.2);
      \draw[thick,-](0,0)--(0,1.7);
      \draw[thick,-](2.2,0)--(2.2,1.7);
      \node at (0.2,1.2+\hw*0.5) {$a$};
      \node at (1.6,0.7+\hw*0.5) {$b$};
      \node at (0.8,0.2+\hw*0.5) {$c$};
    \end{scope}
    \begin{scope}[shift={(0.0,\possh)}]
      \node (a) at (0.4,0.7) {$a$};
      \node at (.22,.7) {$\ibullet$};
      \node (c) at (0.4,-0.7) {$c$};
      \node (b) at (1.8,0) {$b$};
      \path (a) edge (b) (c) edge (b) (a) edge (c);
    \end{scope}
  \end{tikzpicture}
  \caption{Activity intervals (top)
    and corresponding iposets (bottom),
    see Example \ref{ex:iposets1}.
    Full arrows indicate precedence order;
    dashed arrows indicate event order;
    bullets indicate interfaces.}
  \label{fi:iposets1}\vspace*{-2mm}
\end{figure}
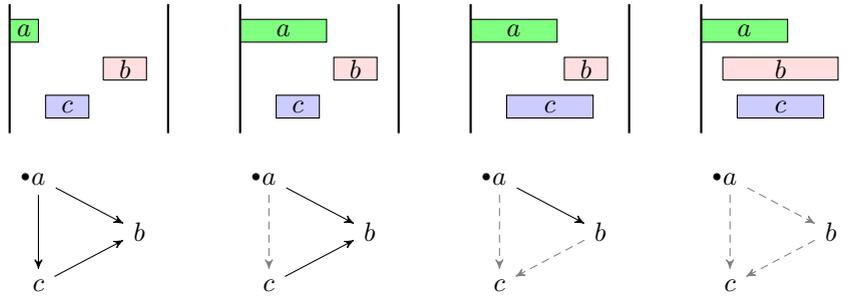

\begin{example}
  \label{ex:iposets1}
  Figure \ref{fi:iposets1}
  shows some simple examples
  of activity intervals of events and the corresponding iposets.
  The left iposet consists of three totally ordered events,
  given that the intervals do not overlap;
  the event $a$ is already active at the beginning
  and hence in the source interface.
  In the other iposets,
  the activity intervals do overlap
  and hence the precedence order is partial
  (and the event order non-trivial).
\end{example}

\eject
Given that the precedence relation $<$ of an iposet
represents activity intervals of events,
it is an \emph{interval order} \cite{book/Fishburn85}.
In other words, any of the iposets we will encounter admits an \emph{interval representation}:
functions $b$ and $e$ from $P$ to real numbers such that
$b(x)\le e(x)$ for all $x\in P$ and
$x <_P y\iff e(x)<b(y)$ for all $x, y\in P$.
We will only consider interval iposets in this paper
and hence omit the qualification ``interval''.
This is \emph{not} a restriction, but rather induced by the semantics.
The following property is trivial, but we will make heavy use of it later.

\begin{lemma}
  If $P$ is an (interval) iposet and $A\subseteq P$,
  then the set difference $P-A$ is an (interval) iposet as well.
\end{lemma}

Iposets may be refined by shortening the activity intervals of events,
so that some events stop being concurrent.
This corresponds to expanding the precedence relation $<$
(and, potentially, removing event order).
The inverse to refinement is called subsumption and defined as follows.
For iposets $P$ and $Q$, we say that $Q$ \emph{subsumes} $P$
(or that $P$ is a \emph{refinement} of $Q$)
and write $P\subsu Q$
if there exists a bijection $f: P\to Q$ (a \emph{subsumption})
which
\begin{itemize}
\item
  respects interfaces and labels: $f(S_P)=S_Q$, $f(T_P)=T_Q$, and $\lambda_Q\circ f=\lambda_P$;
\item
  reflects precedence: $f(x)<_Q f(y)$ implies $x<_P y$; and
\item
  preserves essential event order:
  $x\evord_P y$ implies $f(x)\evord_Q f(y)$
  whenever $x$ and $y$ are concurrent (that is, $x\not<_P y$ and $y\not<_P x$).
\end{itemize}
(Event order is essential for concurrent events,
but by transitivity, it also appears between non-concurrent events;
subsumptions may ignore such non-essential event order.)

\begin{example}
  In Figure~\ref{fi:iposets1},
  there is a sequence of refinements from right to left,
  each time shortening some activity intervals.
  Conversely, there is a sequence of subsumptions from left to right:

  \begin{equation*}
    \vcenter
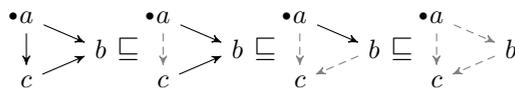
{\hbox{
        \begin{tikzpicture}[x=.7cm, y=.6cm]
          \path[use as bounding box] (.4,.8) -- (.4,-.8) -- (1.8,0);
          \node (a) at (0.4,0.7) {$a$};
          \node at (.15,.7) {$\ibullet$};
          \node (c) at (0.4,-0.7) {$c$};
          \node (b) at (1.8,0) {$b$};
          \path (a) edge (b) (c) edge (b) (a) edge (c);
        \end{tikzpicture}}}
    \;\subsu\;
    \vcenter{\hbox{
        \begin{tikzpicture}[x=.7cm, y=.6cm]
          \path[use as bounding box] (.4,.8) -- (.4,-.8) -- (1.8,0);
          \node (a) at (0.4,0.7) {$a$};
          \node at (.15,.7) {$\ibullet$};
          \node (c) at (0.4,-0.7) {$c$};
          \node (b) at (1.8,0) {$b$};
          \path (a) edge (b) (c) edge (b);
          \path[densely dashed, gray]  (a) edge (c);
        \end{tikzpicture}}}
    \;\subsu\;
    \vcenter{\hbox{
        \begin{tikzpicture}[x=.7cm, y=.6cm]
          \path[use as bounding box] (.4,.8) -- (.4,-.8) -- (1.8,0);
          \node (a) at (0.4,0.7) {$a$};
          \node at (.15,.7) {$\ibullet$};
          \node (c) at (0.4,-0.7) {$c$};
          \node (b) at (1.8,0) {$b$};
          \path (a) edge (b);
          \path[densely dashed, gray]  (b) edge (c) (a) edge (c);
        \end{tikzpicture}}}
    \;\subsu\;
    \vcenter{\hbox{
        \begin{tikzpicture}[x=.7cm, y=.6cm]
          \path[use as bounding box] (.4,.8) -- (.4,-.8) -- (1.8,0);
          \node (a) at (0.4,0.7) {$a$};
          \node at (.15,.7) {$\ibullet$};
          \node (c) at (0.4,-0.7) {$c$};
          \node (b) at (1.8,0) {$b$};
          \path[densely dashed, gray] (a) edge (b) (b) edge (c) (a) edge (c);
        \end{tikzpicture}}}
  \end{equation*}
  Interfaces need to be preserved across subsumptions,
  so in our example,
  the left endpoint of the $a$-interval must stay at the boundary.
\end{example}

Iposets and subsumptions form a category.
The isomorphisms in that category are invertible subsumptions,
and isomorphism classes of iposets are called \emph{ipomsets}.
Concretely, an isomorphism $f: P\to Q$ of iposets is a bijection
which
\begin{itemize}
\item
  respects interfaces and labels: $f(S_P)=S_Q$, $f(T_P)=T_Q$, and $\lambda_Q\circ f=\lambda_P$;
\item
  respects precedence: $x<_P y\iff f(x)<_Q f(y)$; and
\item
  respects essential event order:
  $x\evord_P y\iff f(x)\evord_Q f(y)$
  whenever $x\not<_P y$ and $y\not<_P x$.
\end{itemize}
Isomorphisms between iposets are unique
(because of the requirement that all elements
be ordered by $<$ or $\evord$),
hence we may switch freely between ipomsets and concrete representations,
see \cite{conf/concur/FahrenbergJSZ22} for details.
We write $P\cong Q$ if iposets $P$ and $Q$ are isomorphic
and let $\iiPoms$ denote the set of (interval) ipomsets.

Ipomsets may be \emph{glued},
using a generalisation of the standard serial composition of pomsets \cite{DBLP:journals/fuin/Grabowski81}.
For ipomsets $P$ and $Q$, their \emph{gluing} $P*Q$ is defined
if the targets of $P$ match the sources of $Q$: $T_P\cong S_Q$.
In that case, its carrier set is the quotient $(P\sqcup Q)_{/x\equiv f(x)}$,
where $f: T_P\to S_Q$ is the unique isomorphism,
the interfaces are $S_{P*Q}=S_P$ and $T_{P*Q}=T_Q$,
$\evord_{P*Q}$ is the transitive closure of ${\evord_P}\cup {\evord_Q}$,
and $x<_{P*Q} y$ if and only if $x<_P y$, or $x<_Q y$, or $x\in P-T_P$ and $y\in Q-S_Q$.
We will often omit the ``$*$'' in gluing compositions.
For ipomsets with empty interfaces, $*$ is serial pomset composition;
in the general case, matching interface points are glued,
see \cite{Hdalang, DBLP:journals/iandc/FahrenbergJSZ22} or below for examples.

A \emph{language} is, a priori, a set of ipomsets $L\subseteq \iiPoms$.
However, we will assume that languages are closed under refinement (inverse subsumption),
so that refinements of any ipomset in $L$ are also in $L$:

\begin{definition}
  \label{de:langs}
  A \emph{language} is a subset $L\subseteq \iiPoms$
  such that $P\subsu Q$ and $Q\in L$ imply $P\in L$.
\end{definition}

Using interval representations,
this means that languages are closed under shortening activity intervals of events.
The set of all languages is denoted $\Langs\subseteq 2^{\iiPoms}$.

For $X\subseteq \iiPoms$ an arbitrary set of ipomsets,
we denote by
\begin{equation*}
  X\down=\{P\in \iiPoms\mid \exists Q\in X: P\subsu Q\}
\end{equation*}
its downward subsumption closure,
that is, the smallest language which contains~$X$.
Then
\begin{equation*}
  \Langs = \{ X\subseteq \iiPoms\mid X\down = X\}.
\end{equation*}

\section{Step decompositions}
\label{s:Step}

An ipomset $P$ is \emph{discrete} if $<_P$ is empty and $\evord_P$ total.
Conclists are discrete ipomsets with empty interfaces.
Discrete ipomsets $\ilo{U}{U}{U}$ are identities for gluing composition and written $\id_U$.
A \emph{starter} is an ipomset $\ilo{U-A}{U}{U}$,
a \emph{terminator} is $\ilo{U}{U}{U-A}$;
these will be written $\starter{U}{A}$ and $\terminator{U}{A}$, respectively.

\begin{figure}[tbp]
  \centering
  \begin{tikzpicture}[y=1.2cm]
    \def\possh{-1.3};
    \def\hw{0.3};
    \node at (1.8,3.5) {%
      $\left[\vcenter{\hbox{\!%
            \begin{tikzpicture}[x=1.2cm]
              \node (a) at (0,0) {$a$};
              \node (c) at (1,0) {$c$};
              \node at (1.15,-0.75) {$\vphantom{b}\ibullet$};
              \node (b) at (0,-.75) {$b$};
              \node at (-.15,-.75) {$\vphantom{b}\ibullet$};
              \node (d) at (1,-.75) {$d$};
              \path (a) edge (c);
              \path (b) edge (d);
              \path (a) edge (d);
              \path[densely dashed, gray] (a) edge (b) (c) edge (b) (c) edge (d);
            \end{tikzpicture}
          \!\!}}\right]$};
    \node at (1.8,2) {%
      $\left[\vcenter{\hbox{\!%
            \begin{tikzpicture}[y=.7cm]
              \node (a) at (0,0) {$a$};
              \node (c) at (1,0) {$c$};
              \node at (1.17,-0.75) {$\vphantom{b}\ibullet$};
              \node (b) at (0,-.75) {$b$};
              \node at (-.17,-.75) {$\vphantom{b}\ibullet$};
              \node (d) at (1,-.75) {$d$};
              \path (a) edge (c);
              \path (b) edge (d);
              \path (a) edge (d);
            \end{tikzpicture}
          \!\!}}\right]$};
    \begin{scope}[shift={(0,0)}]
      \filldraw[fill=green!50!white,-](0.3,0.7)--(0.9,0.7)--(0.9,0.7+\hw)--(0.3,0.7+\hw)--(0.3,0.7);
      \filldraw[fill=pink!50!white,-](0.0,0.2)--(2.1,0.2)--(2.1,0.2+\hw)--(0.0,0.2+\hw)--(0.0,0.2);
      \filldraw[fill=blue!20!white,-](1.5,0.7)--(3.3,0.7)--(3.3,0.7+\hw)--(1.5,0.7+\hw)--(1.5,0.7);
      \filldraw[fill=yellow!50!white,-](2.7,0.2)--(3.6,0.2)--(3.6,0.2+\hw)--(2.7,0.2+\hw)--(2.7,0.2);
      \draw[thick,-](0,0)--(0,1.2);
      \draw[thick,-](3.6,0)--(3.6,1.2);
      \node at (0.6,0.7+\hw*0.5) {$a$};
      \node at (1.05,0.2+\hw*0.5) {$b$};
      \node at (2.4,0.7+\hw*0.5) {$c$};
      \node at (3.15,0.2+\hw*0.5) {$d$};
    \end{scope}
    \node at (4.9,3.5) {%
      $\left[\vcenter{\hbox{%
            \begin{tikzpicture}
              \path[use as bounding box] (-.3,.1) -- (.2,-.9);
              \node (a) at (0,0) {$a$};
              \node (b) at (0,-.75) {$b$};
              \node at (.17,0) {$\ibullet$};
              \node at (-.17,-.75) {$\vphantom{b}\ibullet$};
              \node at (.17,-.75) {$\vphantom{b}\ibullet$};
              \path[densely dashed, gray] (a) edge (b);
            \end{tikzpicture}
          }}\right]$};
    \node at (4.9,2) {$\starter{\bigloset{a\\b}}{a}$};
    \begin{scope}[shift={(4.6,0)}]
      \filldraw[fill=green!50!white,-](0.3,0.7)--(0.6,0.7)--(0.6,0.7+\hw)--(0.3,0.7+\hw)--(0.3,0.7);
      \filldraw[fill=pink!50!white,-](0.0,0.2)--(0.6,0.2)--(0.6,0.2+\hw)--(0.0,0.2+\hw)--(0.0,0.2);
      \draw[thick,-](0,0)--(0,1.2);
      \draw[thick,-](0.6,0)--(0.6,1.2);
      \node at (0.45,0.7+\hw*0.5) {$a$};
      \node at (0.3,0.2+\hw*0.5) {$b$};
    \end{scope}
    \node at (6.1,3.5) {%
      $\left[\vcenter{\hbox{%
            \begin{tikzpicture}
              \path[use as bounding box] (-.3,.1) -- (.2,-.9);
              \node (a) at (0,0) {$a$};
              \node (b) at (0,-.75) {$b$};
              \node at (-.17,0) {$\ibullet$};
              \node at (-.17,-.75) {$\vphantom{b}\ibullet$};
              \node at (.17,-.75) {$\vphantom{b}\ibullet$};
              \path[densely dashed, gray] (a) edge (b);
            \end{tikzpicture}
          }}\right]$};
    \node at (6.15,2) {$\terminator{\bigloset{a\\b}}{a}$};
    \begin{scope}[shift={(5.8,0)}]
      \filldraw[fill=green!50!white,-](0.0,0.7)--(0.3,0.7)--(0.3,0.7+\hw)--(0.0,0.7+\hw)--(0.0,0.7);
      \filldraw[fill=pink!50!white,-](0.0,0.2)--(0.6,0.2)--(0.6,0.2+\hw)--(0.0,0.2+\hw)--(0.0,0.2);
      \draw[thick,-](0,0)--(0,1.2);
      \draw[thick,-](0.6,0)--(0.6,1.2);
      \node at (0.15,0.7+\hw*0.5) {$a$};
      \node at (0.3,0.2+\hw*0.5) {$b$};
    \end{scope}
    \node at (7.3,3.5) {%
      $\left[\vcenter{\hbox{%
            \begin{tikzpicture}
              \path[use as bounding box] (-.3,.1) -- (.2,-.9);
              \node (a) at (0,0) {$c$};
              \node (b) at (0,-.75) {$b$};
              \node at (.17,0) {$\ibullet$};
              \node at (-.17,-.75) {$\vphantom{b}\ibullet$};
              \node at (.17,-.75) {$\vphantom{b}\ibullet$};
              \path[densely dashed, gray] (a) edge (b);
            \end{tikzpicture}
          }}\right]$};
    \node at (7.25,2) {$\starter{\bigloset{c\\b}}{c}$};
    \begin{scope}[shift={(7.0,0)}]
      \filldraw[fill=blue!20!white,-](0.3,0.7)--(0.6,0.7)--(0.6,0.7+\hw)--(0.3,0.7+\hw)--(0.3,0.7);
      \filldraw[fill=pink!50!white,-](0.0,0.2)--(0.6,0.2)--(0.6,0.2+\hw)--(0.0,0.2+\hw)--(0.0,0.2);
      \draw[thick,-](0,0)--(0,1.2);
      \draw[thick,-](0.6,0)--(0.6,1.2);
      \node at (0.45,0.7+\hw*0.5) {$c$};
      \node at (0.3,0.2+\hw*0.5) {$b$};
    \end{scope}
    \node at (8.52,3.5) {%
      $\left[\vcenter{\hbox{%
            \begin{tikzpicture}
              \path[use as bounding box] (-.3,.1) -- (.2,-.9);
              \node (a) at (0,0) {$c$};
              \node (b) at (0,-.75) {$b$};
              \node at (-.17,0) {$\ibullet$};
              \node at (.17,0) {$\ibullet$};
              \node at (-.17,-.75) {$\vphantom{b}\ibullet$};
              \path[densely dashed, gray] (a) edge (b);
            \end{tikzpicture}
          }}\right]$};
    \node at (8.55,2) {$\terminator{\bigloset{c\\b}}{b}$};
    \begin{scope}[shift={(8.2,0)}]
      \filldraw[fill=blue!20!white,-](0.0,0.7)--(0.6,0.7)--(0.6,0.7+\hw)--(0.0,0.7+\hw)--(0.0,0.7);
      \filldraw[fill=pink!50!white,-](0.0,0.2)--(0.3,0.2)--(0.3,0.2+\hw)--(0.0,0.2+\hw)--(0.0,0.2);
      \draw[thick,-](0,0)--(0,1.2);
      \draw[thick,-](0.6,0)--(0.6,1.2);
      \node at (0.3,0.7+\hw*0.5) {$c$};
      \node at (0.15,0.2+\hw*0.5) {$b$};
    \end{scope}
    \node at (9.7,3.5) {%
      $\left[\vcenter{\hbox{%
            \begin{tikzpicture}
              \path[use as bounding box] (-.3,.1) -- (.2,-.9);
              \node (a) at (0,0) {$c$};
              \node (b) at (0,-.75) {$d$};
              \node at (-.17,0) {$\ibullet$};
              \node at (.17,0) {$\ibullet$};
              \node at (.17,-.75) {$\vphantom{b}\ibullet$};
              \path[densely dashed, gray] (a) edge (b);
            \end{tikzpicture}
          }}\right]$};
    \node at (9.65,2) {$\starter{\bigloset{c\\d}}{d}$};
    \begin{scope}[shift={(9.4,0)}]
      \filldraw[fill=blue!20!white,-](0.0,0.7)--(0.6,0.7)--(0.6,0.7+\hw)--(0.0,0.7+\hw)--(0.0,0.7);
      \filldraw[fill=yellow!50!white,-](0.3,0.2)--(0.6,0.2)--(0.6,0.2+\hw)--(0.3,0.2+\hw)--(0.3,0.2);
      \draw[thick,-](0,0)--(0,1.2);
      \draw[thick,-](0.6,0)--(0.6,1.2);
      \node at (0.3,0.7+\hw*0.5) {$c$};
      \node at (0.45,0.2+\hw*0.5) {$d$};
    \end{scope}
    \node at (10.9,3.5) {%
      $\left[\vcenter{\hbox{%
            \begin{tikzpicture}
              \path[use as bounding box] (-.3,.1) -- (.2,-.9);
              \node (a) at (0,0) {$c$};
              \node (b) at (0,-.75) {$d$};
              \node at (-.17,0) {$\ibullet$};
              \node at (-.17,-.75) {$\vphantom{b}\ibullet$};
              \node at (.17,-.75) {$\vphantom{b}\ibullet$};
              \path[densely dashed, gray] (a) edge (b);
            \end{tikzpicture}
          }}\right]$};
    \node at (10.9,2) {$\terminator{\bigloset{c\\d}}{c}$};
    \begin{scope}[shift={(10.6,0)}]
      \filldraw[fill=blue!20!white,-](0.0,0.7)--(0.3,0.7)--(0.3,0.7+\hw)--(0.0,0.7+\hw)--(0.0,0.7);
      \filldraw[fill=yellow!50!white,-](0.0,0.2)--(0.6,0.2)--(0.6,0.2+\hw)--(0.0,0.2+\hw)--(0.0,0.2);
      \draw[thick,-](0,0)--(0,1.2);
      \draw[thick,-](0.6,0)--(0.6,1.2);
      \node at (0.15,0.7+\hw*0.5) {$c$};
      \node at (0.3,0.2+\hw*0.5) {$d$};
    \end{scope}
    \node at (4.1,3.5) {$=$};
    \node at (5.5,3.5) {$*$};
    \node at (6.7,3.5) {$*$};
    \node at (7.9,3.5) {$*$};
    \node at (9.1,3.5) {$*$};
    \node at (10.3,3.5) {$*$};
    \node at (4.1,2) {$=$};
    \node at (5.5,2) {$*$};
    \node at (6.7,2) {$*$};
    \node at (7.9,2) {$*$};
    \node at (9.1,2) {$*$};
    \node at (10.3,2) {$*$};
    \node at (4.1,0.6) {$=$};
    \node at (5.5,0.6) {$*$};
    \node at (6.7,0.6) {$*$};
    \node at (7.9,0.6) {$*$};
    \node at (9.1,0.6) {$*$};
    \node at (10.3,0.6) {$*$};
  \end{tikzpicture}
  \caption{Sparse decomposition of ipomset into starters and terminators.}
  \label{fig:Ndecomp}
\end{figure}
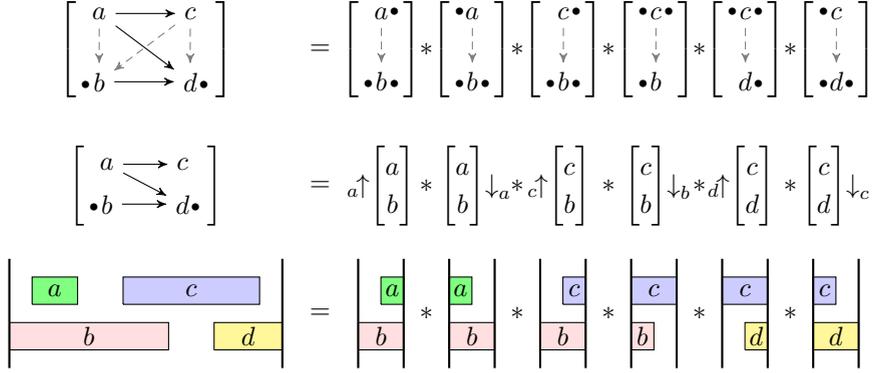

Any ipomset can be presented as a gluing of starters and terminators
\cite[Proposition~21]{DBLP:journals/iandc/FahrenbergJSZ22}.
(This is related to the fact that a partial order is interval
if and only if its antichain order is total,
see \cite{DBLP:journals/tcs/JanickiK93, book/Fishburn85, DBLP:journals/fuin/JanickiK19}).
Such a presentation we call a \emph{step decomposition};
if starters and terminators are alternating, the decomposition is \emph{sparse}.

\begin{example}
  Figure \ref{fig:Ndecomp} shows a sparse decomposition of an ipomset
  into starters and terminators.
  The top line shows the graphical representation,
  in the middle the representation using the notation we have introduced for starters and terminators,
  and the bottom line shows activity intervals.
\end{example}

We show that sparse step decompositions of ipomsets are unique.
For an ipomset $P$,
we denote by $P^m\subseteq P$ the subset of $<$-minimal elements and
\[
  P^s=\{p\in P\mid \forall\; p'\in P-P^m:\; p<p'\}.
\]
That is, $P^s$ contains precisely those minimal elements
which have arrows to all non-minimal elements.
Clearly, both $P^m$ and $P^s$ are conclists, and $P^s\subseteq P^m\supseteq S_P$.
We need a few technical lemmas.

\begin{lemma}
  \label{l:STGluing}
  Let $P$ be an ipomset, $U$ a conclist, and $A\subseteq U$.
  \begin{enumerate}
  \item
    Assume that $U\cong S_P$ and $P'=\starter{U}{A}* P$.
    Then $P'$ and $P$ are isomorphic as pomsets,
    $T_P\cong T_{P'}$ and $S_{P'}\cong S_P-A$.
  \item
    Assume that $U-A\cong S_P$ and $P'=\terminator{U}{A}* P$.
    Then $P'\cong P\cup A$ as sets, and $P\cong P'-A$ as pomsets,
    $T_P\cong T_{P'}$ and $S_{P'}\cong U$. \qed
  \end{enumerate}
\end{lemma}

\begin{proof}
  Simple calculations.
\end{proof}

Consider a presentation $P\cong Q R$.
From the definition follows that $P^m\cong Q^m$ and $S_P\cong S_Q$
This implies:

\begin{lemma}
  \label{l:StartCrit}
  Assume that $P\cong Q R$ and $Q$ is either a (non-identity) starter or a terminator.
  Then $Q$ is a starter if{}f $S_P\subsetneq P^m$,
  and $Q$ is a terminator if{}f $S_P=P^m$.
\end{lemma}

\begin{proof}
  We have $P^m\cong Q=Q^m$ and $S_P\cong S_Q$.
  But $Q$ is a terminator if and only if $S_Q=Q$, and a (non-identity) starter if and only if $S_Q\subsetneq Q$.
  \qed
\end{proof}

\begin{lemma}
  \label{l:SparsePreLemma}
  Assume that $P\cong Q Q' R$.
  \begin{enumerate}
  \item
    If $Q$ is a non-identity starter and $Q'$ is a non-identity terminator, then $Q\cong \starter{P^m}{P^m-S_P}$.
  \item
    If $Q$ is a non-identity terminator and $Q'$ is a non-identity starter, then $Q\cong\terminator{P^m}{P^s}$.
  \end{enumerate}
\end{lemma}

\begin{proof}
  Consider the first case. Then $P$ and $Q' R$ are isomorphic as pomsets,
  and
  \[
    Q=T_Q\cong S_{Q' R}
    \overset{\text{Lemma \ref{l:StartCrit}}}=
    (Q' R)^m
    \overset{\text{Lemma \ref{l:STGluing}}}\cong
    P^m.
  \]	
  Equality $S_Q=S_P$ follows immediately from the definition.

\medskip
  In the second case, we have $Q=S_Q\cong S_P\overset{\text{Lemma \ref{l:StartCrit}}}=P^m$,
  and $Q' R\cong P-(Q-T_Q)$ as pomsets (Lemma \ref{l:STGluing}).	
  By Lemma \ref{l:StartCrit} we have
  \[
    P^m\cap (Q' R)= Q\cap (Q' R)=T_Q\cong S_{Q' R}
    \overset{\text{Lemma \ref{l:StartCrit}}}\subsetneq
    (Q' R)^m.
  \]	
  Hence there exists an element $p\in Q' R$
  that is minimal in $Q' R$ but not in $P$.
  For every $p'\in P^s$ we have $p'<p$ and, therefore, $p'\not\in Q' R$.
  As a consequence, $P^s\subseteq P-(Q' R)=Q-T_Q$ (Lemma \ref{l:STGluing}).
	
\medskip
  On the other hand, if $p'\in P^m-P^s$, then there exists $p\in P-P^m=P-Q$
  such that $p'\not< p$. Thus, $p'$ must belong to $T_Q$. \qed
\end{proof}

\begin{proposition}
  \label{p:SparsePresentation}
  Every ipomset $P$ has a unique sparse step decomposition.
\end{proposition}

\begin{proof}
  Let $P = P_1*\dotsm*P_n = Q_1*\dotsm*Q_m$ be two sparse presentations.
  If $n=1$, then $m=1$ and equality follows trivially,
  so assume $n, m\ge 2$ and write $P_2*\dotsm*P_n=P'$ and $Q_2*\dotsm*Q_m=Q'$.

\medskip
  Assume first that $P_1$ is a starter.
  By Lemma \ref{l:SparsePreLemma}, $P_1\cong \starter{P^m}{P^m-S_P}$.
  By Lemma \ref{l:STGluing},
  $S_P\cong S_{P'}-(P^m-S_P)$.
  Hence $S_{P'}\cong P^m$, implying $S_P\subsetneq P^m$.
  By Lemma \ref{l:StartCrit}, $Q_1$ is a starter.
  By Lemma \ref{l:SparsePreLemma}, $Q_1\cong \starter{P^m}{P^m-S_P}$.
  Thus $P_1\cong Q_1$, and we may proceed inductively with $P'=Q'$.

  Now assume instead that $P_1$ is a terminator.
  By Lemma \ref{l:SparsePreLemma}, $P_1\cong \terminator{P^m}{P^s}$.
  By Lemma \ref{l:STGluing},
  $S_P\cong P^m$.
  By Lemma \ref{l:StartCrit}, $Q_1$ is a terminator.
  By Lemma \ref{l:SparsePreLemma}, $Q_1\cong \terminator{P^m}{P^s}$.
  Thus $P_1\cong Q_1$, and we may proceed inductively with $P'=Q'$. \qed
\end{proof}

\section{Higher-dimensional automata and their languages}
\label{se:hda}

An HDA is a collection of \emph{cells} which are connected according to specified \emph{face maps}.
Each cell has an associated list of \emph{labelled events}
which are interpreted as being executed in that cell,
and the face maps may terminate some events or, inversely,
indicate cells in which some of the current events were not yet started.
Additionally, some cells are designated \emph{start} cells and some others \emph{accept} cells;
computations of an HDA begin in a start cell and proceed by starting and terminating events
until they reach an accept cell.

\subsection{Precubical sets and HDAs}

To make the above precise,
let $\sq$ denote the set of conclists.
A \emph{precubical set} consists of a set of cells $X$
together with a mapping $\ev: X\to \sq$ which to every cell assigns its list of active events.
For a conclist $U$ we write $X[U]=\{x\in X\mid \ev(x)=U\}$ for the cells of type $U$.
Further, for every $U\in \sq$ and subset $A\subseteq U$ there are \emph{face maps}
$\delta_A^0, \delta_A^1: X[U]\to X[U-A]$.
The \emph{upper} face maps $\delta_A^1$ terminate the events in $A$,
whereas the \emph{lower} face maps $\delta_A^0$ ``unstart'' these events:
they map cells $x\in X[U]$ to cells $\delta_A^0(x)\in X[U-A]$
where the events in $A$ are not yet active.

If $A, B\subseteq U$ are disjoint,
then the order in which events in $A$ and $B$ are terminated or unstarted
should not matter, so we require that $\delta_A^\nu \delta_B^\mu = \delta_B^\mu \delta_A^\nu$
for $\nu, \mu\in\{0, 1\}$: the \emph{precubical identities}.
A \emph{higher-dimensional automaton} (\emph{HDA})
is a precubical set $X$ together with subsets $\bot_X, \top_X\subseteq X$
of \emph{start} and \emph{accept} cells.
For a precubical set $X$ and subsets $Y, Z\subseteq X$
we denote by $X_Y^Z$ the HDA with precubical set $X$, start cells $Y$ and accept cells $Z$.
We do \emph{not} generally assume that precubical sets or HDAs are finite.
The \emph{dimension} of an HDA $X$ is $\dim(X)=\sup\{|\ev(x)|\mid x\in X\}\in \Nat\cup\{\infty\}$.

\begin{example}
  One-dimensional HDAs $X$ are standard automata.
  Cells in $X[\emptyset]$ are states,
  cells in $X[a]$ for $a\in \Sigma$ are $a$-labelled transitions.
  Face maps $\delta_a^0$ and $\delta_a^1$ attach source and target states to transitions.
  In contrast to ordinary automata we allow start and accept \emph{transitions}
  instead of merely states,
  so languages of such automata may contain not only words
  but also ``words with interfaces''.
  In any case, at most one event is active at any point in time,
  so the event order is unnecessary.
\end{example}

\begin{figure}[h]
\vspace*{-2mm}
  \centering
  \begin{tikzpicture}[x=1cm, y=.95cm]
    \node[circle,draw=black,fill=black!10,inner sep=0pt,minimum size=15pt]
    (aa) at (0,0) {$\vphantom{hy}v$};
    \node[circle,draw=black,fill=black!10,inner sep=0pt,minimum size=15pt]
    (ac) at (0,4) {$\vphantom{hy}x$};
    \node[circle,draw=black,fill=black!10,inner sep=0pt,minimum size=15pt]
    (ca) at (4,0) {$\vphantom{hy}w$};
    \node[circle,draw=black,fill=black!10,inner sep=0pt,minimum size=15pt]
    (cc) at (4,4) {$\vphantom{hy}y$};
    \node[circle,draw=black,fill=red!30,inner sep=0pt,minimum size=15pt]
    (ba) at (2,0) {$\vphantom{hy}e$};
    \node[circle,draw=black,fill=red!30,inner sep=0pt,minimum size=15pt]
    (bc) at (2,4) {$\vphantom{hy}f$};
    \node[circle,draw=black,fill=green!30,inner sep=0pt,minimum size=15pt]
    (ab) at (0,2) {$\vphantom{hy}g$};
    \node[circle,draw=black,fill=green!30,inner sep=0pt,minimum size=15pt]
    (cb) at (4,2) {$\vphantom{hy}h$};
    \node[circle,draw=black,fill=yellow!60,inner sep=0pt,minimum size=15pt]
    (bb) at (2,2) {$\vphantom{hy}q$};
    \node[right] at (5,4) {$X[\emptyset]=\{v,w,x,y\}$};
    \node[right] at (5,3.2) {$X[a]=\{e,f\}$};
    \node[right] at (5,2.4) {$X[b]=\{g,h\}$};
    \node[right] at (5,1.6) {$X[ab]=\{q\}$};
    \path (ba) edge node[above] {$\delta^0_a$} (aa);
    \path (ba) edge node[above] {$\delta^1_a$} (ca);
    \path (bb) edge node[above] {$\delta^0_a$} (ab);
    \path (bb) edge node[above] {$\delta^1_a$} (cb);
    \path (bc) edge node[above] {$\delta^0_a$} (ac);
    \path (bc) edge node[above] {$\delta^1_a$} (cc);
    \path (ab) edge node[left] {$\delta^0_b$} (aa);
    \path (ab) edge node[left] {$\delta^1_b$} (ac);
    \path (bb) edge node[left] {$\delta^0_b$} (ba);
    \path (bb) edge node[left] {$\delta^1_b$} (bc);
    \path (cb) edge node[left] {$\delta^0_b$} (ca);
    \path (cb) edge node[left] {$\delta^1_b$} (cc);
    \path (bb) edge node[above left] {$\delta^1_{ab}\!\!$} (cc);
    \path (bb) edge node[above left] {$\delta^0_{ab}\!\!$} (aa);
    \node[below left] at (aa) {$\bot\;$};
    \node[above right] at (cb) {$\;\top$};
    \node[above right] at (cc) {$\;\top$};
    \node[right] at (5,0.8) {$\bot_X=\{v\}$};
    \node[right] at (5,0) {$\top_X=\{h,y\}$};
    \begin{scope}[shift={(8.5,1)}]
      \filldraw[color=black!10!white] (0,0)--(2,0)--(2,2)--(0,2)--(0,0);			
      \filldraw (0,0) circle (0.05);
      \filldraw (2,0) circle (0.05);
      \filldraw (0,2) circle (0.05);
      \filldraw (2,2) circle (0.05);
      \path (0,0) edge node[below,color=red!70!black] {$a$} (1.95,0);
      \path (0,2) edge (1.95,2);
      \path (0,0) edge node[left,color=green!70!black] {$b$} (0,1.95);
      \path (2,0) edge (2,1.95);
      \node[left] at (0,0) {$\bot$};
      \node[right] at (2,2) {$\top$};
      \node[right] at (2,1) {$\top$};
    \end{scope}
  \end{tikzpicture}\vspace*{-1mm}
  \caption{A two-dimensional HDA $X$ on $\Sigma=\{{\color{red!70!black}{a}}, {\color{green!70!black}{b}}\}$, see Example \ref{ex:abcube}.}
  \label{fig:abcube}\vspace*{-3mm}
\end{figure}
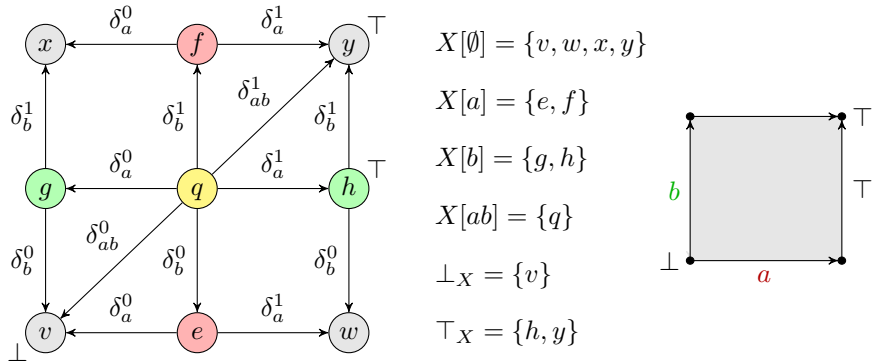

\begin{example}
  \label{ex:abcube}
  Figure \ref{fig:abcube} shows a two-dimensional HDA $X$
  both as a combinatorial object (left) and in a more geometric realisation (right).
  We write isomorphism classes of conclists as lists of labels
  and omit the set braces in $\delta_{\{a\}}^0$ etc.
  $X$ has four zero-dimensional cells, or states, displayed in grey on the left;
  four one-dimensional transitions, two labelled $a$ and displayed in red
  and two labelled $b$ and shown in green;
  and one two-dimensional cell displayed in yellow.
\end{example}

An \emph{HDA-map} between HDAs $X$ and $Y$ is a function $f: X\to Y$ that preserves structure:
types of cells ($\ev_Y\circ f=\ev_X$),
face maps ($f(\delta^\nu_A(x))=\delta^\nu_A(f(x))$)
and start/accept cells ($f(\bot_X)\subseteq \bot_Y$, $f(\top_X)\subseteq \top_Y$).
Similarly, a precubical map is a function that preserves the first two of these three.
HDAs and HDA-maps form a category, as do precubical sets and precubical maps.

\subsection{Paths and their labels}

Computations of HDAs are paths:
sequences of cells connected by face maps.
A \emph{path} in $X$ is, thus, a sequence
\begin{equation}
  \label{eq:path}
  \alpha=(x_0, \phi_1, x_1, \dotsc, x_{n-1}, \phi_n, x_n),
\end{equation}
where the $x_i$ are cells of $X$ and the $\phi_i$ indicate types of face maps:
for every $i$, $(x_{i-1}, \phi_i, x_i)$ is either
\begin{itemize}
\item $(\delta^0_A(x_i), \arrO{A}, x_i)$ for $A\subseteq \ev(x_i)$ (an \emph{upstep})
\item or $(x_{i-1}, \arrI{B}, \delta^1_B(x_{i-1}))$ for $B\subseteq \ev(x_{i-1})$ (a \emph{downstep}).
\end{itemize}
Upsteps start events in $A$ while downsteps terminate events in $B$.
The \emph{source} and \emph{target} of $\alpha$ as in \eqref{eq:path} are $\src(\alpha)=x_0$ and $\tgt(\alpha)=x_n$.

\medskip
The set of all paths in $X$ starting at $Y\subseteq X$ and terminating in $Z\subseteq X$
is denoted by $\Path(X)_Y^Z$;
we write $\Path(X)_Y=\Path(X)_Y^X$, $\Path(X)^Z=\Path(X)_X^Z$, and $\Path(X)=\Path(X)_X^X$.
A path $\alpha$ is \emph{accepting} if $\src(\alpha)\in \bot_X$ and $\tgt(\alpha)\in \top_X$.
Paths $\alpha$ and $\beta$ may be concatenated
if $\tgt(\alpha)=\src(\beta)$;
their concatenation is written $\alpha*\beta$,
and we omit the ``$*$'' in concatenations if convenient.

\emph{Path equivalence} is the congruence $\simeq$
generated by $(z\arrO{A} y\arrO{B} x)\simeq (z\arrO{A\cup B} x)$,
$(x\arrI{A} y\arrI{B} z)\simeq (x\arrI{A\cup B} z)$, and
$\gamma \alpha \delta\simeq \gamma \beta \delta$ whenever $\alpha\simeq \beta$.
Intuitively, this relation allows to assemble subsequent upsteps or downsteps into one ``bigger'' step.
A path is \emph{sparse} if its upsteps and downsteps are alternating,
so that no more such assembling may take place.
Every equivalence class of paths contains a unique sparse path.

\begin{example}
  In one-dimensional HDAs,
  paths are sequences of transitions connected at states.
  Path equivalence is a trivial relation, and all paths are sparse.
\end{example}

\begin{example}
  \label{ex:paths}
  The HDA $X$ of Figure~\ref{fig:abcube} admits
  five sparse accepting paths:
  \begin{gather*}
    v\arrO{a} e\arrI{a} w\arrO{b} h, \qquad\quad
    v\arrO{a} e\arrI{a} w\arrO{b} h\arrI{b} y, \\
    v\arrO{ab} q\arrI{a} h, \qquad
    v\arrO{ab} q\arrI{ab} y, \qquad
    v\arrO{b} g\arrI{b} x\arrO{a} f\arrI{a} y.
  \end{gather*}
\end{example}

The observable content or \emph{event ipomset} $\ev(\alpha)$
of a path $\alpha$ is defined recursively as follows:
\begin{itemize}
\item If $\alpha=(x)$, then
  $\ev(\alpha)=\id_{\ev(x)}$.
\item If $\alpha=(y\arrO{A} x)$, then
  $\ev(\alpha)=\starter{\ev(x)}{A}$.
\item If $\alpha=(x\arrI{B} y)$, then
  $\ev(\alpha)=\terminator{\ev(x)}{B}$.
\item If $\alpha=\alpha_1*\dotsm*\alpha_n$ is a concatenation, then
  $\ev(\alpha)=\ev(\alpha_1)*\dotsm*\ev(\alpha_n)$.
\end{itemize}
\cite[Lemma~8]{conf/concur/FahrenbergJSZ22} shows that $\alpha\simeq \beta$ implies $\ev(\alpha)=\ev(\beta)$.
Further, if $\alpha=\alpha_1*\dotsm*\alpha_n$ is a sparse path,
then $\ev(\alpha)=\ev(\alpha_1)*\dotsm*\ev(\alpha_n)$ is a sparse step decomposition.

\begin{figure}[tbp]
  \centering
  \begin{tikzpicture}[x=1cm, y=.95cm]
    \begin{scope}
      \filldraw[color=black!10!white] (0,0)--(4,0)--(4,4)--(2,4)--(2,2)--(0,2)--(0,0);			
      \filldraw (0,0) circle (0.05);
      \filldraw (2,0) circle (0.05);
      \filldraw (4,0) circle (0.05);
      \filldraw (0,2) circle (0.05);
      \filldraw (2,2) circle (0.05);
      \filldraw (4,2) circle (0.05);
      \filldraw (2,4) circle (0.05);
      \filldraw (4,4) circle (0.05);
      \path (0,0) edge node[below] {$a$} (1.95,0);
      \path (2,0) edge node[below] {$c$} (3.95,0);
      \path (0,2) edge (1.95,2);
      \path (2,2) edge (3.95,2);
      \path (2,4) edge (3.95,4);
      \path (0,0) edge (0,1.95);
      \path (2,0) edge (2,1.95);
      \path (2,2) edge node[left] {$\vphantom{y}d$} (2,3.95);
      \path (4,0) edge node[right] {$\vphantom{y}b$} (4,1.95);
      \path (4,2) edge (4,3.95);
      \node[left] at (0,1) {$\bot$};
      \node[right] at (4,3) {$\top$};
      \node at (1,1) {$\vphantom{by}x$};
      \node at (3,1) {$\vphantom{b}y$};
      \node at (3,3) {$\vphantom{by}z$};
    \end{scope}
  \end{tikzpicture}
  \caption{HDA $Y$ consisting of three squares glued along common faces.}
  \label{fig:ab*cb*cd}\vspace*{-2mm}
\end{figure}
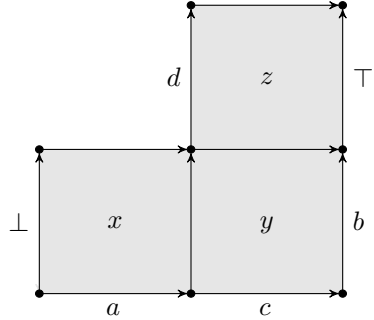

\begin{example}
  Event ipomsets of paths in one-dimensional HDAs are words, possibly with interfaces.
  Sparse step decompositions of words are obtained by splitting symbols into starts and terminations,
  for example, $\ibullet a b=\ibullet a*{b\ibullet}*{\ibullet b}$.
\end{example}

\begin{example}
  The event ipomsets of the five sparse accepting paths in the HDA $X$ of Figure~\ref{fig:abcube}
  are $a b\ibullet$, $a b$, $\loset{a\\b&\!\!\ibullet\!}$, $\loset{a\\b}$, and $b a$.
  Figure \ref{fig:ab*cb*cd} shows another HDA, which admits an accepting path
  \begin{equation*}
    (\delta_a^0 x\arrO{a} x\arrI{a} \delta_a^1 x\arrO{c} y\arrI{b} \delta_b^1 y\arrO{d} z\arrI{c} \delta_c^1 z).
  \end{equation*}
  Its event ipomset is precisely the ipomset of Figure~\ref{fig:Ndecomp},
  with the indicated sparse step decomposition arising from the sparse presentation above.
\end{example}

\subsection{Languages of HDAs}

The \emph{language} of an HDA $X$ is
\begin{equation*}
  \Lang(X) = \{ \ev(\alpha)\mid \alpha \text{ accepting path in } X\}.
\end{equation*}
\cite[Proposition~10]{conf/concur/FahrenbergJSZ22} shows that languages of HDAs
are sets of ipomsets which are closed under subsumption,
\ie~languages in the sense of Definition~\ref{de:langs}.

\medskip
A language is \emph{regular} if it is the language of a finite HDA.

\begin{example}
  The languages of our example HDAs are
  \begin{equation*}
    \Lang(X) = \big\{ \loset{a\\b&\!\!\ibullet\!}, \loset{a\\b}\big\}\down
    =\big\{ \loset{a\\b&\!\!\ibullet\!}, a b\ibullet, \loset{a\\b}, ab, ba\big\}
  \end{equation*}
  and
  \begin{equation*}
    \Lang(Y) = \left\{ \left[\vcenter{\hbox{\!%
            \begin{tikzpicture}[x=1cm, y=.7cm]
              \node (a) at (0,0) {$a$};
              \node (c) at (1,0) {$c$};
              \node at (1.2,-.79) {$\ibullet$};
              \node (b) at (0,-.75) {$b$};
              \node at (-.2,-.79) {$\ibullet$};
              \node (d) at (1,-.75) {$d$};
              \path (a) edge (c);
              \path (b) edge (d);
              \path (a) edge (d);
            \end{tikzpicture}
          \!\!}}\right]\right\}\down.
  \end{equation*}
\end{example}

We say that a cell $x\in X$ in an HDA $X$ is
\begin{itemize}
\itemsep=0.9pt
\item
  \emph{accessible} if $\Path(X)_\bot^x\ne \emptyset$, \ie
  $x$ can be reached by a path from a start cell;
\item
  \emph{coaccessible} if $\Path(X)_x^\top\ne \emptyset$, \ie
  there is a path from $x$ to an accept cell;
\item
  \emph{essential} if it is both accessible and coaccessible.
\end{itemize}
A path is \emph{essential} if its source and target cells are essential.
This implies that all its cells are essential.
Segments of accepting paths are always essential.

\eject

The set of essential cells of $X$ is denoted by $\ess(X)$;
this is not necessarily a sub-HDA of $X$
given that faces of essential cells may be non-essential.
For example, all bottom cells of the HDA $Y$
in Figure~\ref{fig:ab*cb*cd} are inaccessible and hence non-essential.

\begin{lemma}
  \label{l:HDAGen}
  Let $X$ be an HDA.
  There exists a smallest sub-HDA $X^{\ess}\subseteq X$ that contains all essential cells,
  and $\Lang(X^\ess)=\Lang(X)$.
  If $\ess(X)$ is finite, then $X^{\ess}$ is also finite.
\end{lemma}

\begin{proof}
  The set of all faces of essential cells
  \[
    X^{\ess}=\{\delta^0_A\delta^1_B(x)\mid x\in \ess(X),\; A,B\subseteq \ev(x),\; A\cap B=\emptyset\}
  \]
  is a sub-HDA of $X$.
  Clearly every sub-HDA of $X$ that contains $\ess(X)$ must also contain $X^\ess$.
  Since all accepting paths are essential, $\Lang(X^\ess)=\Lang(X)$.
  If $|\ess(X)|=n$ and $|\ev(x)|\le d$ for all $x\in \ess(X)$,
  then $|X^{\ess}|\le n\cdot 3^d$,
  since a cell of dimension $\leq d$ has at most $3^d$ faces.
  \qed
\end{proof}

\subsection{Track objects}

Track objects, introduced in \cite{Hdalang},
provide a mapping from ipomsets to HDAs and are a powerful tool for reasoning about languages.
Below we adapt the definition from \cite[Section~5.3]{Hdalang}.

\begin{definition}
  \label{de:pobj}
  The \emph{track object} of an ipomset $P$
  is the HDA $\sq^P$ defined as follows:
  \begin{itemize}
  \item
    $\sq^P$ is the set of all functions
    $x:P\to\{ 0, *, 1\}$
    such that
    \[
      p<q \implies (x(p),x(q))\in\{(0,0),(*,0),(1,0),(1,*),(1,1)\}.
    \]
  \item
    For $x\in \sq^P$, $\ev(x)=x^{-1}(*)$\quad (the condition above ensures that $x^{-1}(*)$ is discrete);
  \item
    For $x\in\sq^P$, $\nu\in\{0,1\}$ and $A\subseteq \ev(x)$,
    \begin{equation*}
      \delta_A^\nu( x)( p)=
      \begin{cases}
        \nu &\text{for } p\in A, \\
        x( p) &\text{for } p\not\in A;
      \end{cases}
    \end{equation*}
  \item $\bot_{\sq^P}=\{c_\bot^P\}$ and $\top_{\sq^P}=\{c^\top_P\}$, where
    \begin{equation*}
      c_\bot^P(p)=
      \begin{cases}
        * & \text{if } p\in S_P, \\
        0 & \text{if } p\not\in S_P,
      \end{cases}
      \qquad
      c^\top_P(p)=
      \begin{cases}
        * & \text{if } p\in T_P, \\
        1 & \text{if } p\not\in T_P;
      \end{cases}
    \end{equation*}
  \end{itemize}
\end{definition}

We list some properties of track objects needed later.

\begin{lemma}
  \label{l:PathTrack}
  Let $X$ be an HDA, $x,y\in X$ and $P\in\iiPoms$.
  The following conditions are equivalent:
  \begin{enumerate}
  \item There exists a path $\alpha\in\Path(X)_x^y$ such that $\ev(\alpha)=P$.
  \item There is an HDA-map $f:\sq^P\to X_x^y$ (\ie $f(c_\bot^P)=x$ and $f(c^\top_P)=y$).
  \end{enumerate}
\end{lemma}

\eject
\begin{proof}
  This is an immediate consequence of \cite[Proposition~89]{Hdalang}. \qed
\end{proof}

\begin{lemma}
  \label{l:SubsuTrack}
  If $P, Q\in \iiPoms$ are such that $P\subsu Q$, then there exists an HDA-map $\sq^P\to \sq^Q$.
\end{lemma}

\begin{proof}
  This is \cite[Lemma 63]{Hdalang}. \qed
\end{proof}

\begin{lemma}
  \label{l:SubsuPath}
  Let $X$ be an HDA, $x,y\in X$, $\beta\in\Path(X)_x^y$ and $P\subsu Q=\ev(\beta)$.
  Then there exists $\alpha\in\Path(X)_x^y$
  such that $\ev(\alpha)=P$.
\end{lemma}

\begin{proof}
  This follows immediately from Lemmas \ref{l:PathTrack} and \ref{l:SubsuTrack}.
\end{proof}

\begin{lemma}
  \label{l:PathDivision}
  Let $X$ be an HDA, $x,y\in X$ and $\gamma\in\Path(X)_x^y$.
  Assume that $\ev(\gamma)=P*Q$ for ipomsets $P$ and $Q$.
  Then there exist paths $\alpha\in\Path(X)_x$ and $\beta\in\Path(X)^y$
  such that $\ev(\alpha)=P$, $\ev(\beta)=Q$ and $\tgt(\alpha)=\src(\beta)$.
\end{lemma}

\begin{proof}
  By Lemma \ref{l:PathTrack},
  there is an HDA-map $f: \sq^{P Q}\to X_x^y$.
  By \cite[Lemma~65]{Hdalang},
  there exist precubical maps $j_P: \sq^P\to \sq^{P Q}$, $j_Q: \sq^Q\to \sq^{P Q}$
  such that $j_P(c_\bot^P)=c_\bot^{P Q}$, $j_P(c^\top_P)=j_Q(c_\bot^Q)$ and $j_Q(c^\top_Q)=c^\top_{P Q}$.
  Let $z=f(j_P(c_\bot^P))$,
  then $f\circ j_P: \sq^P\to X_x^z$ and $f\circ j_Q: \sq^Q\to X_z^y$ are HDA-maps,
  and by applying Lemma \ref{l:PathTrack} again to $j_P$ and $j_Q$
  we obtain $\alpha$ and $\beta$. \qed
\end{proof}

\section{Myhill-Nerode theorem}
\label{se:mn}

The \emph{prefix quotient} of
a language $L\in \Langs$ by an ipomset $P$ is the language
\begin{equation*}
  P\lquo L=\{Q\in \iiPoms\mid P Q\in L\}.
\end{equation*}
Similarly, the \emph{suffix quotient} of $L$ by $P$ is $L/P=\{Q\in \iiPoms\mid Q P\in L\}$.
Denote
\[
  \suff(L)=\{P\lquo L\mid P\in \iiPoms\},
  \qquad
  \pref(L)=\{L/P\mid P\in\iiPoms\}.
\]
We record the following property of quotient languages.

\begin{lemma}
  \label{l:QuotIncl}
  If $L$ is a language and $P\subsu Q$,
  then $Q\lquo L\subseteq P\lquo L$.
\end{lemma}

\begin{proof}
  If $P\subsu Q$, then $PR\subsu QR$.
  Thus,
  \[
    R\in Q\lquo L
    \iff
    QR\in L
    \implies
    PR\in L
    \iff
    R\in P\lquo L.  \vspace*{-7mm}
  \]	
  \end{proof}

The main goal of this section is to show the following.

\begin{theorem}
  \label{t:MN}
  For a language $L\in\Langs$
  the following conditions are equivalent.
  \begin{itemize}
  \item[\rm (a)] $L$ is regular.
  \item[\rm (b)] The set $\suff(L)\subseteq\Langs$ is finite.
  \item[\rm (c)] The set $\pref(L)\subseteq\Langs$ is finite.
  \end{itemize}
\end{theorem}

We prove only the equivalence between (a) and (b);
equivalence between (a) and (c) is symmetric.
First we prove the implication (a)$\implies$(b).
Let $X$ be an HDA with $\Lang(X)=L$.
For $x\in X$ define
languages $\mathsf{Pre}(x)=\Lang(X_\bot^x)$
and $\mathsf{Post}(x)=\Lang(X_x^\top)$.

\begin{lemma}
  For every $P\in\iiPoms$,
    $P\lquo L = \bigcup\{\mathsf{Post}(x)\mid x\in X,\; P\in \mathsf{Pre}(x)  \}$.
\end{lemma}

\begin{proof}
  We have
  \begin{align*}
    Q\in P\lquo L
    \iff
    P Q\in L
   	\;\;\overset{\text{Lem.\@ \ref{l:PathTrack}}}\iff\;\; &
    \exists\; f: \sq^{P Q}\to X = X_\bot^\top \\
    \;\;\overset{\text{Lem.\@ \ref{l:PathDivision}}}\iff\;\; &
    \exists\; x\in X, g: \sq^{P}\to X_\bot^x,\; h: \sq^{Q}\to X_x^\top \\
   	\;\;\overset{\text{Lem.\@ \ref{l:PathTrack}}}\iff\;\; &
    \exists\; x\in X:
    P\in \Lang(X_\bot^x),\;
    Q\in \Lang(X_x^\top) \\
    \;\;\overset{\phantom{\text{Lem.\@ \ref{l:PathTrack}}}}\iff\;\; &
    \exists\; x\in X:
    P\in \mathsf{Pre}(x),\;
    Q\in \mathsf{Post}(x).
  \end{align*}
  The last condition says that $Q$ belongs to the right-hand side of the equation. \qed
\end{proof}

\begin{varproof}[of Theorem~\ref{t:MN}, {\rm (a)$\implies$(b)}]
  The family of languages $\{P\lquo L\mid P\in \iiPoms\}$
  is a subfamily of
  $ \{\bigcup_{x\in Y} \mathsf{Post}(x)\bigmid Y\subseteq X\}$
  which is finite. \qed
\end{varproof}

\subsection{HDA construction}
\label{se:MN(L)}

Now we show that (b) implies (a).
Fix a language $L\in\Langs$, with $\suff(L)$ finite or infinite.
We will construct an HDA $\MN(L)$ that recognises $L$
and show that if $\suff(L)$ is finite,
then the essential part $\MN(L)^\ess$ is finite.
The cells of $\MN(L)$ are equivalence classes of ipomsets
under a relation $\Lapprox$ induced by $L$ which we will introduce below.
The relation $\Lapprox$ is defined using prefix quotients,
but needs to be stronger than prefix quotient equivalence.
This is because events may be concurrent and because ipomsets have interfaces.
We give examples just after the construction.

\medskip
For an ipomset $\ilo{S}{P}{T}$ define its \emph{(target) signature}
to be the starter $\fin(P)
=\starter{T}{T-S}$.
Thus $\fin(P)$ collects all target events of $P$,
and its source interface contains those events that are also in the source interface of $P$.
We also write $\rfin(P)=T-S\subseteq \fin(P)$:
the set of all target events of $P$ that are not source events.
An important property is that removing elements of $\rfin(P)$
does not change the source interface of $P$.
For example,
\[
  \fin
  \left(
    \loset{\!\ibullet\!\!&a&\!\!\ibullet\! \\
      \!\ibullet\!\!&a\\
      &c&\!\!\ibullet\!		
    }
  \right)
  =
  \loset{\!\ibullet\!\!&a&\!\!\ibullet\! \\
    &c&\!\!\ibullet\!		
  }
  ,
  \quad
  \fin
  \left(
    \loset{\!\ibullet\!\!&a c&\!\!\ibullet\!\\\!\ibullet\!\!&b&\!\!\ibullet\!}
  \right)
  =
  \loset{&c&\!\!\ibullet\!\\ \!\ibullet\!\!&b&\!\!\ibullet\!},
  \quad
  \fin
  \left(
    \loset{a c&\!\!\ibullet\!\\b&\!\!\ibullet\!}
  \right)
  =
  \loset{c&\!\!\ibullet\!\\ b&\!\!\ibullet\!};
\]
$\rfin$ is $\{c\}$ in the first two examples and equal to $\loset{c\\b}$ in the last.

\eject
We define two equivalence relations on $\iiPoms$ induced by $L$:
\begin{itemize}
\item
  Ipomsets $P$ and $Q$ are \emph{weakly equivalent} ($P\Lsim Q$)
  if $\fin(P)\cong \fin(Q)$ and $P\lquo L=Q\lquo L$.
  Obviously, $P\Lsim Q$ implies $T_P\cong T_Q$
  and $\rfin(P)\cong \rfin(Q)$.
\item
  Ipomsets $P$ and $Q$ are \emph{strongly equivalent} ($P\Lapprox Q$)
  if
  $P\Lsim Q$ and
  for all $A\subseteq \rfin(P)\cong\rfin(Q)$
  we have $(P-A)\lquo L=(Q-A)\lquo L$.
\end{itemize}
Evidently $P\Lapprox Q$ implies $P\Lsim Q$, but the inverse does not always hold.
We explain in Example \ref{ex:strongeq} below
why $\Lapprox$, and not $\Lsim$, is the proper relation to use for constructing $\MN(L)$.

\begin{lemma}
\label{l:StrongEqDef}
  If $P\Lapprox Q$,
  then $P-A\Lapprox Q-A$ for all $A\subseteq \rfin(P)\cong\rfin(Q)$.
\end{lemma}

\begin{proof}
	For every $A$ we have $(P-A)\lquo L=(Q-A)\lquo L$, and
	\[
		\fin(P-A)=\fin(P)-A\cong \fin(Q)-A=\fin(Q-A),
	\]
	Thus, $P-A\Lsim Q-A$.
	Further, for every $B\subseteq \rfin(P-A)\cong \rfin(Q-A)$,
	\[
		((P-A)-B)\lquo L
		=
		(P-(A\cup B))\lquo L
		=
		(Q-(A\cup B))\lquo L
		=
		((Q-A)-B)\lquo L,
	\]
	which shows that $P-A\Lapprox Q-A$. \qed
\end{proof}

Now define an HDA $\MN(L)$ as follows.
For $U\in\sq$, write $\iiPoms_U=\{P\in\iiPoms\mid T_P\cong U\}$ and let
\[
  \MN(L)[U] = \iiPoms_U/{\Lapprox} \cup \{w_U\},
\]
where the $w_U$ are new \emph{subsidiary} cells which are introduced
solely to define some lower faces. (They will not affect the language of $\MN(L)$).

\medskip
The $\Lapprox$-equivalence class of $P$
will be denoted by $\eqcl{P}$ (but often just $P$ in examples).
Face maps are defined as follows,
for $A\subseteq U\in\sq$ and $P\in \iiPoms_U$:
\begin{equation}
\label{e:FaceMaps}
	\delta^0_A(\eqcl{P})
	=
	\begin{cases}
		\eqcl{P-A} & \text{if $A\subseteq\rfin(P)$}, \\
		w_{U-A} & \text{otherwise},
	\end{cases}
	\qquad
	\delta^1_A(\eqcl{P})
	=
	\eqcl{P*\terminator{U}{A}},
\end{equation}
\[
	\delta^0_A(w_U)=\delta^1_A(w_U)=w_{U-A}.
\]
In other words,
if $A$ has no source events of $P$, then
$\delta^0_A$ removes $A$ from $P$ (the source interface of $P$ is unchanged).
If $A$ contains any source event, then $\delta^0_A(P)$ is a subsidiary cell.

\medskip
Finally, start and accept cells are given by
\[
  \bot_{\MN(L)} = \{\eqcl{\id_U}\}_{U\in \sq},
  \qquad
  \top_{\MN(L)}=\{\eqcl{P}\mid P\in L\}.
\]
The cells $\eqcl{P}$ will be called \emph{regular}.
They are $\Lapprox$-equivalence classes of ipomsets,
lower face maps unstart events, and upper face maps terminate events.
All faces of subsidiary cells $w_U$ are subsidiary,
and upper faces of regular cells are regular.
Below we present several examples,
in which we show only the essential part $\MN(L)^\ess$ of $\MN(L)$.

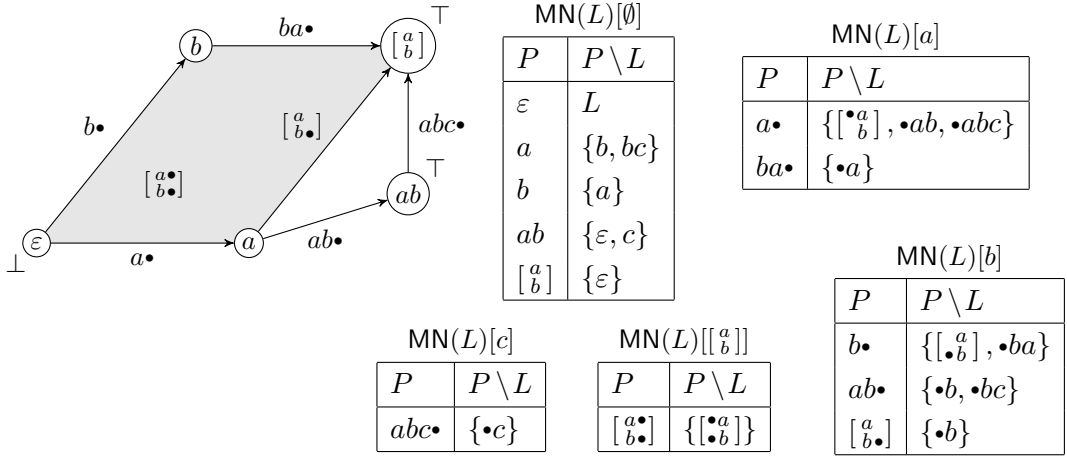
\begin{figure}[tbp]
  \centering
  \begin{tikzpicture}[x=1.4cm, y=1.3cm]
    \filldraw[color=black!10!white] (0,0)--(2,0)--(3.5,2)--(1.5,2)--(0,0);
    \node[state] (eps) at (0,0) {$\epsilon$};
    \node[below left] at (eps) {$\bot$};
    \node[state] (a) at (2,0) {$a$};
    \node[state] (b) at (1.5,2) {$b$};
    \node[state] (ba) at (3.5,2) {$\loset{a\\b}$};
    \node[above right=1.5mm] at (ba) {$\top$};
    \node[state] (ab) at (3.5,.5) {$a b$};
    \node[above right=0.9mm] at (ab) {$\top$};
    \path (eps) edge node[swap] {$a\ibullet$} (a);
    \path (eps) edge node {$b\ibullet$} (b);
    \path (a) edge node {$\loset{a\\b&\!\!\ibullet\!}$\!\!} (ba);
    \path (b) edge node {$b a\ibullet$} (ba);
    \path (a) edge node[below] {$a b\ibullet$} (ab);
    \path (ab) edge node[swap] {$a b c\ibullet$} (ba);
    \node at (1.2,.6) {$\loset{a&\!\!\ibullet\!\\b&\!\!\ibullet\!}$};
    \begin{scope}[shift={(5,0)}]
      \node at (0.2,2.3) {$\MN(L)[\emptyset]$};
      \node[below] at (0.2,2.2) {$%
        \normalsize%
        \begin{array}[t]{|l|l|}
          \hline
          P & P\lquo L \\\hline
          \epsilon & L \\
          a & \{b, b c\} \\
          b & \{a\} \\
          a b & \{\epsilon, c\} \\
          \loset{a\\b} & \{\epsilon\}\\
          \hline
        \end{array}
        $};
      \node at (3,2.1) {$\MN(L)[a]$};
      \node[below] at (3,2) {$%
        \normalsize%
        \begin{array}[t]{|l|l|}
          \hline
          P & P\lquo L \\\hline
          a\ibullet & \{\loset{\!\ibullet\!\!&a\\&b}, \ibullet a b, \ibullet a b c\} \\
          b a\ibullet & \{\ibullet a\}\\
          \hline
        \end{array}
        $};
      \node at (3.6,-0.15) {$\MN(L)[b]$};
      \node[below] at (3.6,-0.25) {$%
        \normalsize%
        \begin{array}[t]{|l|l|} \hline
          P & P\lquo L \\\hline
          b\ibullet & \{\loset{&a\\\!\ibullet\!\!&b}, \ibullet b a\} \\
          a b\ibullet & \{\ibullet b, \ibullet b c\} \\
          \loset{a\\b&\!\!\ibullet\!} & \{\ibullet b\}\\
          \hline
        \end{array}
        $};
      \node at (-1,-1.0) {$\MN(L)[c]$};
      \node[below] at (-1,-1.1) {$%
        \normalsize%
        \begin{array}[t]{|l|l|}\hline
          P & P\lquo L \\\hline
          a b c\ibullet & \{\ibullet c\}\\ \hline
        \end{array}
        $};
      \node at (1.1,-1.0) {$\MN(L)[\loset{a\\b}]$};
      \node[below] at (1.1,-1.1) {$%
        \normalsize%
        \begin{array}[t]{|l|l|}\hline
          P & P\lquo L \\\hline
          \loset{a&\!\!\ibullet\!\\b&\!\!\ibullet\!} & \{\loset{\!\ibullet\!\!&a\\\!\ibullet\!\!&b}\}\\ \hline
        \end{array}
        $};
    \end{scope}
  \end{tikzpicture}
  \caption{HDA $\MN(L)$ of Example \ref{ex:nondet}, showing names of cells instead of labels
    (labels are target interfaces of names).
    Tables show essential cells together with prefix quotients.}
  \label{fig:ex.nondet}\vspace*{-2mm}
\end{figure}

\begin{example}
  \label{ex:nondet}
  Let $L=\{\loset{a\\b}, a b c\}\down=\{\loset{a\\b}, a b, b a, a b c\}$.
  Figure \ref{fig:ex.nondet} shows the HDA $\MN(L)^\ess$
  together with a list of essential cells of $M(L)$
  and their prefix quotients in $L$.
  Note that the state $\eqcl{a}$ has \emph{two} outgoing $b$-labelled edges:
  $\eqcl{a b\ibullet}$ and $\eqcl{\loset{a\\b&\!\!\ibullet\!}}$.
  The generating ipomsets have different prefix quotients
  because of $\{\loset{a\\b}, a b c\}\subseteq L$, but the same lower face $\eqcl{a}$.
  (Note that $\eqcl{b a\ibullet}=\eqcl{\loset{a\ibullet\\b\phantom{\ibullet}}}$.)

\medskip
  Intuitively, $\MN(L)^\ess$ is thus \emph{non-deterministic};
  this is interesting because the standard Myhill-Nerode theorem
  for finite automata constructs deterministic automata.
  We will give a precise definition of determinism for HDAs in the next section
  and show in Example \ref{ex:NonDetLanguage} that no deterministic HDA $X$ exists with $\Lang(X)=L$.
\end{example}

\begin{example}
  \label{ex:strongeq}
  Here we explain why we need to use $\Lapprox$-equivalence classes
  and not $\Lsim$-equivalence classes.
  The example is one-dimensional, which means that it applies to standard finite automata.
  The reason one does not see the problem in the standard Myhill-Nerode construction
  for finite automata is that this operates only on states and not on transitions.

\medskip
  Let $L=\{aa, ab, ba\}$, then $\MN(L)^\ess$ is as below.
  
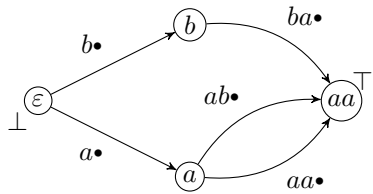
\begin{equation*}
    \begin{tikzpicture}[y=1cm]
      \node[state] (eps) at (0,0) {$\epsilon$};
      \node[below left] at (eps) {$\bot$};
      \node[state] (a) at (2,-1) {$a$};
      \node[state] (b) at (2,1) {$b$};
      \node[state] (aa) at (4,0) {$aa$};
      \node[above right=.1] at (aa) {$\top$};
      \path (eps) edge node[swap] {$a\ibullet$} (a);
      \path (eps) edge node {$b\ibullet$} (b);
      \path (a) edge[bend right] node[swap] {$aa\ibullet$} (aa);
      \path (a) edge[bend left] node {$ab\ibullet$} (aa);
      \path (b) edge[bend left] node {$ba\ibullet$} (aa);
    \end{tikzpicture}
  \end{equation*}
  We have $aa\ibullet\lquo L=ba\ibullet\lquo L=\{\ibullet a\}$,
  thus $aa\ibullet\Lsim ba\ibullet$.
  Yet $aa\ibullet$ and $ba\ibullet$ are not strongly equivalent,
  because $a\lquo L=\{a,b\}\ne \{a\}=b\lquo L$.
  This provides an example of weakly equivalent ipomsets whose lower faces are not weakly equivalent
  and shows why we cannot use $\Lsim$ to construct $\MN(L)$.
\end{example}

\begin{remark}
  \label{re:MNonedim}
  As the previous example indicates, if $L$ is one-dimensional and all words in $L$ have empty interfaces,
  then $\ess(\MN(L))$ is the standard Myhill-Nerode finite automaton for $L$.
\end{remark}

\begin{example}
  \label{ex:aa}
  The language $L=\{\loset{\!\ibullet\!\!& aa&\!\!\ibullet\!\\ \!\ibullet\!\!& a &\!\!\ibullet\! }\}$
  is recognised by the HDA  $\MN(L)^\ess$ below:
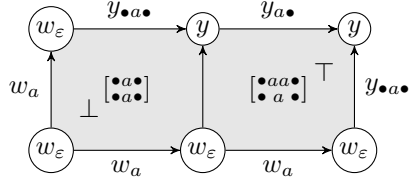
  \[
    \begin{tikzpicture}[x=1cm, y=.8cm]
      \filldraw[color=black!10!white] (0,0)--(4,0)--(4,2)--(0,2)--(0,0);
      \node[state] (eps) at (0,0) {$w_\epsilon$};
      \node[state] (a) at (2,0) {$w_\epsilon$};
      \node[state] (b) at (0,2) {$w_\epsilon$};
      \node[state] (ab) at (2,2) {$y$};
      \node[state] (aa) at (4,0) {$w_\epsilon$};
      \node[state] (aab) at (4,2) {$y$};
      \path (eps) edge node[swap] {$w_a$} (a);
      \path (eps) edge node {$w_a$} (b);
      \path (a) edge node[swap] {$w_a$} (aa);
      \path (b) edge node {$y_{\ibullet a\ibullet}$} (ab);
      \path (a) edge(ab);
      \path (ab) edge node {$y_{a\ibullet}$} (aab);
      \path (aa) edge node[swap] {$y_{\ibullet a\ibullet}$} (aab);
      \node at (1,1) {${\loset{ \!\ibullet\!\!& a&\!\!\ibullet\!  \\ \!\ibullet\!\!& a &\!\!\ibullet\!}}$};
      \node at (3,1) {${\loset{ \!\ibullet\!\!& aa&\!\!\ibullet\!  \\ \!\ibullet\!\!& a &\!\!\ibullet\!}}$};
      \node at (0.5,0.7) {$\bot$};
      \node at (3.6,1.3) {$\top$};
    \end{tikzpicture}
  \]
  Cells with the same names are identified.
  Here we see subsidiary cells $w_\epsilon$ and $w_a$,
  and regular cells that are not coaccessible (denoted by $y$ indexed with their signature).
  The middle vertical edge is $\eqcl{\loset{\!\ibullet\!\!& a \\ \!\ibullet\!\!& a&\!\!\ibullet\! }}$,
  $y_{\ibullet a\ibullet}
  =\eqcl{\loset{\!\ibullet\!\!& a&\!\!\ibullet\!\\ \!\ibullet\!\!& a}}
  =\eqcl{\loset{\!\ibullet\!\!& aa\\ \!\ibullet\!\!& a&\!\!\ibullet\! }}$,
  $y_{a\ibullet}=\eqcl{\loset{\!\ibullet\!\!& aa&\!\!\ibullet\!\\ \!\ibullet\!\!& a}}$,
  and $y=\eqcl{\loset{\!\ibullet\!\!& a\\ \!\ibullet\!\!& a}}=\eqcl{\loset{\!\ibullet\!\!& aa\\ \!\ibullet\!\!& a}}$.
\end{example}

\subsection{$\MN(L)$ is well-defined}

We need to show that $\MN(L)$ is well-defined,
\ie~that the formulas \eqref{e:FaceMaps}
do not depend on the choice of a representative in $\eqcl{P}$
and that the precubical identities are satisfied.

\begin{lemma}
  \label{l:ExtQ}
  Let $P$, $Q$ and $R$ be ipomsets with $T_P=T_Q=S_R$.  Then
  \[
    P\lquo L \subseteq Q\lquo L \implies (P R)\lquo L
    \subseteq (Q R)\lquo L.
  \]
  In particular,
  $P\lquo L = Q\lquo L$ implies $(P R)\lquo L = (Q R)\lquo L$.
\end{lemma}

\begin{proof}
  For $N\in \iiPoms$ we have
  \begin{multline*}
    N\in (P R)\lquo L \iff P R N\in L \iff
    R N\in P\lquo L \\
    \implies R N\in Q\lquo L \iff Q R N\in L \iff N\in
    (Q R)\lquo L.\quad  \qed
  \end{multline*}

  \vspace*{-7mm}
\end{proof}

The next lemma shows an operation to ``add order'' to an ipomset $P$.
This is done by first removing some points $A\subseteq T_P$
and then adding them back in, forcing arrows from all other points in $P$.
The result is obviously subsumed by~$P$.

\begin{lemma}
  \label{l:MinExt}
  For $P\in\iiPoms$ and $A\subseteq \rfin(P)$, $(P-A)*\starter{T_P}{A}\subsu P$. \qed
\end{lemma}

The next two lemmas, whose proofs are again obvious,
state that events may be unstarted or terminated in any order.

\begin{lemma}
  \label{l:TerminatorComp}
  Let $U$ be a conclist and $A, B\subseteq U$ disjoint subsets.
  Then
  \[
    \terminator{U}{B}*\terminator{(U-B)}{A}=\terminator{U}{A\cup B}=\terminator{U}{A}*\terminator{(U-A)}{B}. \quad\qed
  \]
\end{lemma}

\begin{lemma}
  \label{l:Comm}
  Let $P\in \iiPoms$ and $A, B\subseteq T_P$ disjoint subsets.
  Then
  \begin{equation*}
    (P*\terminator{T_P}{B})-A
    =
    (P-A)*\terminator{(T_P-A)}{B}. \quad\qed
  \end{equation*}
\end{lemma}

\begin{lemma}
  \label{l:Approx}
  Assume that $P\Lapprox Q$ for $P,Q\in\iiPoms_U$.
  Then $P*\terminator{U}{B}\Lapprox Q*\terminator{U}{B}$ for every $B\subseteq U$.
\end{lemma}

\begin{proof}
  Obviously $\fin(P*\terminator{U}{B})=\fin(P)-B\cong \fin(Q)-B=\fin(Q*\terminator{U}{B})$.
  For every $A\subseteq \rfin(P)-B\simeq \rfin(Q)-B$ we have
  \[
    ((P-A)*\terminator{(U-A)}{B})\lquo L
    =
    ((Q-A)*\terminator{(U-A)}{B})\lquo L
  \]
  by assumption and Lemma \ref{l:ExtQ}.
  But
  $(P*\terminator{U}{B})-A=(P-A)*\terminator{(U-A)}{B}$
  and
  $(Q*\terminator{U}{B})-A=(Q-A)*\terminator{(U-A)}{B}$
  by Lemma \ref{l:Comm}. \qed
\end{proof}

\begin{proposition}
  $\MN(L)$ is a well-defined HDA.
\end{proposition}

\begin{proof}
  The face maps are well-defined:
  for $\delta^0_A$ this follows from
  Lemma \ref{l:StrongEqDef},
  for $\delta^1_B$ from Lemma \ref{l:Approx}.
  The precubical identities $\delta^\nu_A \delta^\mu_B=\delta^\mu_B\delta^\nu_A$
  are clear for $\nu=\mu=0$,
  follow from Lemma \ref{l:TerminatorComp} for $\nu=\mu=1$,
  and from Lemma \ref{l:Comm} for $\{\nu,\mu\}=\{0,1\}$.\! \qed
\end{proof}

\subsection{Paths and essential cells of $\MN(L)$}

The next lemma provides paths in $\MN(L)$.

\begin{lemma}
  \label{l:PathDetail}
  For every $N,P\in \iiPoms$ such that $T_N\cong S_P$
  there exists a path $\alpha\in \Path(\MN(L))_{\eqcl{N}}^{\eqcl{N P}}$
  such that $\ev(\alpha)=P$.
\end{lemma}

\begin{proof}
  Choose a decomposition
  $P=Q_1*\dotsm*Q_n$ into starters and terminators.
  Denote $U_k=T_{Q_k}=S_{Q_{k+1}}$
  and define
  \begin{equation*}
    x_k=\eqcl{N*Q_1*\dotsm*Q_k}, \qquad
    \phi_k=
    \begin{cases}
      \arrO{A} & \text{if $Q_k=\starter{U_k}{A}$}, \\
      \arrI{B} & \text{if $Q_k=\terminator{U_{k-1}}{B}$}
    \end{cases}
  \end{equation*}
  for $k=1,\dotsc,n$.
  If $\phi_k=\arrO{A}$ and $Q_k=\starter{U_k}{A}$, then
  \begin{equation*}
    \delta^0_A(x_k)
    =
    \eqcl{N*Q_1*\dotsm*Q_{k-1}*\starter{U_k}{A}-A} \\
    =
    \eqcl{N*Q_1*\dotsm*Q_{k-1}*\id_{U_k-A}}
    =
    x_{k-1}.
  \end{equation*}
  If $\phi_k=\arrI{B}$ and $Q_k=\terminator{U_{k-1}}{B}$, then
  \[
    \delta^1_B(x_{k-1})
    =
    \eqcl{N*Q_1*\dotsm*Q_{k-1}*\terminator{U_{k-1}}{B}}
    =
    x_{k}.
  \]
  Thus, $\alpha=(x_0,\phi_1,x_1,\dotsc,\phi_n,x_n)$
  is a path
  with $\ev(\alpha)=P$, $\src(\alpha)=\eqcl{N}$ and $\tgt(\alpha)=\eqcl{N*P}$. \qed
\end{proof}

Our goal is now to describe essential cells of $\MN(L)$.

\begin{lemma}
  \label{l:AccPos}
  All regular cells of $\MN(L)$ are accessible.
  If $P\lquo L\ne \emptyset$, then $\eqcl{P}$ is coaccessible.
\end{lemma}

\begin{proof}
  Both claims follow from Lemma \ref{l:PathDetail}.
  For every $P$ there exists a path
  from $\eqcl{\id_{S_P}}$ to $\eqcl{\id_{S_P}*P}=\eqcl{P}$.
  If $Q\in P\lquo L$,
  then there exists a path $\alpha\in \Path(\MN(L))_{\eqcl{P}}^{\eqcl{P Q}}$,
  and $P Q\in L$ entails that $\eqcl{P Q}\in \top_{\MN(L)}$. \qed
\end{proof}

\begin{lemma}
  \label{l:AccNeg}
  Subsidiary cells of $\MN(L)$ are not accessible.
  If $P\lquo L=\emptyset$,
  then the cell $\eqcl{P}$ is not coaccessible.	
\end{lemma}

\begin{proof}
  If $\alpha\in\Path(\MN(L))_\bot^{w_U}$,
  then it contains a step $\beta$ from a regular cell to a subsidiary cell
  (since all start cells are regular).
  Yet $\beta$ can be neither an upstep (since lower faces of subsidiary cells are subsidiary)
  nor a downstep (since upper faces of regular cells are regular).
  This contradiction proves the first claim.

\medskip
  To prove the second part we use a similar argument.
  If $P\lquo L=\emptyset$, then a path $\alpha\in\Path(\MN(L))_{\eqcl{P}}^\top$
  contains only regular cells (as shown above).
  Given that $R\lquo L\ne\emptyset$ for all $\eqcl{R}\in \top_{\MN(L)}$,
  $\alpha$ must contain a step $\beta$ from $\eqcl{Q}$ to $\eqcl{R}$
  such that $Q\lquo L=\emptyset$ and $R\lquo L\ne\emptyset$.
  If $\beta$ is a downstep, \ie $\beta=(\eqcl{Q}\arrI{A} \eqcl{Q*\terminator{U}{A}})$,
  and $N\in R\lquo L=(Q*\terminator{U}{A})\lquo L$,
  then $\terminator{U}{A}*N\in Q\lquo L\ne\emptyset$: a contradiction.	
  If $\beta=(\eqcl{R-A}\arrO{A} \eqcl{R})$ is an upstep and $N\in R\lquo L$,
  then, by Lemma \ref{l:MinExt},
  \[
    (R-A)*\starter{U}{A}*N\subsu
    R*N
    \in L,
  \]
  implying that $Q\lquo L=(R-A)\lquo L\ne \emptyset$ by Lemma \ref{l:QuotIncl}:
  another contradiction. \qed
\end{proof}

Lemmas \ref{l:AccPos} and \ref{l:AccNeg} together immediately imply the following.

\begin{proposition}
  \label{p:Acc}
  $\ess(\MN(L))=\{\eqcl{P}\mid P\lquo L\ne \emptyset\}.$ \qed
\end{proposition}

\subsection{$\MN(L)$ recognises $L$}

We are finally ready to show that $\Lang(\MN(L))=L$.
One inclusion follows directly from Lemma \ref{l:PathDetail}:

\begin{lemma}
  \label{l:Path}
  $L\subseteq\Lang(\MN(L))$.
\end{lemma}

\begin{proof}
  For every $P\in \iiPoms$
  there exists a path $\alpha\in\Path(\MN(L))_{\eqcl{\id_{S_P}}}^{\eqcl{P}}$
  such that $\ev(\alpha)=P$.
  If $P\in L$, then $\epsilon\in P\lquo L$,
  \ie $\eqcl{P}$ is an accept cell.
  Thus $\alpha$ is accepting and
  $P=\ev(\alpha)\in \Lang(\MN(L))$. \qed
\end{proof}

The converse inclusion requires more work.
For a regular cell $\eqcl{P}$ of $\MN(L)$ denote $\eqcl{P}\lquo L=P\lquo L$
(this obviously does not depend on the choice of $P$).

\begin{lemma}
\label{l:Limit}
	If $S\in\sq$ and $\alpha\in\Path(\MN(L))_{\eqcl{\id_S}}$,
	then $\tgt(\alpha)\lquo L\subseteq \ev(\alpha)\lquo L$.	
\end{lemma}

\begin{proof}
  By Lemma \ref{l:AccNeg},
  all cells appearing along $\alpha$ are regular.
  We proceed by induction on the length of $\alpha$.
  For $\alpha=(\eqcl{\id_S})$ the claim is obvious.
  If $\alpha$ is non-trivial, we have two cases.
  \begin{itemize}
  \item
    $\alpha=\beta*(\delta^0_A(\eqcl{P})\arrO{A} \eqcl{P})$,
    where $\eqcl{P}\in\MN(L)[U]$ and $A\subseteq \rfin(P)\subseteq U\cong T_P$.
    By the induction hypothesis,
    \[
      (P-A)\lquo L
      =\delta^0_A(\eqcl{P})\lquo L
      =\tgt(\beta)\lquo L
      \subseteq \ev(\beta)\lquo L.
    \]
    For $Q\in \iiPoms$ we have
    \begin{align*}
      Q\in P\lquo L
      \iff P Q\in L
      &\implies (P-A)*\starter{U}{A}*Q\in L
      \tag{Lemma \ref{l:MinExt}}\\
      &\iff \starter{U}{A}*Q \in (P-A)\lquo L\\
      &\implies \starter{U}{A}*Q \in \ev(\beta)\lquo L
      \tag{induction hypothesis}\\
      &\iff \ev(\beta)*\starter{U}{A}*Q \in L\\
      &\iff \ev(\alpha)*Q \in L
      \iff  Q \in \ev(\alpha)\lquo L.
    \end{align*}
    Thus, $\eqcl{P}\lquo L=P\lquo L\subseteq \ev(\alpha)\lquo L$.
  \item	
    $\alpha=\beta*(\eqcl{P}\arrI{B}\delta^1_B(\eqcl{P}))$,
    where $\eqcl{P}\in\MN(L)[U]$ and $B\subseteq U\cong T_P$.
    By inductive assumption,
    $
      P\lquo L=\tgt(\beta)\lquo L \subseteq \ev(\beta)\lquo L
    $.
    Thus,
                \begin{equation*}
			\tgt(\alpha)\lquo L
			=
			\delta^1_B(\eqcl{P})\lquo L
			=
			\eqcl{P*\terminator{U}{B}}\lquo L
			\subseteq
			(\ev(\beta)*\terminator{U}{B})\lquo L
			=
			\ev(\alpha)\lquo L.
                \end{equation*}
	\end{itemize}
	The inclusion above follows from Lemma \ref{l:ExtQ}. \qed
\end{proof}

\begin{proposition}
  \label{p:MNLang}
  $\Lang(\MN(L))=L$.
\end{proposition}

\begin{proof}
  The inclusion $L\subseteq \Lang(\MN(L))$ is shown in Lemma \ref{l:Path}.
  For the converse, let $S\in\sq$ and $\alpha\in\Path(\MN(L))_{\eqcl{\id_S}}$,
  then Lemma \ref{l:Limit} implies
  \[
    \tgt(\alpha)\in \top_{\MN(L)}
    \iff
    \epsilon\in \tgt(\alpha)\lquo L
    \implies
    \epsilon\in \ev(\alpha)\lquo L
    \iff
    \ev(\alpha)\in L,
  \]
  that is, if $\alpha$ is accepting, then $\ev(\alpha)\in L$. \qed
\end{proof}

\subsection{Finiteness of $\MN(L)$}

The HDA $\MN(L)$ is not finite,
since it contains infinitely many subsidiary cells $w_U$.
Below we show that its essential part $\MN(L)^\ess$ is finite
if $L$ has finitely many prefix quotients.

\begin{lemma}
  \label{l:FiniteEss}
  If $\suff(L)$ is finite,
  then $\ess(\MN(L))$ is finite.
\end{lemma}

\begin{proof}
  For $\eqcl{P},\eqcl{Q}\in\ess(L)$,
  we have $\eqcl{P}=\eqcl{Q}\iff f(\eqcl{P})=f(\eqcl{Q})$,
  where
  \[
    f(\eqcl{P})
    =
    (
    P\lquo L,
    \fin(P),
    ((P-A)\lquo L)_{A\subseteq \rfin(P)}
    ).
  \]
  We will show that $f$ takes only finitely many values on $\ess(L)$.
  Indeed, $P\lquo L$ belongs to the finite set $\suff(L)$.
  Further, all ipomsets in $P\lquo L$ have source interfaces equal to $T_P$.
  Since $P\lquo L$ is non-empty, $\fin(P)$ is a starter with $T_P$ as underlying conclist.
  Yet, there are only finitely many starters on any conclist.
  The last coordinate also may take only finitely many values,
  since $\rfin(P)$ is finite and $(P-A)\lquo L\in\suff(L)$. \qed
\end{proof}

\begin{varproof}[of Theorem~\ref{t:MN}, {\rm (b)$\implies$(a)}]	
  By Lemma \ref{l:FiniteEss} and Lemma \ref{l:HDAGen},
  $\MN(L)^\ess$ is a finite HDA.
  With Proposition~\ref{p:MNLang}, $\Lang(\MN(L)^\ess)=\Lang(\MN(L))=L$. \qed
\end{varproof}

\begin{example}
  \label{exa:hda-loop}
  We finish this section with another example,
  which shows some subtleties related to higher-dimensional loops.
  Let $L$ be the language of the HDA shown to the left of Figure~\ref{fig:hda-loop}
  (a looping version of the HDA of Figure~\ref{fig:ab*cb*cd}),
  then
  \begin{equation*}
    L = \{\ibullet a\ibullet\} \cup \{\loset{\!\ibullet\!\!& aa&\!\!\ibullet\! \\ &b}^n\mid n\ge 1\}\down.
  \end{equation*}
  Our construction yields $\MN(L)^\ess$ as shown on the right of the figure.
  Here, $e=\eqcl{\loset{\!\ibullet\!\!&a \\ &b&\!\!\ibullet\!}}$, and
  the two $e$-labelled edges and their corresponding faces are identified.
  These identifications follow from the fact that
  $\loset{ \!\ibullet\!\!& aa \\ &bb&\!\!\ibullet\!}
  \Lapprox \loset{\!\ibullet\!\!&a \\ &b&\!\!\ibullet\!}$,
  $\loset{ \!\ibullet\!\!& aa \\ &bb }
  \Lapprox \loset{\!\ibullet\!\!&a \\ &b}$, and
  $\loset{ \!\ibullet\!\!& aa \\ &b }
  \Lapprox \ibullet a$.
  Note that $\loset{ \!\ibullet\!\!&a&\!\!\ibullet\! \\ &b&\!\!\ibullet\!}$
  and $\loset{ \!\ibullet\!\!& aa&\!\!\ibullet\! \\ &bb&\!\!\ibullet\!}$
  are not strongly equivalent, since they have different signatures:
  $\loset{ \!\ibullet\!\!& a&\!\!\ibullet\! \\ &b&\!\!\ibullet\!}$
  and $\loset{ a&\!\!\ibullet\! \\ b&\!\!\ibullet\!}$, respectively.
\end{example}

\begin{figure}[h]
  \centering
  \begin{tikzpicture}[x=1.2cm, y=.91cm]
    \filldraw[color=black!10!white] (0,0)--(4,0)--(4,4)--(2,4)--(2,2)--(0,2)--(0,0);
    \node[state] (eps) at (0,0) {$w_\epsilon$};
    \node[state] (a) at (2,0) {${\ibullet a}$};
    \node[state] (b) at (0,2) {$w_\epsilon$};
    \node[state] (ab) at (2,2) {${\loset{ \!\ibullet\!\!& a \\ & b }}$};
    \node[state] (aa) at (4,0) {${\ibullet aa}$};
    \node[state] (aab) at (4,2) {${\ibullet a}$};
    \node[state] (abb) at (2,4) {${\loset{\!\ibullet\!\!& a\\ &b} b}$};
    \node[state] (aabb) at (4,4) {${\loset{ \!\ibullet\!\!& a \\ &b }}$};
    \path (eps) edge node[swap] {${\ibullet a\ibullet}$} (a);
    \node[above] at (1,0) {$\bot\top$};
    \path (eps) edge node {$w_a$} (b);
    \path (a) edge node[swap] {${\ibullet aa\ibullet}$} (aa);
    \path (b) edge node {${\loset{ \!\ibullet\!\!& a&\!\!\ibullet\! \\ &b }}$} (ab);
    \path (a) edge node[swap] {$e$} (ab);
    \path (2,.4) edge[color=blue,-, very thick] (2,1.2);
    \path (ab) edge node {${\loset{ \!\ibullet\!\!& aa&\!\!\ibullet\! \\ &b }}$} (aab);
    \path (abb) edge node {${\loset{ \!\ibullet\!\!& aa&\!\!\ibullet\! \\ &bb }}$} (aabb);
    \node[below] at (3,2) {$\top$};
    \path (aa) edge node[swap] {${\loset{ \!\ibullet\!\!& aa&\!\!\ibullet\! \\ &b&\!\!\ibullet\!}}$} (aab);
    \path (aab) edge node[swap] {$e$} (aabb);
    \path (4,2.4) edge[color=blue,-, very thick] (4,3.2);
    \path (ab) edge node {$ {\loset{\!\ibullet\!\!& a \\ &b} b\ibullet}$} (abb);
    \node at (1,1) {${\loset{ \!\ibullet\!\!& a&\!\!\ibullet\! \\ &b&\!\!\ibullet\!}}$};
    \node at (3,1) {${\loset{ \!\ibullet\!\!& aa&\!\!\ibullet\! \\ &b&\!\!\ibullet\!}}$};
    \node at (3,3) {${\loset{ \!\ibullet\!\!& aa&\!\!\ibullet\! \\ &bb&\!\!\ibullet\!}}$};
    \begin{scope}[x=1cm, y=.8cm, shift={(-6.3,2.0)}]
      \filldraw[color=black!10!white] (0,0)--(4,0)--(4,2)--(0,2)--(0,0);			
      \filldraw (0,0) circle (0.05);
      \filldraw (2,0) circle (0.05);
      \filldraw (4,0) circle (0.05);
      \filldraw (0,2) circle (0.05);
      \filldraw (2,2) circle (0.05);
      \filldraw (4,2) circle (0.05);
      \path (0,0) edge (1.95,0);
      \path (.5,0) edge[color=blue,-, very thick] (1.5,0);
      \path (2,0) edge node[below] {$a$} (3.95,0);
      \path (0,2) edge node[above] {$a$} (1.95,2);
      \path (2,2) edge (3.95,2);
      \path (2.5,2) edge[-,color=blue,very thick] (3.5,2);
      \path (0,0) edge node[left] {$b$} (0,1.95);
      \draw (2,0) edge (2,1.95);
      \path (4,0) edge node[right] {$b$} (4,1.95);
      \node[below] at (1,0) {$\bot\top$};
      \node[above] at (3,2) {$\bot\top$};
    \end{scope}
  \end{tikzpicture}
  \caption{Two HDAs recognising the language of Example \ref{exa:hda-loop}.
    On the left side, start/accept edges are identified;
    on the right, $e$-labelled edges are identified.}
  \label{fig:hda-loop}\vspace*{-2mm}
\end{figure}

\section{Determinism}
\label{se:det}

We now make precise our notion of determinism and show that not all HDAs may be determinised.
Recall that we do not assume finiteness.

\begin{definition}
  An HDA $X$ is \emph{deterministic} if
  \begin{enumerate}
  \item
    for every $U\in \sq$ there is at most one initial cell in $X[U]$, and
  \item
    for all $V\in \sq$, $A\subseteq V$ and any essential cell $x\in X[V-A]$
    there exists at most one essential cell $y\in X[V]$
    such that $x=\delta^0_A(y)$.
  \end{enumerate}
\end{definition}

That is, in any essential cell $x$ in a deterministic HDA $X$ and for any set $A$ of events,
there is at most one way to start $A$ in $x$ and remain in the essential part of $X$
(recall that termination of events is always deterministic).
We allow multiple initial cells because ipomsets in $\Lang(X)$ may have different source interfaces;
for each source interface in $\Lang(X)$, there can be at most one matching start cell in $X$.
Note that we must restrict our definition to essential cells
as inessential cells may not always be removed
(in contrast to the case of standard automata).

A language is deterministic
if it is recognised by a deterministic HDA.
We develop a language-internal criterion for being deterministic.

\begin{definition}
  A language $L$ is \emph{swap-invariant}
  if it holds for all $P, Q, P', Q'\in\iiPoms$
  that $P P'\in L$, $Q Q'\in L$ and $P\subsu Q$ imply $Q P'\in L$.
\end{definition}

That is, if the $P$ prefix of $P P'\in L$
is subsumed by $Q$ (which is, thus, ``more concurrent'' than $P$),
and if $Q$ itself may be extended to an ipomset in $L$,
then $P$ may be swapped for $Q$ in the ipomset $P P'$ to yield $Q P'\in L$.

\begin{lemma}
  \label{le:swapi}
  $L$ is swap-invariant if and only if $P\subsu Q$ implies
  $P\lquo L = Q\lquo L$ for all $P, Q\in \iiPoms$,
  unless $Q\lquo L=\emptyset$.
\end{lemma}

\begin{proof}
  Assume that $L$ is swap-invariant and let
  $P\subsu Q$.
  The inclusion $Q\lquo L\subseteq P\lquo L$ follows from Lemma \ref{l:QuotIncl},
  and
  \[
    R\in Q\lquo L,\; R'\in P\lquo L
    \iff
    QR, PR'\in L
    \implies
    QR'\in L
    \iff
    R'\in Q\lquo L
  \]
  implies that $P\lquo L\subseteq Q\lquo L$.
  The calculation
  \begin{equation*}
    PP',QQ'\in L,\; P\subsu Q
    \iff
    P'\in P\lquo L,\; Q'\in Q\lquo L,\; P\subsu Q
    \implies
    P'\in Q\lquo L
    \iff
    QP'\in L
  \end{equation*}
  shows the converse.
\end{proof}

Our main goal is to show the following criterion,
which will be implied by Propositions~\ref{l:DetImpliesSInv} and \ref{l:SInvImpliesDet} below.

\begin{theorem}
  \label{t:swapdet}
  A language $L$ is deterministic if and only if it is swap-invariant.
\end{theorem}

\begin{example}
\label{ex:NonDetLanguage}
  The regular language $L=\{\loset{a\\b}, a b, b a, a b c\}$ from Example \ref{ex:nondet}
  is not swap-invariant:
  using Lemma \ref{le:swapi},
  $ab\ibullet\subsu\loset{a\\b&\!\!\ibullet\!}$, but
  $\{ab\ibullet\}\lquo L
  =
  \{\ibullet b, \ibullet bc\}
  \ne
  \{\ibullet b\}
  =
  \{\loset{a\\b&\!\!\ibullet\!}\}\lquo L$.
  Hence $L$ is not deterministic.
\end{example}

The next examples explain why we need to restrict to essential cells
in the definition of deterministic HDAs.

\begin{example}
  The HDA in Example \ref{ex:aa} is deterministic.
  There are two different $a$-labelled edges starting at $w_\epsilon$
  ($w_a$ and $\eqcl{\loset{\!\ibullet\!\!& a \\ \!\ibullet\!\!& a&\!\!\ibullet\! }}$),
  yet it does not disturb determinism since $w_\epsilon$ is not accessible.
\end{example}

\begin{example}
  Let $L=\{ab, \loset{a&\!\!\ibullet\!\\ b&\!\!\ibullet\!}\}$.
  Then $\MN(L)^{\ess}$ is as follows:
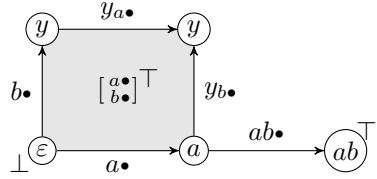
\vspace*{-1mm}
  \[
    \begin{tikzpicture}[x=1cm, y=.8cm]
      \filldraw[color=black!10!white] (0,0)--(2,0)--(2,2)--(0,2)--(0,0);
      \node[state] (eps) at (0,0) {$\epsilon$};
      \node at (-.3,-.2) {$\bot$};
      \node[state] (a) at (2,0) {$a$};
      \node[state] (b) at (0,2) {$y$};
      \node[state] (ab) at (2,2) {$y$};
      \node[state] (aa) at (4,0) {$ab$};
      \node[above right=0.1] at (aa) {$\top$};
      \path (eps) edge node[swap] {$a\ibullet$} (a);
      \path (eps) edge node {$b\ibullet$} (b);
      \path (a) edge node {$ab\ibullet$} (aa);
      \path (b) edge node {$y_{a\ibullet}$} (ab);
      \path (a) edge node[swap] {$y_{b\ibullet}$} (ab);
      \node (m) at (1,1) {${\loset{ a&\!\!\ibullet\! \\ b&\!\!\ibullet\!}}$};
      \node at (1.4,1.2) {$\top$};
    \end{tikzpicture}
  \]	
  It is deterministic: there are two $b$-labelled edges leaving $a$,
  namely $y_{b\ibullet}$ and $ab\ibullet$,
  but only the latter is coaccessible.
\end{example}

The next lemma shows that up to path equivalence,
paths on deterministic HDAs are determined by their labels.
That is, up to path equivalence, deterministic HDAs are \emph{unambiguous}.
Note that \cite{DBLP:conf/ictac/AmraneBFZ23} shows that
non-deterministic HDAs may exhibit unbounded ambiguity.

\begin{lemma}
  \label{l:DetEquiPaths}
  Let $X$ be a deterministic HDA and $\alpha,\beta\in\Path(X)_\bot$
  with $\tgt(\alpha),$ $\tgt(\beta)\in\ess(X)$.
  If $\ev(\alpha)=\ev(\beta)$,
  then $\alpha\simeq \beta$.
\end{lemma}

\begin{proof}
  We can assume that $\alpha=\alpha_1*\dotsm*\alpha_n$ and $\beta=\beta_1*\dotsm*\beta_m$ are sparse;
  note that all of these cells are essential.
  We show that $\alpha_k=\beta_k$ for all $k$ which implies the claim.

\medskip
  Denote $P=\ev(\alpha)=\ev(\beta)$,
  then
  \[
    P=\ev(\alpha_1)*\dotsm*\ev(\alpha_n)
  \]
  is a sparse step decomposition of $P$.
  Similarly, $P=\ev(\beta_1)*\dotsm*\ev(\beta_m)$ is a sparse step decomposition.
  Yet sparse step decompositions are unique by Proposition~\ref{p:SparsePresentation};
  hence, $m=n$ and $\ev(\alpha_k)=\ev(\beta_k)$ for every $k$.

\medskip
  We show by induction that $\alpha_k=\beta_k$.
  First, $\ev(\alpha_0)=\ev(\beta_0)$ implies $\alpha_0=\beta_0$ by determinism.
  Now assume that $\alpha_{k-1}=\beta_{k-1}$.
  Let $x=\src(\alpha_k)=\tgt(\alpha_{k-1})=\tgt(\beta_{k-1})=\src(\beta_k)$.
  If $P_k=\ev(\alpha_k)=\ev(\beta_k)$ is a terminator $\terminator{U}{B}$,
  then
  $\alpha_k = \delta_B^1(x) = \beta_k$.
  If $P_k$ is a starter $\starter{U}{A}$,
  then there are $y, z\in X$ such that $\delta^0_A(y)=\delta^0_A(z)=x$.
  As $y$ and $z$ are essential
  and $X$ is deterministic,
  this implies $y=z$ and $\alpha_k=\beta_k$. \qed
\end{proof}

\begin{lemma}
  \label{l:DetSubPaths}
  Let $\alpha$ and $\beta$ be essential paths on a deterministic HDA $X$.
  Assume that $\src(\alpha)=\src(\beta)$ and $\ev(\alpha)\subsu\ev(\beta)$.
  Then $\tgt(\alpha)=\tgt(\beta)$.
\end{lemma}

\begin{proof}
  By Lemma \ref{l:SubsuPath},
  there exists a path $\gamma\in\Path(X)_{\src(\beta)}^{\tgt(\beta)}$ such that $\ev(\gamma)=\ev(\alpha)$.
  Lemma \ref{l:DetEquiPaths} implies that $\gamma\simeq \alpha$
  and then $\tgt(\alpha)=\tgt(\gamma)=\tgt(\beta)$.
\end{proof}

\begin{proposition}
  \label{l:DetImpliesSInv}
  If $L$ is deterministic, then $L$ is swap-invariant.
\end{proposition}

\begin{proof}
  Let $X$ be a deterministic HDA that recognises $L$
  and fix ipomsets $P\subsu Q$.
  From Lemma \ref{l:QuotIncl} follows that $Q\lquo L\subseteq P\lquo L$.
  It remains to prove that if $Q\lquo L\ne\emptyset$,
  then $P\lquo L\subseteq Q\lquo L$.
  Denote $U\cong S_P\cong S_Q$.
	
  Let $R\in Q\lquo L$
  and let $\omega\in\Path(X)_{\eqcl{\id_U}}^\top$ be an accepting path that recognises $QR$.
  By Lemma \ref{l:PathDivision}, there exists a path $\beta\in\Path(X)_{\eqcl{\id_U}}$
  such that $\ev(\beta)=Q$.	
	
  Now assume that $R'\in P\lquo L$,
  and let $\omega'\in\Path(X)_{\eqcl{\id_U}}^\top$ be a path such that $\ev(\omega')=PR'$.
  By Lemma \ref{l:PathDivision},
  there exist paths $\alpha\in\Path(X)_{\eqcl{\id_U}}$ and
  $\gamma\in\Path(X)^{\tgt(\omega')}$
  such that $\tgt(\alpha)=\src(\gamma)$,
  $\ev(\alpha)=P$ and $\ev(\gamma)=R'$.
  From Lemma \ref{l:DetSubPaths} and $P\subsu Q$ follows that
  $\tgt(\alpha)=\tgt(\beta)$.
  Thus, $\beta$ and $\gamma$ may be concatenated
  to an accepting path $\beta*\gamma$.
  By $\ev(\beta*\gamma)=QR'$
  we have $QR'\in L$, \ie $R'\in Q\lquo L$. \qed
\end{proof}

\begin{lemma}
  \label{l:EssentialLowerFaces}
  If $\eqcl{P}\in\ess(\MN(L))$ and $A\subseteq \rfin(P)$, then $\eqcl{P-A}\in\ess(\MN(L))$.
\end{lemma}

\begin{proof}
  By Lemma \ref{l:AccPos},
  $\eqcl{P-A}$ is accessible.
  By assumption, $\eqcl{P}$ is coaccessible and $(\eqcl{P-A}\arrO{A}\eqcl{P})$
  is a path, so $\eqcl{P-A}$ is also coaccessible. \qed
\end{proof}

\begin{proposition}
  \label{l:SInvImpliesDet}
  If $L$ is swap-invariant, then $\MN(L)$ and $\MN(L)^\ess$ are deterministic.
\end{proposition}

\begin{proof}
  Since $\MN(L)^{\ess}$ is a sub-HDA of $\MN(L)$, it suffices to prove that $\MN(L)$ is deterministic.
  $\MN(L)$ contains only one start cell $\eqcl{\id_U}$ for every $U\in\sq$.

  Fix $U\in\sq$, $P,Q\in\iiPoms_U$ and $A\subseteq U$.
  Assume that $\delta_A^0(\eqcl{P})=\delta_A^0(\eqcl{Q})$, \ie
  $\eqcl{P-A}=\eqcl{Q-A}$,
  and $\eqcl{P},\eqcl{Q},\eqcl{P-A}\in\ess(\MN(L))$.
  We will prove that $\eqcl{P}=\eqcl{Q}$,
  or equivalently, $P\Lapprox Q$.
	
  We have $\fin(P-A)=\fin(Q-A)=:\starter{(U-A)}{S}$.
  First, notice that $A$, regarded as a subset of $P$ (or $Q$), contains no start events:
  else, we would have $\delta^0_A(\eqcl{P})=w_{U-A}$ (or $\delta^0_A(\eqcl{Q})=w_{U-A}$).
  As a consequence, $\fin(P)=\fin(Q)=\starter{U}{S}$.
	
\medskip
  For every $B\subseteq \rfin(P)=\rfin(Q)$ we have
  \begin{align*}
    P-A \Lapprox Q-A &\implies
    (P-(A\cup B))\lquo L = (Q-(A\cup B))\lquo L\\
    &\implies ((P-(A\cup B))*\starter{U}{(A-B)})\lquo L = ((Q-(A\cup B))*\starter{U}{(A-B)})\lquo L.
  \end{align*}
  The first implication follows from the definition,
  and the second from Lemma \ref{l:ExtQ}.
  From Lemma \ref{l:MinExt} follows that
  \[
    (P-(A\cup B))*\starter{U}{(A-B)}\subsu P-B,\quad
    (Q-(A\cup B))*\starter{U}{(A-B)}\subsu Q-B.
  \]
  Thus, by swap-invariance we have
  $
  (P-B)\lquo L=(Q-B)\lquo L;
  $
  note that Lemma \ref{l:EssentialLowerFaces} guarantees that neither of these languages is empty. \qed
\end{proof}

\section{Higher-dimensional automata with interfaces}
\label{se:ihda}

Higher-dimensional automata with interfaces (iHDAs) were introduced in \cite{conf/concur/FahrenbergJSZ22}
as a tool that allowed to prove a Kleene theorem for HDAs.
Both HDAs and iHDAs recognise the same class of languages,
yet, compared to iHDAs, HDAs have a flaw:
they enforce introducing non-essential cells that serve solely as faces of other cells.
We will show below that essential parts of iHDAs are again iHDAs,
a fact which allows us to give a Myhill-Nerode construction using iHDAs
which proceeds along different lines and, we believe, is more simple and principled.

\medskip
We will also provide a notion of deterministic iHDAs
which, again, is simpler in that it does not have to restrict to essential cells,
and show that the notions of deterministic languages of HDAs and iHDAs agree.

\subsection{Iprecubical sets and iHDAs}

The main difference between HDAs and iHDAs is that events in iHDAs
may be marked as \emph{source events} or \emph{target events}.
Accepting runs may never terminate target events and, similarly,
source events must have been present from the very beginning of an accepting run.

A \emph{concurrency list with interfaces} (\emph{iconclist})
$(U, {\evord}, S, T, \lambda)$ is a conclist $(U, {\evord}, \lambda)$
together with subsets $S, T\subseteq U$.
Equivalently, iconclists are iposets with empty precedence relation;
conclists are iconclists with empty interfaces.
We write $\ilo{S}{U}{T}$
for an iconclist as above.

\medskip
Let $\isq$ denote the set of iconclists.
An \emph{iprecubical set} consists of a set of cells $X$
together with a mapping $\iev: X\to \isq$.
For an iconclist $\ilo{S}{U}{T}$ we write $X[\ilo{S}{U}{T}]=\{x\in X\mid \iev(x)=\ilo{S}{U}{T}\}$.
Face maps in iprecubical sets cannot unstart events in source interfaces
neither terminate events in target interfaces.
That is, for every iconclist $\ilo{S}{U}{T}$ and subsets $A, B\subseteq U$
such that $A\cap S=B\cap T=\emptyset$ there are \emph{face maps}
\begin{equation*}
  \delta_A^0: X[\ilo{S}{U}{T}]\to X[\ilo{S}{(U-A)}{(T-A)}], \qquad
  \delta_B^1: X[\ilo{S}{U}{T}]\to X[\ilo{(S-B)}{(U-B)}{T}].
\end{equation*}
Further, for $A, B\subseteq U$ with $A\cap B=\emptyset$ and $\nu, \mu\in\{0, 1\}$,
$\delta_A^\nu \delta_B^\mu = \delta_B^\mu \delta_A^\nu$
whenever these are defined.

\medskip
A \emph{higher-dimensional automaton with interfaces} (\emph{iHDA})
is an iprecubical set $X$ together with subsets $\bot_X, \top_X\subseteq X$
of \emph{start} and \emph{accept} cells such that
for all $x\in \bot_X$ with $\iev(x)=\ilo{S}{U}{T}$, $S=U$ and
for all $x\in \top_X$ with $\iev(x)=\ilo{S}{U}{T}$, $T=U$.
That is, events in start cells are source events and cannot be unstarted,
and events in accept cells are target events and cannot be terminated.

\begin{remark}
  \label{r:ihdanoni}
  Every precubical set $X$ may be regarded as an iprecubical set $X'$
  such that $X'[\ilo \emptyset U\emptyset]=X[U]$
  and  $X'[\ilo SUT]=\emptyset$ whenever $S\ne\emptyset$ or $T\ne\emptyset$.
  If $X$ is an HDA and all its start and accept cells are vertices (elements of $X[\emptyset]$),
  then $X'$ may be regarded as an iHDA as well.
  This fails in presence of higher-dimensional start or accept cells
  due to the condition on event iconclists of such cells.
\end{remark}

\begin{example}
  \label{ex:aiHDA}
  Let $X$ be the iHDA defined by $X=\{x, e_1,e_2,e_3,e_4\}$,
  $\ev(x)=\emptyset$,
  \[
    \ev(e_1)=\ilo aaa,
    \quad
    \ev(e_2)=\ilo aa\emptyset,
    \quad
    \ev(e_3)=\ilo \emptyset a\emptyset,
    \quad
    \ev(e_4)=\ilo \emptyset aa,
  \]
  $\delta^1_a(e_2)=\delta^0_a(e_3)=\delta^1_a(e_3)=\delta^0_a(e_4)=x$,
  and $\bot_X=\{e_1,e_2\}$, $\top_X=\{e_1,e_4\}$.
  Note that $e_1$ has neither an upper nor a lower face since its only event $a$
  is in both interfaces.
  For the opposite reason, the edge $e_3$ can be neither start nor accept cell.
  \[
    \begin{tikzpicture}[x=1cm, y=.8cm]
      \node[state] (q) at (0,0) {$x$};
      \draw [->] (q) to[out=30, in=0,looseness=2] (0,1) to[out=180,in=150,looseness=2] (q);
      \node[above] at (0,1) {$e_3$};
      \node[state,color=gray,fill=white,dashed,minimum size=10pt] (a) at (-2,0) {};
      \node[state,color=gray,fill=white,dashed,minimum size=10pt] (z) at (2,0) {};
      \draw[-](a) edge[->] node[below] {$\bot$} node[above]{$e_2$} (q);
      \draw[-](q) edge[->] node[above] {$\top$} node[below]{$e_4$} (z);
      \node[state,color=gray,fill=white,dashed,minimum size=10pt] (aa) at (-5,0) {};
      \node[state,color=gray,fill=white,dashed,minimum size=10pt] (zz) at (-3,0) {};
      \draw[-](aa) edge[->] node[below] {$\bot$} node[above]{$\top\rlap{$e_1$}$} (zz);
    \end{tikzpicture}
  \]
\end{example}

\begin{example}
  \label{ex:iTrON}
  Figure \ref{fig:iTrON} shows an example of a two-dimensional iHDA.
  The initial cell has event iconclist $\ilo{a}{a}{\emptyset}$ and hence no lower face.
  This lack of lower face propagates to the left two-dimensional cell,
  with event iconclist $\loset{\!\ibullet\!\!& a \\ &b}$.
  Hence iHDAs are \emph{partial} HDAs in the sense of \cite{%
    DBLP:conf/calco/FahrenbergL15,
    DBLP:conf/fossacs/Dubut19},
  but the notion of partiality is more restricted here,
  given that it is on the level of events.
\end{example}

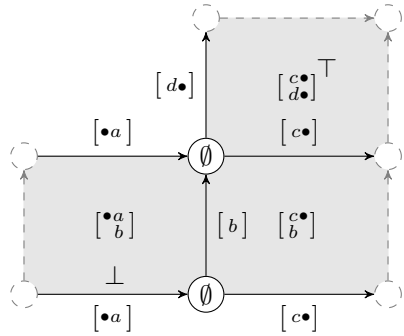
\begin{figure}[h]
  \centering
  \begin{tikzpicture}[x=1.2cm, y=.91cm]
    \filldraw[color=black!10!white] (0,0)--(4,0)--(4,4)--(2,4)--(2,2)--(0,2)--(0,0);
    \node[state,color=gray,fill=white,dashed,minimum size=10pt] (eps) at (0,0) {};
    \node[state] (a) at (2,0) {$\emptyset$};
    \node[state,color=gray,fill=white,dashed,minimum size=10pt] (c) at (0,2) {};
    \node[state] (ac) at (2,2) {$\emptyset$};
    \node[state,color=gray,fill=white,dashed,minimum size=10pt] (aa) at (4,0) {};
    \node[state,color=gray,fill=white,dashed,minimum size=10pt] (aac) at (4,2) {};
    \node[state,color=gray,fill=white,dashed,minimum size=10pt] (acc) at (2,4) {};
    \node[state,color=gray,fill=white,dashed,minimum size=10pt] (aacc) at (4,4) {};
    \path (eps) edge node[swap] {$\loset{\!\ibullet \!\! & a}$} (a);
    \node[above] at (1,0) {$\bot$};
    \node[above right] at (3.1,3) {$\top$};
    \path (eps) edge[color=gray,dashed]  (c);
    \path (a) edge node[swap] {$\loset{c & \!\!\ibullet\! }$} (aa);
    \path (c) edge node {$\loset{\!\ibullet \!\! & a}$} (ac);
    \path (a) edge node[swap] {$\loset{b}$} (ac);
    \path (ac) edge node {$\loset{c & \!\!\ibullet\! }$} (aac);
    \path (acc)  edge[color=gray,dashed] (aacc);
    \path (aa) edge[color=gray,dashed] (aac);
    \path (aac) edge[color=gray,dashed] (aacc);
    \path (ac) edge node {$\loset{d & \!\!\ibullet\! }$} (acc);
    \node at (1,1) {${\loset{ \!\ibullet\!\!& a \\ &b}}$};
    \node at (3,1) {${\loset{ c & \!\!\ibullet\! \\ b& }}$};
    \node at (3,3) {${\loset{ c & \!\!\ibullet\! \\ d&\!\!\ibullet\! }}$};
  \end{tikzpicture}
  \caption{An example of an iHDA. Cells are marked with their event iconclists.}
  \label{fig:iTrON}\vspace*{-2mm}
\end{figure}

\subsection{Paths and languages}

Paths on iHDAs are defined as for HDAs.
Namely, a path is a sequence
$\alpha=(x_0,\varphi_1,x_1,\dotsc,x_n)$
such that each $(x_{i-1},\varphi_i,x_i)$
is either
\begin{itemize}
\item
  an \emph{upstep} $(\delta^0_A(x_i),\arrO{A},x_i)$ for $x_i\in X[\ilo SUT]$, $A\subseteq U-S$, or
\item
  a \emph{downstep} $(x_{i-1},\arrI{B},\delta^1_B(x_{i-1}))$ for $x_{i-1}\in X[\ilo SUT]$, $B\subseteq U-T$.
\end{itemize}
A path $\alpha$ is accepting if $\src(\alpha)=x_0$ is a start cell and $\tgt(\alpha)=x_n$ is an accept cell.

\medskip
For a cell $x\in X$ of an iHDA $X$ we denote by $\ev(x)$ the underlying conclist of $\iev(x)$;
note that
\[
  \iev(x)=(S_{\iev(x)},\ev(x),T_{\iev(x)}).
\]
The \emph{event ipomset} of a path $\alpha$ is defined inductively
as before:
$\ev((x))=\id_{\ev(x)}$,
$\ev(y\arrO{A} x)=\starter{\ev(x)}{A}$,
$\ev(x\arrI{B} y)=\terminator{\ev(x)}{B}$,
and $\ev(\alpha*\beta)=\ev(\alpha)*\ev(\beta)$.
The \emph{language} of an iHDA $X$ is
\begin{equation*}
  \Lang(X) = \{ \ev(\alpha)\mid \alpha \text{ accepting path in } X\}.
\end{equation*}

\begin{example}
  The language of the iHDA from Example \ref{ex:aiHDA} is $\{\ibullet a\ibullet\}\cup \{\ibullet a a^n a\ibullet\mid n\ge 0\}$.
  The language of the iHDA from Example \ref{ex:iTrON} is
  \begin{equation*}
    \left\{ \left[\vcenter{\hbox{\!%
            \begin{tikzpicture}[x=1cm, y=.7cm]
              \node (a) at (0,0) {$a$};
              \node (c) at (1,0) {$c$};
              \node at (1.2,0) {$\ibullet$};
              \node (b) at (0,-.75) {$b$};
              \node at (-.2,0) {$\ibullet$};
              \node (d) at (1,-.75) {$d$};
              \node at (1.2,-.79) {$\ibullet$};
              \path (a) edge (c);
              \path (b) edge (d);
              \path (a) edge (d);
            \end{tikzpicture}
          \!\!}}\right]\right\}\down.
  \end{equation*}
\end{example}

Because of the requirement that events in start cells may not be unstarted
and those in accept cells may not be terminated,
an event in an iHDA carries information whether it will be eventually terminated,
and whether it has been present from the beginning.
This is expressed by the following lemma which shows that
iconclists of cells may be recovered from ipomsets of accepting paths:

\begin{lemma}
  \label{l:EventsofSrcTgt}
  Let $X$ be an iHDA, $\alpha$ be a path in $X$ and $P=\ev(\alpha)$.
  \begin{enumerate}
  \item
    If $\src(\alpha)\in\bot_X$, then $\iev(\tgt(\alpha))\cong(S_P\cap T_P,T_P,Z)$
    for some $Z\subseteq T_P$.
  \item
    If $\tgt(\alpha)\in\top_X$, then $\iev(\src(\alpha))\cong(Y,S_P,S_P\cap T_P)$
    for some $Y\subseteq S_P$.
  \item
    If $\alpha$ is accepting, then
    $\iev(\src(\alpha)\cong (S_P,S_P,S_P\cap T_P)$
    and $\iev(\tgt(\alpha)\cong (S_P\cap T_P,T_P,T_P)$.
  \end{enumerate}
\end{lemma}

\begin{proof}
  It is sufficient to prove 1., 2.\@ is then obtained by reversal and 3.\@ follows from 1.\@ and 2.
  Induction with respect to the length of $\alpha$.
  If $\alpha=(x)$, then $P=\id_{U}=(U,U,U)$ and $\iev(x)=(U,U,Z)$.

\medskip
  If $\alpha=\beta*(\delta^0_A(x)\arrO{A}x)$,	
  $\iev(x)=(S,U,T)$,
  and $A\subseteq U-S$,
  then $\ev(\beta)=P-A$ and
  \[
    \iev(\delta^0_A(x))
    =(S,U-A,T-A)
    =(S_{P-A}\cap T_{P-A},T_{P-A},Z)
    =(S_P\cap T_P, T_P-A,Z)
  \]
  by the inductive hypothesis for $\beta$. Thus, $(S,U,T)=(S_P\cap T_P,T_P,Z)$.

\medskip
  Finally, let $\alpha=\beta*(y\arrI{B}\delta^1_B(y))$,	
  $\iev(y)=(S,U,T)$,
  and $B\subseteq U-T$.
  Denote $\ev(\beta)=Q$, then
  we have $P=Q*\terminator {T_Q}{B}$,
  $S_P=S_Q$ and $T_P=T_Q-B$.
  Therefore,
  \begin{multline*}
    \iev(\tgt(\alpha))
    =
    (S-B,U-B,T)
    =
    (S,U,T)-B
    =
    \iev(\tgt(\beta))-B
    \overset{\text{ind.}}=
    \\
    (S_Q\cap T_Q,T_Q,Z)-B
    =
    (S_Q\cap(T_Q-B),T_Q-B,Z-B)
    =
    (S_P\cap T_P, T_P,Z-B).
  \end{multline*}
  The proof is complete.
\end{proof}

\subsection{HDAs vs.\ iHDAs}

HDAs and iHDAs are related via a pair of adjoint functors:
resolution which maps an HDA $X$ to an iHDA $\Res(X)$
by adjoining all possible assignments of interfaces,
and its left adjoint closure,
which maps an iHDA $X$ to an HDA $\Cl(X)$ by filling in missing faces.
These are introduced in \cite{Kleenearxiv} and have the important property that they preserve languages.
We define them below and develop some lemmas.

The \emph{resolution} of an HDA $X$
is the iHDA $\Res(X)$ defined as follows.
For $\ilo{S}{U}{T}\in\isq$, $A\subseteq U-S$ and $B\subseteq U-T$ we put
\begin{gather*}
  \Res(X)[\ilo{S}{U}{T}]
  =
  \{(x;S,T)\mid x\in X[U]\},
  \\
  \delta^0_{A}((x;S,T))
  =
  (\delta^0_{A}(x); S, T-A),
  \qquad
  \delta^1_{B}((x;S,T))
  =
  (\delta^1_{B}(x); S-B, T).
\end{gather*}
A cell
$(x;S,T)\in\Res(X)[\ilo{S}{U}{T}]$ is a start cell if $x\in X_\bot$
and $S=U$, and an accept cell if $x\in X^\top$ and $T=U$.
Every cell $x\in X[U]$ thus produces
$4^{|U|}$ cells in $\Res(X)$, hence if $X$ is finite, then so is $\Res(X)$.

\begin{example}
  For the precubical set $X$ with $X[a]=\{x\}$ and $X[\emptyset]=\{v, w\}$ we have
  \begin{equation*}
    \Res\left(
      \vcenter
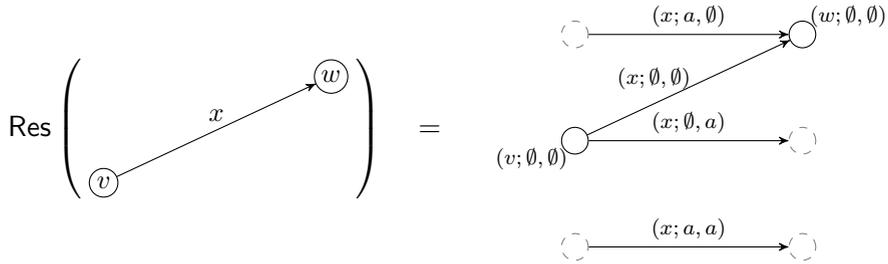
{\hbox{
          \begin{tikzpicture}[x=3cm, y=1.4cm]
            \node[state] (0) at (0,0) {$v$};
            \node[state] (1) at (1,1) {$w$};
            \path (0) edge node[above=-0.1mm] {$x$} (1);
          \end{tikzpicture}
        }}
    \right)
    \quad = \quad
    \vcenter{\hbox{
        \begin{tikzpicture}[x=3cm, y=1.4cm]
          \node[state,minimum size=10pt] (0) at (0,0) {};
          \node[state,minimum size=10pt] (1) at (1,1) {};
          \node[state,color=gray,fill=white,dashed,minimum size=10pt] (0a) at (0,1) {};
          \node[state,color=gray,fill=white,dashed,minimum size=10pt] (1a) at (1,0) {};
          \node[state,color=gray,fill=white,dashed,minimum size=10pt] (0b) at (0,-1) {};
          \node[state,color=gray,fill=white,dashed,minimum size=10pt] (1b) at (1,-1) {};
          \path (0) edge node[above left=-1.6mm] {\scriptsize{$(x;\emptyset,\emptyset)$}} (1);
          \path (0a) edge node[above=-0.4mm] {\scriptsize{$(x;a,\emptyset)$}}  (1);
          \path (0) edge  node[above=-0.4mm] {\scriptsize{$(x;\emptyset,a)$}}   (1a);
          \path (0b) edge  node[above=-0.4mm] {\scriptsize{$(x;a,a)$}}   (1b);
          \node at (0) [below left=-0.4mm] {\scriptsize{$(v;\emptyset,\emptyset)$}};
          \node at (1) [above right=-0.4mm] {\scriptsize{$(w;\emptyset,\emptyset)$}};
        \end{tikzpicture}	
      }}
  \end{equation*}
\end{example}

If $((x_0;S_0,T_0),\phi_1,(x_1;S_1,T_1),\phi_2,\dotsc,\phi_n,(x_n;S_n,T_n))$
is an accepting path in $\Res(X)$, then
$(x_0,\phi_1,x_1,\phi_2\dots,\phi_n,x_n)$ is an accepting path in
$X$ with the same event ipomset.
Conversely, for every accepting path
$\alpha=(x_0,\phi_1,\dots,x_n)$ in $X$
there exists unique subsets $S_k,T_k\subseteq \ev(x_k)$
such that  $((x_0;S_0,T_0),\phi_1,\dotsc,(x_n;S_n,T_n))$
is an accepting path in $\Res(X)$.
(Indeed, $S_0=\ev(x_0)$, $T_n=\ev(x_n)$,
$S_k$ and $\phi_k$ determine $S_{k+1}$,
$T_{k+1}$ and $\phi_k$ determine $T_k$).
As a consequence we obtain:

\begin{lemma}
  \label{l:ResEssential}
  Let $X$ be an HDA.
  If $(x;S,T)\in\Res(X)$ is essential, then $x\in X$ is essential.
\end{lemma}

\begin{lemma}[{\cite[Prop.~11.2]{Kleenearxiv}}]
  \label{l:ResPreservesAcceptingTracks}
  For any HDA $X$, $\Lang(\Res(X))=\Lang(X)$.
\end{lemma}

The \emph{closure} of an iHDA $X$ is the HDA $\Cl(X)$
defined, for all $U\in\sq$, by
\begin{itemize}
\item
	$\Cl(X)=\{
	[x;A,B]\mid x\in X,\;
	A\subseteq S_{\iev(x)},\;
	B\subseteq T_{\iev(x)},\;
	A\cap B=\emptyset
	\}$;
\item
	$\ev([x;A,B])=\ev(x)-(A\cup B)$ for $[x;A,B]\in \Cl(X)$;
\item
	 $\delta^0_{C}([x;A,B]) =
	  [\delta^0_{C-S_{\iev(x)}}(x); A\cup (C\cap S_{\iev(x)}), B]$
	  for  $C\subseteq \ev([x;A,B])$;
\item
	 $\delta^1_{D}([x;A,B]) =
	  [\delta^1_{D-T_{\iev(x)}}(x); A, B\cup (D\cap T_{\iev(x)})]$
	  for $D\subseteq \ev([x;A,B])$;
\item
	$\top_{\Cl(X)}=\{[x;\emptyset,\emptyset]\mid x\in \bot_X\}$,
	$\top_{\Cl(X)}=\{[x;\emptyset,\emptyset]\mid x\in \top_X\}$.
\end{itemize}
Intuitively, closure fills in the missing cells of the iHDA $X$.
Lower face maps $\delta^0_C$ of $\Cl(X)$ take as much of
the face map of $X$ as possible, while the remaining events are added to the set $A$;
similarly for upper faces.

\begin{lemma}
  \label{l:PathsvsiPaths}
  Let $X$ be an iHDA and $x,y\in X$.
  The function $\Phi:\Path(X)_x^y\to \Path(\Cl(X))_{[x;\emptyset,\emptyset]}^{[y;\emptyset,\emptyset]}$,
  \[
    \Phi(x_0,\varphi_1,x_1,\dotsc,x_n)
    =
    ([x_0;\emptyset,\emptyset],\varphi_1,[x_1;\emptyset,\emptyset],\dotsc,[x_n;\emptyset,\emptyset])
  \]
  is a bijection. Moreover, $\ev(\alpha)=\ev(\Phi(\alpha))$ for all $\alpha$.
\end{lemma}

\begin{proof}
	Injectivity of $\Phi$ is clear.
	Let $\alpha =([x_0,A_0,B_0],\phi_1,\dotsc,\phi_n,[x_n,A_n, B_n])\in\Path(\Cl(X))_{[x;\emptyset,\emptyset]}^{[y;\emptyset,\emptyset]}$.
	For every step $([x_k;A_k,B_k],\varphi_k,[x_{k+1};A_{k+1},B_{k+1}])$
	it follows from the definition of face maps that $|A_k|\geq |A_{k+1}|$
	and $|B_k|\leq |B_{k+1}|$.
	Thus, $A_k\subseteq A_0=\emptyset$
	and $B_k\subseteq B_n=\emptyset$ for all $k$
	and then $\alpha=\Phi(x_0,\varphi_1,x_1,\dotsc,x_n)$.
	The second claim is obvious.
\end{proof}

\begin{lemma}[{\cite[Prop.~11.4]{Kleenearxiv}}]
  For any iHDA $X$, $\Lang(\Cl(X))=\Lang(X)$.
\end{lemma}

The following two lemmas are analogues to Lemmas \ref{l:SubsuPath} and \ref{l:PathDivision} for iHDAs.

\begin{lemma}
  \label{l:SubsuIPath}
  Let $X$ be an iHDA, $x,y\in X$, $\alpha\in\Path(X)_x^y$ and $P\subsu Q=\ev(\alpha)$.
  Then there exists $\beta\in\Path(X)_x^y$
  such that $\ev(\beta)=P$.
\end{lemma}

\begin{proof}
  This follows from Lemma \ref{l:SubsuPath} applied to $\Cl(X)$ and Lemma \ref{l:PathsvsiPaths}.
\end{proof}

\begin{lemma}
  \label{l:iPathDivision}
  Let $X$ be an iHDA, $x,y\in X$ and $\gamma\in\Path(X)_x^y$.
  Assume that $\ev(\gamma)=P*Q$ for ipomsets $P$ and $Q$.
  Then there exist paths $\alpha\in\Path(X)_x$ and $\beta\in\Path(X)^y$
  such that $\ev(\alpha)=P$, $\ev(\beta)=Q$ and $\tgt(\alpha)=\src(\beta)$.
\end{lemma}

\begin{proof}
  We apply Lemma \ref{l:PathDivision}
  to the path $\Phi(\gamma)$
  and obtain that there are paths $\alpha'$ and $\beta'$ in $\Cl(X)$
  such that $\ev(\alpha')=P$, $\ev(\beta')=Q$ and $\tgt(\alpha')=\src(\beta')$.
  By Lemma \ref{l:PathsvsiPaths}, $\alpha=\Phi^{-1}(\alpha')$ and $\beta=\Phi^{-1}(\beta')$
  satisfy the required conditions.
\end{proof}

\subsection{Essential iHDAs}

As for HDAs, we say that
a cell $x\in X$ of an iHDA $X$ is \emph{essential} if it accessible and coaccessible.
Let $\ess(X)\subseteq X$ be the set of essential cells.
We show below that, contrary to the situation for HDAs,
$\ess(X)$ is itself an iHDA.

Let $\dist(x,y)$ be the minimal length of a path from $x$ to $y$.
A cell $y$ is accessible if $\dist(x,y)<\infty$ for some $x\in \bot_X$
and is coaccessible if $\dist(y,z)<\infty$ for some $z\in\top_X$.
The following follows directly.

\begin{lemma}
  \label{l:DistanceEstimations1}
  Let $y\in X[\ilo SUT]$, $A\subseteq U-S$ and $B\subseteq U-T$.
  \begin{itemize}
  \item For any $x\in X$, $\dist(x, y)\le \dist(x, \delta^0_A(y))+1$ and
    $\dist(x, \delta^1_B(y))\le \dist(x, y)+1$.
  \item For any $z\in X$, $\dist(y, z)\le \dist(\delta^1_B(y), z)+1$ and
    $\dist(\delta^0_A(y), z)\le \dist(y, z)+1$.
  \end{itemize}
\end{lemma}

The next lemma only holds because of the special properties of start and accept cells in iHDAs.

\begin{lemma}
  \label{l:DistanceEstimations2}
  Let $y\in X[\ilo SUT]$, $A\subseteq U-S$ and $B\subseteq U-T$.
  \begin{itemize}
  \item For every $x\in \bot_X$, $\dist(x, \delta^0_A(y))\le \dist(x, y)$.
  \item For every $z\in \top_X$, $\dist(y, z)\ge \dist(\delta^1_B(y), z)$.
  \end{itemize}
\end{lemma}

\begin{proof}
  We only show the first inequality; the second is symmetric.
  We fix $x$ and proceed by
  induction on cells $y$ with respect to $\dist(x,y)$.
  If $\dist(x,y)=0$, then $y=x$ is a start cell.
  Thus, $S=U$, $A=\emptyset$ and $\delta^0_A(y)=y$.

  \medskip
  Now let $n=\dist(x,y)>0$.
  Without loss of generality we may assume that $A=\{a\}$.
  Let $\alpha$ be a path from $x$ to $y$ of length $n$,
  and let $\alpha=\beta*\gamma$ be a decomposition with $\gamma$ having length $1$.
  Clearly, $\dist(x,\src(\gamma))=n-1$.
  Consider three cases:
  \begin{itemize}
  \item
    $\gamma=(\delta^0_B(y)\arrO{B}y)$ and $a\in B$.
    Then $\beta*(\delta^0_B(y)\arrO{B-a}\delta^0_a(y))$
    has length $n$.
  \item
    $\gamma=(\delta^0_B(y)\arrO{B}y)$ and $a\not\in B$.
    Then $\dist(x,\delta^0_a(\delta^0_B(y)))=
    \dist(x, \delta_B^0(\delta^0_a(y)))\le n-1$
    by induction, and then by Lemma~\ref{l:DistanceEstimations1},
    $\dist(x,\delta^0_a(y))\le \dist(x, \delta_B^0(\delta^0_a(y)))+1\le n$.
  \item
    $\gamma=(z\arrI{B}y)$.
    Then $y=\delta^1_B(z)$, and
    $\dist(x,\delta^0_a(z))\le \dist(x,z)=n-1$ by induction.
    By Lemma~\ref{l:DistanceEstimations1},
    $\dist(x,\delta^0_a(y))
    =
    \dist(x,\delta^1_B(\delta^0_a(z)))
    \le
    \dist(x,\delta^0_a(z))+1\le n$.
  \end{itemize}	

  \vspace*{-6mm}
\end{proof}

\begin{proposition}
  For every iHDA $X$, $\ess(X)\subseteq X$ is an iHDA.
\end{proposition}

\begin{proof}
  Let $y\in X[\ilo SUT]$ be essential.
  We show that all faces of $y$ are also essential.
  There exist $x\in\bot_X$ and $z\in \top_X$
  such that $\dist(x,y),\dist(y,z)<\infty$.
  By Lemmas~\ref{l:DistanceEstimations1} and~\ref{l:DistanceEstimations2},
  $\dist(x,\delta^0_A(y)), \linebreak \dist(x,\delta^1_B(y))<\infty$
  and $\dist(\delta^0_A(y),z), \dist(\delta^1_B(y),z)<\infty$
  for all $A\subseteq U-S$, $B\subseteq U-T$ as well.
  Thus, all faces of $y$ are essential, which concludes the proof.
\end{proof}

\section{Myhill-Nerode construction for iHDAs}
\label{se:imn}

We now develop a Myhill-Nerode construction
which for a given regular language $L$ constructs an iHDA $\iMN(L)$.
Our construction proceeds in several steps.
First we construct a universal iHDA $\iFree$ which recognises all ipomsets,
then we restrict $\iFree$ depending on the given language,
and finally we quotient this $\iFree(L)$ by an equivalence relation
which preserves its language
and ensures that the quotient is finite if $L$ is regular.

\subsection{$\iFree$}

The universal iHDA $\iFree$ is defined as follows:
\begin{gather*}
  \iFree=\{(P,Z)\mid P\in\iiPoms,\; Z\subseteq T_P\}, \qquad \iev(P,Z)=(S_P\cap T_P, T_P, Z), \\
  \delta^0_A(P,Z)=(P-A,Z-A) \text{ for } (P,Z)\in\iFree[\ilo S U T], A\subseteq U-S\cong T_P-S_P, \\
  \delta^1_B(P,Z)=(P*\terminator{T_P}{B},Z) \text{ for } (P,Z)\in\iFree[\ilo S U T], B\subseteq U-T\cong T_P-Z, \\
  \bot_\iFree=\{(\id_U,T)\mid T\subseteq U\in\sq\}, \qquad \top_\iFree=\{(P,T_P)\mid P\in\iiPoms\}.
\end{gather*}
It is clear that $\iFree$ is well-defined;
the precubical identities follow easily
from Lemmas \ref{l:TerminatorComp} and \ref{l:Comm}.
We need some lemmas about existence and uniqueness
(up to subsumption) of paths in $\iFree$.

\begin{lemma}
  \label{l:QtoQP}
  Let $P$ and $Q$ be ipomsets such that $T_Q\cong S_P$,
  let $Z\subseteq T_P$, and $Y=S_P\cap Z\cong T_Q\cap Z\subseteq T_Q$.
  There exists a path $\alpha\in\Path(\iFree)_{(Q,Y)}^{(QP,Z)}$
  with $\ev(\alpha)=P$.
\end{lemma}

\begin{proof}
  We use induction on a step decomposition of $P$.
  If $P=\id_{T_Q}$, then $Y=Z$, and $\alpha=((Q,Y))$ satisfies the required conditions.
  If $P=P'*\starter U A$, then $P'=P-A$ and by induction there exists
  $\beta\in \Path(\iFree)_{(P Q,Y)}^{(QP',Z-A)}$
  such that $\ev(\beta)=P'$. Thus,
  $
  \ev(\beta*((QP',Z-A)\arrO{A}(QP,Z))=P
  $.
  Finally, if $P=P'*\terminator{T_{P'}}{B}$, $T_P\cong T_{P'}-B$,
  then $\ev(\beta*((QP',Z)\arrI{B}(QP,Z))=P$.
\end{proof}

\begin{lemma}
  \label{l:PathsInIFree}
  For every $(P,Z)\in\iFree$ and $Y\subseteq S_P$ we have
  \[
    \big\{
    \ev(\alpha)
    \bigmid
    \alpha\in\Path(\iFree)_{(\id_{S_P},Y)}^{(P,Z)}
    \big\}
    =
    \begin{cases}
      \{P\}\down & \text{if $Y=S_P\cap Z$,}\\
      \emptyset & \text{otherwise.}
    \end{cases}
  \]
\end{lemma}
\begin{proof}
  ($\subseteq$).
  It is enough to show that for every path $\alpha\in\Path(\iFree)_{(\id_V,Y)}^{(P,Z)}$ we have
  $V\cong S_P$, $\ev(\alpha)\subsu P$ and $Y=S_P\cap Z$.
  The first statement is clear;
  the rest we prove by induction on the length of $\alpha$.
  If $\alpha$ is constant, then $P=\id_{S_P}$, $Y=Z$,
  and thus $\ev(\alpha)=\id_{S_P}=P$.
  If $\alpha=\beta*((P-A,Z-A)\arrO{A}(P,Z))$ is a concatenation with an upstep,
  then
  \[
    \ev(\alpha)
    =
    \ev(\beta)*\starter {T_P}A
    \overset{\text{ind.}}\subsu
    (P-A)*\starter {T_P}A
    \overset{\text{L.\@ \ref{l:MinExt}}} \subsu
    P
  \]
  and $Y\overset{\text{ind.}}=S_{(P-A)}\cap (Z-A)=S_P\cap Z$, since $A\cap S_P=\emptyset$.
  If $\alpha=\beta*((Q,Z)\arrI{B}(P,Z))$
  and $P\cong Q*\terminator{T_Q}{B}$, then
  \[
    \ev(\alpha)
    =
    \ev(\beta)*\terminator{T_Q}B
    \overset{\text{ind.}}\subsu
    Q*\terminator{T_Q}B	
    =
    P
  \]
  and $Y=S_Q\cap Z=S_P\cap Z$.
	
 \medskip \noindent
  ($\supseteq$).
 This follows from Lemma \ref{l:QtoQP} for $Q=\id_{S_P}$ and Lemma \ref{l:SubsuIPath}.
\end{proof}

\begin{corollary}
  $\Lang(\iFree)=\iiPoms$.
\end{corollary}

\begin{proof}
  For every ipomset $P$
  there is a path $\alpha\in\Path(\iFree)_{(\id_{S_P},S_P\cap T_P)}^{(P,T_P)}$
  such that $\ev(\alpha)=P$.
\end{proof}

\subsection{$\iFree(L)$}

Fix a language $L$;
we will restrict $\iFree$ to an iHDA that recognises $L$.
Let $\top_L=\{(P,T_P)\mid P\in L\}$ and define
\begin{equation*}
  \iFree(L) = \ess((\iFree, \bot_\iFree, \top_L)).
\end{equation*}
That is, we restrict accept cells of $\iFree$ to the ones that accept ipomsets in $L$
and then reduce to the essential part.

\begin{lemma}
  \label{l:PathsInIFreeL}
  $\Lang(\iFree(L))=L$.
\end{lemma}

\begin{proof}
  This follows from Lemma \ref{l:PathsInIFree}:
  \begin{equation*}
    \Lang(\iFree(L))
    =
    \bigcup_{P\in L}
    \big\{
    \ev(\alpha)
    \bigmid
    \alpha\in \Path(\iFree)_{(\id_{S_P},S_P\cap T_P)}^{(P,T_P)}
    \big\}
    =
    \bigcup_{P\in L}\{P\}\down
    =
    L.
  \end{equation*}

  \vspace*{-7mm}
\end{proof}

We provide a description of $\iFree(L)$
in terms of quotient languages.
For an ipomset $P$
and $Z\subseteq T_P$
define the partial quotient language by
\begin{align*}
  P\lquo[Z] L
  & =
  \{Q\in \iiPoms
  \mid
  P Q\in L,
  S_Q\cap T_Q=Z\}
  =
  \{Q\in P\lquo L \mid
  S_Q\cap T_Q=Z\}.
\end{align*}
In other words, $P\lquo[Z] L$ consists of all ``continuations'' of $P$
that do not terminate events of $Z$
(and terminate all other target events of $P$).
Obviously, $P\lquo L = \bigsqcup_{Z\subseteq T_P} P\lquo[Z]  L$.

\begin{example}
  \label{ex:NonResLang}
  Let $L=\{\loset{a& \\b& \!\!\ibullet\!}\}{\downarrow}\cup \{a b\}=\{\loset{a& \\b& \!\!\ibullet\!}, a b\ibullet, a b\}$.
  Then $P\lquo[\emptyset]L=P\lquo L$ whenever $T_P=\emptyset$,
  and
  \begin{alignat*}{2}	
    a\ibullet \lquo[\emptyset]L
    &=\{\loset{\!\ibullet\!\! & a& \\&b& \!\!\ibullet\!}\}{\downarrow}\cup \{\ibullet ab\},
    &
    \qquad
    a\ibullet \lquo[a]L
    &=\emptyset,	
    \\
    b\ibullet \lquo[\emptyset]L
    &=\emptyset,
    &
    \qquad
    b\ibullet \lquo[b]L
    &=\{\loset{& a& \\ \ibullet\!\! &b& \!\!\ibullet}\},
    \\
    ab\ibullet \lquo[\emptyset]L
    &=\{\ibullet b\},
    &
    \qquad
    ab\ibullet \lquo[b]L
    &=\{\ibullet b \ibullet\},
  \end{alignat*}
  $\loset{a&\!\!\ibullet \\ b&\!\!\ibullet}\lquo[b] L
  =\{\loset{\ibullet \!\! &a& \\\ibullet \!\!& b&\!\!\ibullet}\}$,
  and
  $\loset{a&\!\!\ibullet \\ b&\!\!\ibullet}\lquo[\emptyset] L
  =
  \loset{a&\!\!\ibullet \\ b&\!\!\ibullet}\lquo[a] L
  =
  \loset{a&\!\!\ibullet \\ b&\!\!\ibullet}\lquo[a, b] L
  =\emptyset$.	
\end{example}

\begin{lemma}
  A cell $(P,Z)\in\iFree$ belongs to $\iFree(L)$ if and only if $P\lquo[Z] L\ne \emptyset$.
\end{lemma}

\begin{proof}
  By construction, $(P,Z)\in\iFree$ is accessible.
  We show that $(P, Z)$ is coaccessible if and only if $P\lquo[Z] L\ne \emptyset$.
  If $P\lquo[Z] L\ne \emptyset$, then there is $Q$ such that $P Q\in L$ and $S_Q\cap T_Q=Z$.
	By Lemma \ref{l:QtoQP}, there exists a path from $(P,Z)$ to $(P Q,T_Q)$,
	showing that $(P, Z)$ is coaccessible.

   \medskip
  If $(P, Z)$ is coaccessible, then there is a path $\beta$ in $\iFree$ with $\src(\beta)=(P, Z)$ and $\tgt(\beta)\in \top_L$.
  Let $Q=\ev(\beta)$.
  We also have a path $\alpha$ from $\bot_\iFree$ to $(P, Z)$.
  The concatenation $\alpha*\beta$ is a path in $\iFree(L)$ with $\ev(\alpha*\beta)=P Q$.
  Hence $P Q\in L$.
  Further,
  $T_Q=\ev(\tgt(\alpha*\beta))=Z$,
  since $\tgt(\alpha*\beta)$ is an accept cell.
  That is, $Q\in P\lquo[Z] L$.
\end{proof}

\begin{example}
  \label{ex:inondet}
  Let $L=\{\loset{a\\b}, a b c\}\down=\{\loset{a\\b}, a b, b a, a b c\}$
  be the language of Example \ref{ex:nondet}.
  We construct $\iFree(L)$.
  First, note that $\iFree(L)[\ilo{S}{U}{T}]=\emptyset$ if $S\ne \emptyset$,
  given that $P\lquo L=\emptyset$ if $S_P\ne \emptyset$.
  Similarly, $\iFree(L)[\ilo{S}{U}{T}]=\emptyset$ if $T\ne \emptyset$,
  as all ipomsets in $L$ have empty terminating interface.

  \medskip
  That is, $\iFree(L)[U]$ is only non-empty for conclists $U$ (without interfaces).
  For these,
  \begin{align*}
    \iFree(L)[\emptyset] &= \{(\epsilon, \emptyset), (a, \emptyset), (b, \emptyset),
    (a b, \emptyset), (b a, \emptyset), (\loset{a\\b}, \emptyset), (a b c, \emptyset)\}, \\
    \iFree(L)[a] &= \{(a\ibullet, \emptyset), (b a\ibullet, \emptyset), (\loset{a&\!\!\ibullet\!\\b}, \emptyset)\}, \\
    \iFree(L)[b] &= \{(b\ibullet, \emptyset), (a b\ibullet, \emptyset), (\loset{a\\b&\!\!\ibullet\!}, \emptyset)\}, \\
    \iFree(L)[c] &= \{(a b c\ibullet, \emptyset)\}, \\
    \iFree(L)[\loset{a\\b}] &= \{(\loset{a&\!\!\ibullet\!\\b&\!\!\ibullet\!}, \emptyset)\}.
  \end{align*}
  (Compare these with the cells of $\MN(L)$ in Figure \ref{fig:ex.nondet}.)
  Geometrically, $\iFree(L)$ looks as in Figure~9;
  note that it is an HDA in the sense of Remark \ref{r:ihdanoni}.
\end{example}

\begin{figure}[h]
\vspace*{-2mm}
  \centering
  \begin{tikzpicture}[x=1.2cm, y=1.7cm]
    \filldraw[color=black!10!white] (0,0)--(2,-1)--(4,0)--(2,1)--(0,0);
    \node[state] (eps) at (0,0) {$\epsilon$};
    \node[below left] at (eps) {$\bot$};
    \node[state] (a) at (2,-1) {$a$};
    \node[state] (b) at (2,1) {$b$};
    \node[state] (ba) at (4,1) {$b a$};
    \node[above right=.9mm] at (ba) {$\top$};
    \node[state] (a-b) at (4,0) {$\loset{a\\b}$};
    \node[above right=1.2mm] at (a-b) {$\top$};
    \node[state] (ab) at (4,-1) {$a b$};
    \node[above right=.9mm] at (ab) {$\top$};
    \node[state] (abc) at (6,-1) {$a b c$};
    \node[above right=1mm] at (abc) {$\top$};
    \path (eps) edge node[swap] {$a\ibullet$} (a);
    \path (eps) edge node {$b\ibullet$} (b);
    \path (a) edge node {$\loset{a\\b&\!\!\ibullet\!}$\!\!\!\!} (a-b);
    \path (b) edge node {$b a\ibullet$} (ba);
    \path (b) edge node[swap] {$\loset{a&\!\!\ibullet\!\\b}$\!\!\!} (a-b);
    \path (a) edge node[swap] {$a b\ibullet$} (ab);
    \path (ab) edge node[swap] {$a b c\ibullet$} (abc);
    \node at (1.6,0) {$\loset{a&\!\!\ibullet\!\\b&\!\!\ibullet\!}$};
  \end{tikzpicture}\vspace*{2mm}
  \\
  Figure 9: {\ensuremath{\iFree(L)}  for  \ensuremath{L=\{\loset{a\\b}, a b c\}\down}, see Example \ref{ex:inondet}}.\label{fig:9}\vspace*{-2mm}
  \end{figure}

\subsection{$\iMN(L)$}

The iHDA $\iFree(L)$ is infinite as soon as $L$ is.
(It contains at least one accept cell for every element of $L$.)
Analogously to the construction in Section \ref{se:MN(L)},
we introduce an equivalence relation depending on $L$.
Now however, the relation is not defined on ipomsets but directly on $\iFree(L)$.
In order for the quotient iHDA to be well-defined,
we will need our equivalence to be a \emph{congruence}
in the sense that faces of equivalent cells are again equivalent.

\medskip
We say that $(P,Z),(Q,Y)\in \iFree(L)[\ilo S U T]$ are
\emph{weakly equivalent}, and write $(P,Z)\Lsim(Q,Y)$,
if $P\lquo[Z] L=Q\lquo[Y] L$.
(This is the analogue of the relation $\Lsim$ of Section \ref{se:MN(L)}.)
This relation is not necessarily a congruence,
for example, if $L=\{\loset{a\\b}, aa\}\down$ (\cf~Example \ref{ex:strongeq}),
then $(aa\ibullet,\emptyset)\Lsim(ba\ibullet,\emptyset)$
but
\begin{equation*}
  \delta^0_a(aa\ibullet,\emptyset)
  =
  (a,\emptyset)
  \not\Lsim
  (b,\emptyset)
  =
  \delta^0_a(ba\ibullet,\emptyset).
\end{equation*}
Thus we introduce the maximal congruence contained in $\Lsim$.
Say that $(P,Z),(Q,Y)\in \iFree(L)[\ilo S U T]$ are
\emph{strongly equivalent},
denoted $(P,Z)\Lapprox(Q,Y)$, if
\begin{itemize}
\item
	$\delta^0_A(P,Z)\Lsim\delta^0_A(Q,Y)$ for all $A\subseteq U-S$, and
\item
	$\delta^1_B(P,Z)\Lsim\delta^1_B(Q,Y)$ for all $B\subseteq U-T$.
\end{itemize}
The first is equivalent to the condition
$(P-A)\lquo[Z-A] L=(Q-A)\lquo[Y-A] L$
for every $A\subseteq U-S$
(\cf~the conditions for $\Lapprox$ in Section \ref{se:MN(L)});
and the latter is always satisfied.
\medskip
It is obvious that every congruence contained in $\Lsim$
must be contained in $\Lapprox$ as well.
Below we show that $\Lapprox$
is indeed a congruence, hence the biggest congruence contained in $\Lsim$, and describe its quotient iHDA.

\begin{lemma}
  Let $(P,Z),(Q,Y)\in\iFree(L)[\ilo SUT]$. If $(P,Z)\Lapprox(Q,Y)$, then
  \begin{enumerate}
  \item
    $\delta^0_A(P,Z)\Lapprox\delta^0_A(Q,Y)$ for $A\subseteq U-S$;
  \item
    $\delta^1_B(P,Z)\Lapprox\delta^1_B(Q,Y)$ for $B\subseteq U-T$;
  \item
    $(P,Z)\in\bot_{\iFree(L)}\implies (P,Z)=(Q,Y)$,
  \item
    $(P,Z)\in\top_{\iFree(L)}\iff(Q,Y)\in\top_{\iFree(L)}$.
  \end{enumerate}
\end{lemma}

\begin{proof}
  1.\@
  We have $\delta^0_A(P,Z)=(P-A,Z-A), \delta^0_A(Q,Y)=(Q-A,Y-A)\in \iFree[(S, U-A, T-A)]$.
  For every $C\subseteq U-(S\cup A)$,
  \begin{align*}
    (P-A)-C\lquo[(Z-A)-C]L
    &=
    P-(A\cup C)\lquo[Z-(A\cup C)]L
    \\
    &\Lapprox
    Q-(A\cup C)\lquo[Y-(A\cup C)]L
    =
    (Q-A)-C\lquo[(Y-A)-C]L.
  \end{align*}	
  2.\@
  We have $\delta^1_B(P,Z)=(P*\terminator U B,Z), \delta^1_B(Q,Y)=(Q*\terminator U B,Y)\in \iFree[(S-B, U-B, T)]$.
  For every $C\subseteq U-(S\cup B)$,
  \begin{align*}
    (P,Z)\Lapprox (Q,Z)
    &\implies
    (P-C)\lquo[Z-C]L=(Q-C)\lquo[Y-C]L
    \\
    &
    \overset{\text{L.\@ \ref{l:ExtQ}}}\implies
    (P-C)*\terminator{(U-C)}{B}\lquo[Z-C]L=(Q-C)*\terminator{(U-C)}{B}\lquo[Y-C]L	
    \\
    &
    \overset{\text{L.\@  \ref{l:TerminatorComp}}}\iff
    (P*\terminator U B-C)\lquo[Z-C]L=(Q*\terminator U B-C)\lquo[Y-C]L.	
  \end{align*}
  3.\@
  If $(P,Z)\in\bot_{\iFree(L)}$, then
  $P=\id_U$ and $Z=U=T$: the only start cell in $\iFree(L)[\ilo SUT]$.

  \noindent	
  4.\@ If $(P,Z)\in\top_{\iFree(L)}$, then $Z=T_P=T$ and $\id_{T}\in P\lquo[Z]L$.
  Since $(P,Z)\Lapprox(Q,Y)$, we have $\id_{T_Q}\cong \id_{T}\in Q\lquo[Z]L$,
  and $Z=Y=T_Q$. Thus, $(Y,Z)\in\top_{\iFree(L)}$.
\end{proof}

We may thus define the iHDA $\iMN(L)$ as the quotient of $\iFree(L)$ by $\Lapprox$:
\begin{gather*}
  \iMN(L)[\ilo S U T] = \iFree(L)[\ilo S U T]/{\Lapprox}, \\
  \bot_{\iMN(L)}=\{\eqcl{(P,Z)}\mid (P,Z)\in \bot_{\iFree(L)}\}, \qquad
  \top_{\iMN(L)}=\{\eqcl{(P,Z)}\mid (P,Z)\in \top_{\iFree(L)}\}.
\end{gather*}

\begin{remark}
  \label{re:iemptyi}
  If all ipomsets in $L$ have empty interfaces,
  then $\iMN(L)[\ilo S U T]=\emptyset$ unless $S=T=\emptyset$
  (\cf~Example \ref{ex:inondet}).
  Further, $\iMN(L)[\ilo \emptyset U \emptyset]=\ess(\MN(L))[U]$,
  so both constructions effectively coincide.
  We will see below that this is not the case
  if $L$ contains ipomsets with non-empty interfaces.
\end{remark}

\subsection{Examples}

Cells of iHDAs $\iMN(L)$ correspond to equivalence
classes of pairs $(P,Z)$ for $Z\subseteq P$.
For greater clarity, in the examples below we label
a cell $\eqcl{(P,Z)}$ only by the ipomset $P$
but mark target events belonging to $Z$ by asterisks instead of bullets:
for example $(a\ibullet, \{a\})$ is written as $a\ast$
and $(a\ibullet, \emptyset)$ as $a\ibullet$.

\begin{example}
  Let $L=\{\ibullet a\ibullet\}\cup \{\ibullet a a^n a\ibullet\mid n\ge 0\}$.
  Then $\iMN(L)$ is the iHDA from Example \ref{ex:aiHDA}, and
  \begin{gather*}
    e_1=\ibullet a\ast=\eqcl{(\ibullet a\ibullet,a)},\qquad
    e_2=\ibullet a\ibullet=\eqcl{(\ibullet a\ibullet,\emptyset)},\qquad
    e_3=\{(a^na\ibullet,\emptyset)\mid n\ge 0\},\\
    e_4=\{(a^na\ibullet,a)\mid n\ge 0\},\qquad
    x=\{(\ibullet aa^n,\emptyset)\mid n\ge 0\}.
  \end{gather*}
\end{example}

\begin{example}
  \label{ex:NonResolution}
  For the language $L=\{\loset{a& \\b& \!\!\ibullet\!}\}{\downarrow}\cup \{a b\}=\{\loset{a& \\b& \!\!\ibullet\!}, a b\ibullet, a b\}$
  of Example \ref{ex:NonResLang},
  $\iMN(L)$ and $\MN(L)$ are as follows:
  \[
\scalebox{1.08}{
     \begin{tikzpicture}[x=1cm, y=.9cm]
      \filldraw[color=black!10!white] (0,0)--(2,0)--(2,2)--(0,2)--(0,0);
      \node[state] (eps) at (0,0) {$\epsilon$};
      \node at (-.3,-.2) {$\bot$};
      \node[state,minimum size=10pt] (a) at (2,0) {$a$};
      \node[state,color=gray,fill=white,dashed,minimum size=10pt] (b) at (0,2) {};
      \node[state,color=gray,fill=white,dashed,minimum size=10pt] (ab) at (2,2) {};
      \node[state] (aa) at (4,0) {$ab$};
      \path (eps) edge node[swap] {$a\ibullet$} (a);
      \path (eps) edge node {$b\ast$} (b);
      \path (a) edge node[swap] {$ab\ibullet$} (aa);
      \path (b) edge[color=gray,dashed]  (ab);
      \path (a) edge node[swap] {${\loset{ a& \\ b&\!\!\ast\!}}\Lapprox ab\ast$} (ab);
      \node (m) at (1,1) {${\loset{ a&\!\!\ibullet\! \\ b&\!\!\ast\!}}$};
      \node at (2.8,1.3) {$\top$};
      \node[above right=0.1] at (aa) {$\top$};
      \node at (0,3) {$\iMN(L):$};
      \begin{scope}[shift={(6.8,0)}]
        \filldraw[color=black!10!white] (0,0)--(2,0)--(2,2)--(0,2)--(0,0);
        \node[state] (eps) at (0,0) {$\epsilon$};
        \node at (-.3,-.2) {$\bot$};
        \node[state,minimum size=10pt] (a) at (2,0) {$a$};
        \node[state,minimum size=10pt] (b) at (0,2) {$y_\emptyset$};
        \node[state,minimum size=10pt] (ab) at (2,2) {$y_{\emptyset}$};
        \node[state] (aa) at (4,0) {$ab$};
        \path (eps) edge node[swap] {$a\ibullet$} (a);
        \path (eps) edge node {$b\ibullet$} (b);
        \path (a) edge node[swap] {$ab\ibullet$} (aa);
        \path (b) edge node {$y_{a\ibullet}$}  (ab);
        \path (a) edge node[swap] {${\loset{ a& \\ b&\!\!\bullet\!}}$} (ab);
        \node (m) at (1,1) {${\loset{ a&\!\!\ibullet\! \\ b&\!\!\bullet\!}}$};
        \node at (2.8,1.3) {$\top$};
        \node at (3.1,0.2) {$\top$};
        \node[above right=0.1] at (aa) {$\top$};
        \node at (0,3) {$\MN(L):$};
      \end{scope}
    \end{tikzpicture} }\vspace*{-2mm}
  \]
  Note that
  $\loset{ a& \\ b&\!\!\ast\!}
  =(\loset{ a& \\ b&\!\!\ibullet\!}, \{b\})
  \Lapprox
  (ab\ibullet, \{b\})
  =ab\ast$,
  but $\loset{ a& \\ b&\!\!\ibullet\!}\not\Lapprox ab\ibullet$:
  the relation $\Lapprox$ on cells of $\iFree(L)$ is finer than
  the strong equivalence $\Lapprox$ on ipomsets used in the construction of $\MN(L)$ in Section \ref{se:MN(L)}.
\end{example}

\begin{example}
  Let $L=\{ab, \loset{a&\!\!\ibullet\!\\ b&\!\!\ibullet\!}\}$.
  Then $\iMN(L)$ is the same as $\iFree(L)$ and looks as follows:
  \[
  \scalebox{1.08}{
     \begin{tikzpicture}[x=1cm, y=.84cm]
      \filldraw[color=black!10!white] (0,0)--(2,0)--(2,2)--(0,2)--(0,0);
      \node[state] (eps) at (0,0) {$\epsilon$};
      \node at (-.3,-.2) {$\bot$};
	  \node[state,color=gray,fill=white,dashed,minimum size=10pt] (a) at (2,0) {};
	  \node[state,color=gray,fill=white,dashed,minimum size=10pt] (b) at (0,2) {};
	  \node[state,color=gray,fill=white,dashed,minimum size=10pt] (ab) at (2,2) {};
      \node[state] (c) at (1,-1.5) {$a$};
      \node[state] (d) at (3,-1.5) {$ab$};
      \node[above right=0.1] at (d) {$\top$};
      \path (eps) edge node[swap] {$a\ast$} (a);
      \path (eps) edge node {$b\ast$} (b);
      \path (b) edge[color=gray,dashed]  (ab);
      \path (a) edge[color=gray,dashed] (ab);
      \node (m) at (1,1) {${\loset{ a&\!\!\ast \\ b&\!\!\ast}}$};
      \node at (1.4,1.2) {$\top$};
      \path (eps) edge node[swap] {$a\ibullet$} (c);
      \path (c) edge  node[swap] {$ab\ibullet$} (d);
    \end{tikzpicture} } \vspace*{-3mm}
  \] 	
\end{example}

\begin{figure}[tbp]
  \centering
 \scalebox{1.08}{
   \begin{tikzpicture}[x=1.5cm, y=1.1cm]
    \filldraw[color=black!10!white] (0,0)--(4,0)--(4,2)--(6,2)--(6,4)--(2,4)--(2,2)--(0,2)--(0,0);
    \node[state,color=gray,fill=white,dashed,minimum size=10pt] (eps) at (0,0) {};
    \node[state] (a) at (2,0) {${\ibullet a}$};
    \node[state,color=gray,fill=white,dashed,minimum size=10pt] (b) at (0,2) {};
    \node[state] (ab) at (2,2) {${\loset{ \!\ibullet\!\!& a \\ & b }}$};
    \node[state] (aa) at (4,0) {${\ibullet aa}$};
    \node[state] (aab) at (4,2) {${\ibullet a}$};
    \node[state] (abb) at (2,4) {${\loset{\!\ibullet\!\!& a\\ &b} b}$};
    \node[state] (aabb) at (4,4) {${\loset{ \!\ibullet\!\!& a \\ &b }}$};
    \node[state,color=gray,fill=white,dashed,minimum size=10pt] (aaab) at (6,2) {};
    \node[state,color=gray,fill=white,dashed,minimum size=10pt] (aaabb) at (6,4) {};
    \path (eps) edge node[swap] {${\ibullet a\ibullet}$} (a);
    \node[above] at (1,0) {$\bot$};
    \path (eps) edge[color=gray,dashed]  (b);
    \path (a) edge node[swap] {${\ibullet aa\ibullet}$} (aa);
    \path (b) edge node {${\loset{ \!\ibullet\!\!& a&\!\!\ibullet\! \\ &b }}$} (ab);
    \path (a) edge node[swap] {$e$} (ab);
    \path (2,.4) edge[very thick,color=blue,-] (2,1.2);
    \path (ab) edge node {${\loset{ \!\ibullet\!\!& aa&\!\!\ibullet\! \\ &b }}$} (aab);
    \path (abb) edge node {${\loset{ \!\ibullet\!\!& aa&\!\!\ibullet\! \\ &bb }}$} (aabb);
    \path (aaab) edge[color=gray,dashed]  (aaabb);
    \node[below] at (5,4) {$\top$};
    \path (aa) edge node[swap] {${\loset{ \!\ibullet\!\!& aa&\!\!\ibullet\! \\ &b&\!\!\ibullet\!}}$} (aab);
    \path (aab) edge node[swap] {$e$} (aabb);
    \path (4,2.4) edge[very thick,color=blue,-] (4,3.2);
    \path (ab) edge node {$ {\loset{\!\ibullet\!\!& a \\ &b} b\ibullet}$} (abb);
    \path (aab) edge node {$\ibullet aa\ast$} (aaab);
    \path (aabb) edge node {${\loset{ \!\ibullet\!\!& aa&\!\!\ast \\ &b }}$} (aaabb);
    \node at (1,1) {${\loset{ \!\ibullet\!\!& a&\!\!\ibullet\! \\ &b&\!\!\ibullet\!}}$};
    \node at (3,1) {${\loset{ \!\ibullet\!\!& aa&\!\!\ibullet\! \\ &b&\!\!\ibullet\!}}$};
    \node at (3,3) {${\loset{ \!\ibullet\!\!& aa&\!\!\ibullet\! \\ &bb&\!\!\ibullet\!}}$};
    \node at (5,3) {${\loset{ \!\ibullet\!\!& a&\!\!\ast \\ &b&\!\!\ibullet\!}}$};
    \node[state,color=gray,fill=white,dashed,minimum size=10pt] (aa) at (-1.3,4) {};
    \node[state,color=gray,fill=white,dashed,minimum size=10pt] (zz) at (.7,4) {};
    \draw[-](aa) edge[->] node[below] {$\bot$} node[above]{$\top\rlap{$\ibullet a\ast$}$} (zz);
  \end{tikzpicture}  } \vspace*{2mm}
  \\
  {Figure 10: $\iMN(L)$  for  $L = \{\ibullet a\ibullet\} \cup \{\loset{\!\ibullet\!\!& aa&\!\!\ibullet\! \\ &b}^n\mid n\ge 1\}\down$, %
  see Example \ref{exa:ihda-loop}.}\label{fig:ihda-loop}
\end{figure}

\begin{example}
  \label{exa:ihda-loop}
  Let $L = \{\ibullet a\ibullet\} \cup \{\loset{\!\ibullet\!\!& aa&\!\!\ibullet\! \\ &b}^n\mid n\ge 1\}\down$
  be the language of Example \ref{exa:hda-loop},
  then $\iMN(L)$ is displayed in Figure~10.
  Blue arrows marked $e$ are identified,
  as well as their corresponding endpoints.
  We have $e=\eqcl{\loset{ \!\ibullet\!\!& a& \\ &b&\!\!\ibullet\!}}=\eqcl{ab\ibullet}$
  and $\ev(e)=\ilo \emptyset b \emptyset$.
\end{example}

\subsection{$\iMN(L)$ recognises $L$}

For a cell $x\in\iMN(L)$ denote $x\lquo L:= P\lquo[Z]  L$ for any $(P,Z)\in x$.
This clearly does not depend on the choice of a representative.

\begin{lemma}
  \label{l:LanguagesOfLowerFaces}
  Assume that $x\in\iMN(L)[\ilo SUT]$, $A\subseteq U-S$, $B\subseteq U-T$. Then
  \begin{enumerate}
  \item
    $Q\in x\lquo L\implies \starter U A*Q\in \delta^0_A(x) \lquo L$,
  \item
    $\terminator U B*Q\in x\lquo L\iff Q\in \delta^1_B(x) \lquo L$.
  \end{enumerate}
\end{lemma}

\begin{proof}
  Fix $(P,Z)\in x$.
  Recall that
  $(S_Q, S_Q, S_Q\cap T_Q)\cong \ilo U U Z$
  for every $Q\in x\lquo L$.
  For the first part,
  \begin{align*}
    Q\in x\lquo L
    \iff Q\in P\lquo[Z]  L
    & \iff P Q\in L \\
    &\implies
    (P-A)* \starter U A * Q\in L
    \qquad\text{(Lem.\@ \ref{l:QuotIncl} \& \ref{l:MinExt})}
    \\
    &\iff \starter U A * Q\in (P-A)\lquo[Z-A] L \\
    &\iff \starter U A * Q\in \delta^0_A(x)\lquo L.
  \end{align*}
\indent  For the second part of the lemma,
  \begin{align*}
    Q\in \delta^1_B(x)\lquo L
    \iff
    Q\in (P*\terminator U B)\lquo[Z]  L
    &\iff
    P*\terminator U B *Q\in L
    \\
    &\iff
    \terminator U B *Q\in P\lquo[Z]  L
    \iff
    \terminator U B *Q\in x\lquo L.
  \end{align*}

  \vspace*{-7mm}
\end{proof}

\begin{lemma}
  \label{l:SuffixLangInclusion}
  If $\alpha\in \Path(\iMN(L))_\bot$ and $\tgt(\alpha)\in\iMN(L)[\ilo SUT]$,
  then $\tgt(\alpha)\lquo L\subseteq  \ev(\alpha)\lquo[T] L$.
\end{lemma}

\begin{proof}
  Induction on the length of the path $\alpha$.
  If $\alpha=(x)$
  for $x=\eqcl{(\id_U, T)}$,
  then
  \[
    x\lquo L=(\id_U)\lquo[T] L=\ev(\alpha)\lquo[T]L.
  \]
  If $\alpha=\beta*(\delta^0_A(x)\arrO{A}x)$ for
  $x\in\iMN(L)[\ilo SUT]$ and $A\subseteq U-S$, then
  $\ev(\alpha)=\ev(\beta)*\starter U A$ and
  \[
    Q\in x\lquo L
    \overset{\text{Lem.\@ \ref{l:LanguagesOfLowerFaces}}}\implies
    \starter U A * Q\in \delta^0_A(x)\lquo L
    \overset{\text{ind.}}\implies
    \starter U A * Q\in \ev(\beta)\lquo[T-A] L
    \iff
    Q\in \ev(\alpha)\lquo[T] L.
  \]
  If $\alpha=\beta*(x\arrI{B}\delta^1_B(x))$, then
  $\ev(\alpha)=\ev(\beta)*\terminator U B$ and
  \[
    Q\in \delta^1_B(x)\lquo L
    \overset{\text{Lem.\@ \ref{l:LanguagesOfLowerFaces}}}\iff
    \terminator U B*Q\in x\lquo L
    \overset{\text{ind.}}\implies \!
    \terminator \! U B*Q\in \ev(\beta)\lquo[T] L
    \iff
    Q\in \ev(\alpha)\lquo[T]L.
  \]

  \vspace*{-7mm}
\end{proof}

\begin{example}
  \label{ex:itwopaths}
  Let $L = \{\loset{a\\b}\}\down \cup \{baa,cda,cdaa\}$,
  then $\iMN(L)$ is as follows:
  \[
   \scalebox{1.08}
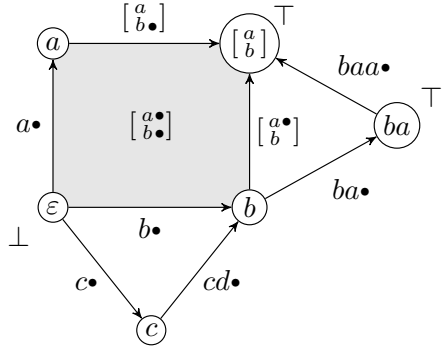
{
    \begin{tikzpicture}[x=1.2cm, y=1cm]
      \filldraw[color=black!10!white] (0,0)--(2,0)--(2,2)--(0,2)--(0,0);
      \node[state] (eps) at (0,0) {${\epsilon}$};
      \node[state] (a) at (0,2) {$a$};
      \node[state] (b) at (2,0) {$b$};
      \node[state] (ab) at (2,2) {$\loset{a\\b}$};
      \node[state] (c) at (1,-1.5) {$c$};
      \node[state] (ba) at (3.5,1) {$ba$};
      \path (eps) edge node {$a\ibullet$} (a);
      \path (eps) edge node[swap] {$b\ibullet$} (b);
      \path (b) edge node[swap] {\!${\loset{a&\!\!\ibullet\! \\ b}}$} (ab);
      \path (a) edge node {${\loset{a& \\ b&\!\!\ibullet\!}}$} (ab);
      \path (eps) edge node[left, pos=.65] {$c\ibullet$\,} (c);
      \path (c) edge node[right, pos=.4] {$cd\ibullet$} (b);
      \path (b) edge node[swap] {$ba\ibullet$} (ba);
      \path (ba) edge node[swap] {$baa\ibullet$} (ab);
      \node at (1,1) {${\loset{a&\!\!\ibullet\! \\ b&\!\!\ibullet\!}}$};
      \node at (2.35,2.35) {$\top$};
      \node at (3.85,1.35) {$\top$};
      \node at (-0.35,-0.35) {$\bot$};
    \end{tikzpicture} }
  \]
  Note that there are two paths recognising $cda$.
  One of them ends at $\loset{a\\b}$,
  yet there is no $P$ such that $cda\subsu P$ and $P\Lapprox \loset{a\\b}$.
  This explains why Lemma \ref{l:SuffixLangInclusion}
  cannot be strengthened.
\end{example}

\begin{lemma}
\label{l:iMNId}
  Let $x\in\iMN(L)[\ilo SUT]$. Then
  $x\in\top_{\iMN(L)}$ if and only if
  $\id_T\in x\lquo L$.
\end{lemma}

\begin{proof}
  We have $x\in\top_{\iMN(L)}$ if and only if
  there exists an ipomset $P\in L$ such that $x=\eqcl{(P,T_P)}$ and $T_P\cong U=T$.
  But $P\in L$ if and only if $\id_T\in P\lquo[T] L=x\lquo L$.
\end{proof}

\begin{proposition}
  $\Lang(\iMN(L))=L$.
\end{proposition}

\begin{proof}
  The quotient map $\iFree(L)\to\iMN(L)$ is an iHDA-map,
  hence it induces an inclusion of languages.
  By Lemma \ref{l:PathsInIFreeL}, $L=\Lang(\iFree(L))\subseteq \iMN(L)$.

\medskip
  For the other direction,
  let $\alpha$ be an accepting path in $\iMN(L)$.
  Since $\tgt(\alpha)\in\iMN(L)[\ilo SUU]$ is an accept cell,
  there exists $Q\in L$ such that $\tgt(\alpha)=\eqcl{(Q,T_Q)}$ with $T_Q\cong U$.
  Thus, by Lemma \ref{l:iMNId}, $\id_U\in Q\lquo[U] L=\tgt(\alpha)\lquo L$.
  By Lemma \ref{l:SuffixLangInclusion},
  $\id_{U}\in \tgt(\alpha)\lquo L\subseteq \ev(\alpha)\lquo[U] L$,
  which implies $\ev(\alpha)\in L$.
  This proves $\Lang(\iMN(L))\subseteq L$.
\end{proof}

\section{Determinism in iHDAs}
\label{se:idet}

The notion of determinism for iHDAs is different from the one for HDAs,
given that we do not have to restrict to essential cells.
Yet we will show that languages recognised by deterministic HDAs and deterministic iHDAs are the same.

\begin{definition}
  An iHDA $X$ is \emph{deterministic} if
  \begin{enumerate}
  \item
    for every $\ilo UUT\in\isq$ there is at most one start cell in $X[\ilo UUT]$, and
  \item
    for every $\ilo SUT\in \isq$, $A\subseteq U-S$ and $x\in X[(S, U-A, T-A)]$,
    there is at most one cell $y\in X[\ilo SUT]$
    such that $x=\delta^0_A(y)$.
  \end{enumerate}
\end{definition}

Compared to deterministic \emph{HDAs}, we now allow one start cell for every pair $U\supseteq T$
of source interface and target interface.
That is because the information of which events may be terminated in an accepting path
is already contained in the target interface of its source cell, \cf~Lemma~\ref{l:EventsofSrcTgt}.

\begin{lemma}
  \label{l:HDAiHDADet}
  If HDA $X$ is deterministic,
  then the iHDA $\ERes{X}$ is also deterministic.
\end{lemma}

\begin{proof}
  The first condition is clear.
  To prove the second,
  fix  $\ilo SUT\in\isq$, $A\subseteq U-S$
  and $(x;S,T-A)\in \ERes{X}[(S, U-A, T-A)]$.
  There is at most one essential $y\in X[U]$ such that $\delta^0_A(y)=x$.

  Let $(z;S,T)\in \ERes{X}[\ilo SUT]$
  such that $\delta^0_A(z;S,T)=(x,S,T-A)$.
  By definition, $\delta^0_A(z)=x$,
  and by Lemma \ref{l:ResEssential} we obtain that $z$ is essential;
  as a consequence, $y=z$.
\end{proof}

From Lemmas \ref{l:ResPreservesAcceptingTracks} and \ref{l:HDAiHDADet} we conclude:

\begin{corollary}
  \label{co:detdet}
  If $L$ is recognised by a deterministic HDA, then it is recognised by a deterministic iHDA.
\end{corollary}

The following two lemmas provide analogues to the unambiguity Lemma \ref{l:DetEquiPaths} for deterministic HDAs.

\begin{lemma}
  \label{le:detpath}
  Let $X$ be a deterministic iHDA and $\alpha,\beta\in\Path(X)_\bot^\top$.
  If $\ev(\alpha)=\ev(\beta)$,
  then $\alpha\simeq\beta$.
\end{lemma}

\begin{proof}
  Denote $P=\ev(\alpha)=\ev(\beta)$.
  Without loss of generality we may assume that $\alpha=\alpha_1*\dotsm*\alpha_n$ and $\beta=\beta_1*\dotsm*\beta_m$ are sparse.
  We show that $n=m$ and $\alpha_k=\beta_k$ for all $k$ which implies the claim.

\medskip
  Both $P=\ev(\alpha_1)*\dotsm*\ev(\alpha_n)$
  and $P=\ev(\beta_1)*\dotsm*\ev(\beta_m)$
  are sparse step decomposition of $P$.
  Proposition~\ref{p:SparsePresentation} implies
  that $m=n$ and $\ev(\alpha_k)=\ev(\beta_k)$ for every $k$.

  Denote $x_0=\src(\alpha)=\src(\alpha_1)$,
  $x_k=\tgt(\alpha_k)=\src(\alpha_{k+1})$,
  $x_n=\tgt(\alpha)=\tgt(\alpha_n)$
  and $y_0=\src(\beta)=\src(\beta_1)$,
  $y_k=\tgt(\beta_k)=\src(\beta_{k+1})$,
  $y_n=\tgt(\beta)=\tgt(\beta_n)$.
  Fix $k$ and denote $\iev(x_k)=(S,U,T)$.
  By Lemma \ref{l:EventsofSrcTgt}.1,
  $(S,U)$ is determined by $Q=\ev(\alpha_1*\dotsm*\alpha_k)$,
  and by Lemma \ref{l:EventsofSrcTgt}.2, $(U,T)$ is determined by $R=\ev(\alpha_{k+1}*\dotsm*\alpha_n)$.
  Similarly, $\iev(y_k)$ is determined by $\ev(\beta_1*\dotsm*\beta_k)$
  and $\ev(\beta_{k+1}*\dotsm*\beta_n)$.
  As a consequence, $\iev(x_k)=\iev(y_k)$ for all $k$.
	
  We show by induction that $x_k=y_k$.
  Since $X$ is deterministic, $\iev(x_0)=\iev(y_0)$
  implies $x_0=y_0$.
  For $k>0$ assume that $x_{k-1}=y_{k-1}$.
  If $\ev(\alpha_k)=\ev(\beta_k)=\starter U A$ is a starter,
  then conditions $\iev(x_k)=\iev(y_k)$ and $\delta^0_A(x_k)=\delta^0_A(y_k)$
  imply $x_k=y_k$ by determinism of $X$.
  If $\ev(\alpha_k)=\ev(\beta_k)=\terminator U B$ is a terminator,
  then $x_k=\delta^1_B(x_{k-1})=\delta^1_B(y_{k-1})=y_k$.
\end{proof}

\begin{lemma}
  \label{le:detp}
  Let $X$ be a deterministic iHDA and $\alpha,\beta\in\Path(X)_\bot$.
  If $\ev(\alpha)=\ev(\beta)$
  and $\iev(\tgt(\alpha))=\iev(\tgt(\beta))$,
  then $\alpha\simeq\beta$.
\end{lemma}

\begin{proof}
  Again, we assume that $\alpha$ and $\beta$ are sparse.
  Denote $x=\tgt(\alpha)$, $y=\tgt(\beta)$,
  $\ilo SUT=\iev(x)=\iev(y)$.
  Modify $X$ by adding accept cells
  \[
    x'=\delta^1_{U-T}(x),y'=\delta^1_{U-T}(y)\in X[\ilo SUU].
  \]
  The paths $\alpha'=\alpha*(x\arrI{U-T} x')$ and $\beta'=\beta*(y\arrI{U-T} y')$
  are accepting,
  and $\ev(\alpha')=\ev(\beta')=\ev(\alpha)*\terminator U{U-T}$.
  Let $\alpha''$ and $\beta''$ be sparse paths
  that are equivalent to $\alpha'$ and $\beta'$, respectively.
  By Lemma \ref{le:detpath}, $\alpha''=\beta''$
  and thus both $\alpha'$ and $\beta'$ are refinements of $\gamma:=\alpha''=\beta''$.
  If $U-T=\emptyset$, then $\alpha=\beta=\gamma$.
  Otherwise,
  decompose $\gamma=\gamma'*\omega$, where $\omega$ is the last step of $\gamma$.
  Then $\omega$ is a downstep $(z\arrI{B} x')$
  such that $U-T\subseteq B$
  and $\alpha=\beta=\gamma'*(z\arrI{B-T}\delta^1_{B-T}(z))$.
\end{proof}

\begin{lemma}
  \label{l:iDetSubPaths}
  Let $X$ be a deterministic iHDA and $\alpha,\beta\in\Path(X)_\bot$.
  If $\ev(\alpha)\subsu\ev(\beta)$
  and $\iev(\tgt(\alpha))=\iev(\tgt(\beta))$,
  then $\tgt(\alpha)=\tgt(\beta)$.
\end{lemma}

\begin{proof}
  By Lemma \ref{l:SubsuIPath} there exists $\alpha'\in\Path(X)_{\src(\beta)}^{\tgt(\beta)}$
  such that $\ev(\alpha')=\ev(\alpha)$.
  From Lemma \ref{le:detp} follows that $\tgt(\alpha)=\tgt(\alpha')=\tgt(\beta)$.
\end{proof}

\begin{proposition}
  \label{pr:DetImpliesSInv}
  If $L$ is recognised by a deterministic iHDA, then $L$ is swap-invariant.
\end{proposition}

\begin{proof}
  Like the proof of Proposition \ref{l:DetImpliesSInv},
  swapping out the applications of Lemma \ref{l:PathDivision} with Lemma \ref{l:iPathDivision}
  and the one of Lemma \ref{l:DetSubPaths} with Lemma \ref{l:iDetSubPaths}.
\end{proof}

Together with Theorem \ref{t:swapdet},
Corollary \ref{co:detdet} and Proposition \ref{pr:DetImpliesSInv}
now imply that a language is recognised by a deterministic iHDA if and only if it is swap-invariant.

\section{Conclusion and further work}

We have proven a Myhill-Nerode type theorem for higher-dimensional automata (HDAs),
stating that a language is regular if and only if it has finite prefix quotient.
We have also introduced deterministic HDAs and shown that not all finite HDAs are determinizable.
Lastly, we have seen that both notions are somewhat simpler
when using higher-dimensional automata with interfaces (iHDAs),
given that no restrictions to essential parts are necessary.

HDAs are arguably simpler than iHDAs, and also somewhat more standard as a model for concurrent computations.
On the other hand, we have seen in \cite{conf/concur/FahrenbergJSZ22, Kleenearxiv} and now also here
that because of the structural axioms of HDAs,
certain concepts are easier to state and prove for iHDAs than for HDAs.
This same observation has led to the introduction of partial HDAs
in \cite{DBLP:conf/calco/FahrenbergL15, DBLP:conf/fossacs/Dubut19},
of which iHDAs are a more restricted event-based version.
In particular, it appears that the trees of Dubut's \cite{DBLP:conf/fossacs/Dubut19}
are related to some of our iHDA constructions developed here.

Our Myhill-Nerode theorem provides a language-internal criterion for whether a language is regular,
and we have developed a similar one to distinguish deterministic languages.
Another important aspect is the \emph{decidability} of these questions,
together with other standard problems such as membership or language inclusion.
Together with coauthors A.~Amrane and H.~Bazille,
we show in \cite{DBLP:conf/ictac/AmraneBFZ23} that these are decidable.

Given that we have shown that not all regular languages are deterministic,
one might ask for the approximation of deterministic languages
by other, less restrictive notions.
It is shown in \cite{DBLP:conf/ictac/AmraneBFZ23} that
non-deterministic HDAs may exhibit unbounded ambiguity,
but other approaches such as for example history-determinism \cite{DBLP:journals/siglog/BokerL23}
or residuality \cite{DBLP:conf/stacs/DenisLT01}
remain to be explored.
It appears that our Myhill-Nerode HDAs may be residual in some sense,
which would open connections to for example automata learning
\cite{DBLP:journals/iandc/Angluin87, DBLP:conf/ijcai/BolligHKL09, DBLP:conf/fossacs/HeerdtKR021}.

\subsection*{Acknowledgement}

We are indebted to members and associates of the (i)Po(m)set Project\footnote{%
  \url{https://ulifahrenberg.github.io/pomsetproject/}}
for numerous discussions regarding the subjects of this paper;
any errors, however, are exclusively ours.

\newcommand{\Afirst}[1]{#1} \newcommand{\afirst}[1]{#1}

\end{document}